%% file: ms.tex
\documentclass[preprint]{elsarticle}
\usepackage{amsthm,amssymb}
\usepackage{bm}

\usepackage{graphicx}
\usepackage{amsmath,latexsym,amssymb,tikz,xspace,graphicx,stfloats,graphics,latexsym,amssymb}
\usetikzlibrary{calc}
\usetikzlibrary{decorations.pathreplacing}
\usetikzlibrary{decorations.markings}
\usepackage{colortbl}
\usepackage[]{color}
\usepackage{enumerate}
\usepackage{subcaption}
\allowdisplaybreaks
\usepackage{multicol}

\newif\ifdraftpaper
\draftpaperfalse

	\definecolor{GreenYellow}   {cmyk}{0.15,0,0.69,0}
\definecolor{Yellow}        {cmyk}{0,0,1,0}
\definecolor{myYellow}      {cmyk}{0,0,0.15,0}

\definecolor{Goldenrod}     {cmyk}{0,0.10,0.84,0}
\definecolor{Dandelion}     {cmyk}{0,0.29,0.84,0}
\definecolor{Apricot}       {cmyk}{0,0.32,0.52,0}
\definecolor{Peach}         {cmyk}{0,0.50,0.70,0}
\definecolor{Melon}         {cmyk}{0,0.46,0.50,0}
\definecolor{myMelon}       {cmyk}{0,0.1,0.1,0}

\definecolor{YellowOrange}  {cmyk}{0,0.42,1,0}
\definecolor{Orange}        {cmyk}{0,0.61,0.87,0}
\definecolor{BurntOrange}   {cmyk}{0,0.51,1,0}
\definecolor{Bittersweet}   {cmyk}{0,0.75,1,0.24}
\definecolor{RedOrange}     {cmyk}{0,0.77,0.87,0}
\definecolor{Mahogany}      {cmyk}{0,0.85,0.87,0.35}
\definecolor{Maroon}        {cmyk}{0,0.87,0.68,0.32}
\definecolor{BrickRed}      {cmyk}{0,0.89,0.94,0.28}
\definecolor{Red}           {cmyk}{0,1,1,0}
\definecolor{OrangeRed}     {cmyk}{0,1,0.50,0}
\definecolor{RubineRed}     {cmyk}{0,1,0.13,0}
\definecolor{WildStrawberry}{cmyk}{0,0.96,0.39,0}
\definecolor{Salmon}        {cmyk}{0,0.53,0.38,0}
\definecolor{CarnationPink} {cmyk}{0,0.63,0,0}
\definecolor{Magenta}       {cmyk}{0,1,0,0}
\definecolor{VioletRed}     {cmyk}{0,0.81,0,0}
\definecolor{Rhodamine}     {cmyk}{0,0.82,0,0}
\definecolor{Mulberry}      {cmyk}{0.34,0.90,0,0.02}
\definecolor{RedViolet}     {cmyk}{0.07,0.90,0,0.34}
\definecolor{Fuchsia}       {cmyk}{0.47,0.91,0,0.08}
\definecolor{myFuchsia}     {cmyk}{0.17,0.18,0,0.04}

\definecolor{Lavender}      {cmyk}{0,0.48,0,0}
\definecolor{myLavender}      {cmyk}{0,0.1,0,0}

\definecolor{Thistle}       {cmyk}{0.12,0.59,0,0}
\definecolor{Orchid}        {cmyk}{0.32,0.64,0,0}
\definecolor{DarkOrchid}    {cmyk}{0.40,0.80,0.20,0}
\definecolor{Purple}        {cmyk}{0.45,0.86,0,0}
\definecolor{Plum}          {cmyk}{0.50,1,0,0}
\definecolor{Violet}        {cmyk}{0.79,0.88,0,0}
\definecolor{RoyalPurple}   {cmyk}{0.75,0.90,0,0}
\definecolor{BlueViolet}    {cmyk}{0.86,0.91,0,0.04}
\definecolor{Periwinkle}    {cmyk}{0.57,0.55,0,0}
\definecolor{CadetBlue}     {cmyk}{0.62,0.57,0.23,0}
\definecolor{CornflowerBlue}{cmyk}{0.65,0.13,0,0}
\definecolor{MidnightBlue}  {cmyk}{0.98,0.13,0,0.43}
\definecolor{NavyBlue}      {cmyk}{0.94,0.54,0,0}
\definecolor{RoyalBlue}     {cmyk}{1,0.50,0,0}
\definecolor{Blue}          {cmyk}{1,1,0,0}
\definecolor{Cerulean}      {cmyk}{0.94,0.11,0,0}
\definecolor{Cyan}          {cmyk}{1,0,0,0}
\definecolor{ProcessBlue}   {cmyk}{0.96,0,0,0}
\definecolor{SkyBlue}       {cmyk}{0.62,0,0.12,0}
\definecolor{Turquoise}     {cmyk}{0.85,0,0.20,0}
\definecolor{TealBlue}      {cmyk}{0.86,0,0.34,0.02}
\definecolor{Aquamarine}    {cmyk}{0.82,0,0.30,0}
\definecolor{BlueGreen}     {cmyk}{0.85,0,0.33,0}
\definecolor{Emerald}       {cmyk}{1,0,0.50,0}
\definecolor{JungleGreen}   {cmyk}{0.99,0,0.52,0}
\definecolor{SeaGreen}      {cmyk}{0.69,0,0.50,0}
\definecolor{Green}         {cmyk}{1,0,1,0}
\definecolor{myGreen}         {cmyk}{0.50,0,0.50,0}
\definecolor{ForestGreen}   {cmyk}{0.91,0,0.88,0.12}
\definecolor{PineGreen}     {cmyk}{0.92,0,0.59,0.25}
\definecolor{LimeGreen}     {cmyk}{0.50,0,1,0}
\definecolor{YellowGreen}   {cmyk}{0.44,0,0.74,0}
\definecolor{SpringGreen}   {cmyk}{0.26,0,0.76,0}
\definecolor{OliveGreen}    {cmyk}{0.64,0,0.95,0.40}
\definecolor{RawSienna}     {cmyk}{0,0.72,1,0.45}
\definecolor{Sepia}         {cmyk}{0,0.83,1,0.70}
\definecolor{Brown}         {cmyk}{0,0.81,1,0.60}
\definecolor{Tan}           {cmyk}{0.14,0.42,0.56,0}
\definecolor{Gray}          {cmyk}{0,0,0,0.50}
\definecolor{myGray}          {cmyk}{0,0,0,0.30}
\definecolor{Black}         {cmyk}{0,0,0,1}
\definecolor{White}         {cmyk}{0,0,0,0}

\newcommand{\fA}{\mathfrak{A}}
\newcommand{\fB}{\mathfrak{B}}%
\newcommand{\fG}{\mathfrak{G}}%
\newcommand{\cC}{\mathcal{C}}%

\renewcommand{\phi}{\varphi} 
\renewcommand{\cal}{\mathcal} 

\newcommand{\set}[1]{\ensuremath{\{ #1 \}}}  

\newcommand{\sizeof}[1]{|\!|#1|\!|}

\newcommand{\Sat}{\ensuremath{\textit{Sat}}}
\newcommand{\FinSat}{\ensuremath{\textit{FinSat}}}

\newcommand{\FO}{\mbox{FO}}
\newcommand{\FOt}{\mbox{$\mbox{\rm FO}^2$}}

\newcommand{\GFt}{\mbox{$\mbox{\rm GF}^2$}}

\newcommand{\GF}{\mbox{$\mbox{\rm GF}$}}

\newcommand{\FLt}{{\mathcal{FL}}^2}
\newcommand{\FL}{{\mathcal{ FL}}}
\newcommand{\FLotrans}{{\mathcal{FL}1\mbox{\rm T}}}
\newcommand{\FLotranstMinus}{{\mathcal{FL}^{2}1\mbox{\rm T}^{\/u}}} 

\newcommand{\FLotranstEq}{{\mathcal{FL}^2_=1\mbox{\rm T}}} 
\newcommand{\FLotranstEqMinus}{{\mathcal{FL}^{2}_=1\mbox{\rm T}^{\/u}}} 

\newcommand{\FLtwotrans}{{\mathcal{FL}^22\mbox{\rm T}}} 	
\newcommand{\FLtwotranstEq}{{\mathcal{FL}^2_=2\mbox{\rm T}}} 	
\newcommand{\FLthreetrans}{{\mathcal{FL}^23\mbox{\rm T}}}	

\newcommand{\Fl}{{\mathcal{FL}}} 
\newcommand{\FlEq}{{\mathcal{FL}_=}} 
\newcommand{\FlVars}[1]{{\mathcal{FL}^{#1}}} 
\newcommand{\FlEqVars}[1]{{\mathcal{FL}^{#1}_=}} 
\newcommand{\FlTrans}[1]{{\mathcal{FL}{#1}\mbox{\rm T}}} 
\newcommand{\FlEqTrans}[1]{{\mathcal{FL}_={#1}\mbox{\rm T}}} 
\newcommand{\FlTransVars}[2]{{\mathcal{FL}^{#2}{#1}\mbox{\rm T}}} 
\newcommand{\FlEqTransVars}[2]{{\mathcal{FL}^{#2}_={#1}\mbox{\rm T}}} 

\newcommand{\NPTime}{\textsc{NPTime}}

\newcommand{\PSpace}{\textsc{PSpace}}
\newcommand{\ExpTime}{\textsc{ExpTime}}
\newcommand{\ExpSpace}{\textsc{ExpSpace}}
\newcommand{\NExpTime}{\textsc{NExpTime}}
\newcommand{\TwoExpTime}{2\textsc{-ExpTime}}
\newcommand{\TwoNExpTime}{2\textsc{-NExpTime}}
\newcommand{\ThreeNExpTime}{3\textsc{-NExpTime}}

\newcommand{\Tower}{\textsc{Tower}}

\newcommand{\N}{{\mathbb N}}

\newcommand{\Z}{{\mathbb Z}}

\newcommand{\ftp}{\ensuremath{\mbox{\rm ftp}}}

\newcommand{\ctp}{\ensuremath{\mbox{\rm ctp}}} 
\newcommand{\cstp}{\ensuremath{\mbox{\rm cstp}}} 

\ifdraftpaper
 \newcommand{\nb}[1]{$|$\marginpar{\scriptsize\raggedright\textcolor{red}{#1}}}
 \else
\newcommand{\nb}[1]{}
\fi
 
\newtheorem{theorem}{Theorem}
\newtheorem{example}[theorem]{Example}
\newtheorem{lemma}[theorem]{Lemma}

\newtheorem{corollary}[theorem]{Corollary}
\newtheorem{proposition}[theorem]{Proposition}

\newcommand{\modu}[2]{\ensuremath{\lfloor #2 \rfloor_{#1}}}
\newcommand{\moduSix}[1]{\ensuremath{\lfloor#1\rfloor}}

\begin{document}

\begin{frontmatter}
\bibliographystyle{plain}

\title{The Fluted Fragment with Transitive Relations\tnoteref{t1}}
\tnotetext[t1]{This is a revised and substantially extended version of the MFCS 2019 paper \cite{P-HT19}.}

\author[1]{Ian Pratt-Hartmann}
\ead{ipratt@cs.man.ac.uk}
\address[1]{University of Opole, Poland/University of Manchester, UK}


\author[2]{Lidia Tendera\corref{cor1}%
}
\ead{tendera@uni.opole.pl}

\address[2]{University of Opole, Institute of Computer Science, Oleska 48, 45-052 Opole, Poland}
\cortext[cor1]{Corresponding author}

	\begin{abstract}
		We study the satisfiability problem for the fluted fragment extended with transitive relations. The logic enjoys the finite model property when only one transitive relation is available and the finite model property is lost when additionally either equality or a second transitive relation is allowed. 
		We show that the satisfiability problem for the fluted fragment with one transitive relation and equality remains decidable. 
		On the other hand we show that the satisfiability problem is undecidable already for  the two-variable fragment of the logic in the presence of three transitive relations (or two transitive relations and equality). 
	\end{abstract}
	

\end{frontmatter}


\input{introduction}
\input{preliminaries}
\input{1-DecidableEq}

\input{2-Undecidable-2T-Eq}
\input{2-Undecidable-3T}

\input{conclusions}
\bibliography{purdyTrans}

	\end{document}

%% file: introduction.tex
\section{Introduction}
\label{sec:intro}
The \textit{fluted fragment}, here denoted $\Fl$, is a fragment of first-order logic in which, roughly speaking, the order of quantification of variables coincides with the order in which those variables appear as arguments of predicates. The allusion is presumably architectural: we are invited to think of arguments of predicates as being `lined up' in columns. The following formulas are sentences of $\Fl$
\begin{align}
& \mbox{
	\begin{minipage}{10cm}
	\begin{tabbing}
	No student admires every professor\\
	$\forall x_1 (\mbox{student}(x_1) \rightarrow \neg \forall x_2 (\mbox{prof}(x_2) \rightarrow \mbox{admires}(x_1, x_2)))$
	\end{tabbing}
	\end{minipage}
}
\label{eq:eg1}\\
& \mbox{
	\begin{minipage}{8.5cm}
	\begin{tabbing}
	No lecturer introduces any professor to every student\\
	$\forall x_1 ($\=$\mbox{lecturer}(x_1) \rightarrow$\\
	\qquad $\neg \exists x_2 ($\=$\mbox{prof}(x_2)
	\wedge \forall x_3 (\mbox{student}(x_3) \rightarrow \mbox{intro}(x_1,x_2,x_3))))$,
	\end{tabbing}
	\end{minipage}
}
\label{eq:eg2}
\end{align}
with the `lining up' of variables illustrated in Fig.~\ref{fig:lining}. By contrast, none of the formulas
%
\begin{align*}
& \forall x_1 . r(x_1, x_1)\\
& \forall x_1 \forall x_2 (r(x_1, x_2) \rightarrow r(x_2,x_1))\\
& \forall x_1 \forall x_2 \forall x_3 (r(x_1,x_2) \wedge r(x_2,x_3) \rightarrow r(x_1,x_3)),
\end{align*}
expressing, respectively, the reflexivity, symmetry and transitivity of the relation $r$, is fluted, as the atoms involved cannot be arranged so that their argument sequences `line up' in the fashion of Fig.~\ref{fig:lining}. 
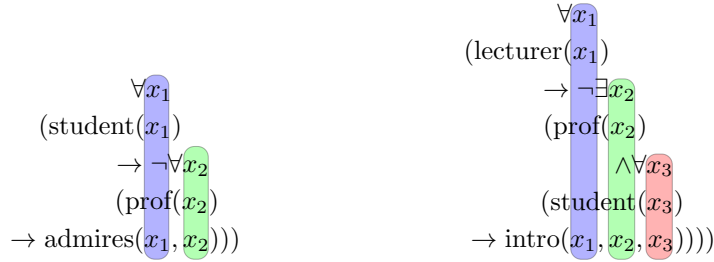
\begin{figure}
	\begin{center}
		\begin{tikzpicture}[scale=0.5]
		\filldraw[rounded corners,opacity=0.3,fill=blue] (3.1,7.4) rectangle (3.75,2.5);
		\filldraw[rounded corners,opacity=0.3,fill=green] (4.15,5.5) rectangle (4.8,2.5);
		\draw (3.3,7) node {${\forall} x_1$};
		\draw (2.2,6) node {$(\mbox{student}(x_1)$};
		\draw (3.6,5) node {$\rightarrow \neg {\forall} x_2$};
		\draw (3.7,4) node {$(\mbox{prof}(x_2)$};
		\draw (2.6,3) node {$ \rightarrow \mbox{admires}(x_1, x_2)))$};
		\end{tikzpicture}
		\hspace{2.5cm}
		\begin{tikzpicture}[scale=0.5]
		\filldraw[rounded corners,opacity=0.3,fill=blue] (3.0,7.3) rectangle (3.7,0.5);		
		\filldraw[rounded corners,opacity=0.3,fill=green] (4,5.3) rectangle (4.7,0.5);	
		\filldraw[rounded corners,opacity=0.3,fill=red] (5.0,3.3) rectangle (5.7,0.5);
		\draw (3.2,7) node {$\forall x_1$};
		\draw (2.1,6) node {$(\mbox{lecturer}(x_1)$};
		\draw (3.5,5) node {$ \rightarrow \neg \exists x_2$};
		\draw (3.65,4) node {$(\mbox{prof}(x_2)$};
		\draw (4.95,3) node {$ \wedge \forall x_3$};
		\draw (4.1,2) node {$(\mbox{student}(x_3)$};
		\draw (3.6,1) node {$\rightarrow \mbox{intro}(x_1,x_2,x_3))))$};
		\end{tikzpicture}
	\end{center}
	\caption{The `lining up' of variables in the fluted formulas~\eqref{eq:eg1} and~\eqref{eq:eg2}; all quantification is executed on the right-most available column.}
	\label{fig:lining}
\end{figure}

The history of this fragment is somewhat tortuous.
The basic idea of {\em fluted logic} can be traced to a paper given by W.V.~Quine to the 1968 {\em International Congress of Philosophy}~\cite{purdyTrans:quine69}, in which the author defined the {\em homogeneous $m$-adic formulas}. Quine later relaxed this fragment, in the context of a discussion of predicate-functor logic, to what he called `fluted' quantificational schemata~\cite{purdyTrans:quine76b}, 
claiming that the satisfiability problem for the relaxed fragment is decidable.
The viability of the proof strategy sketched by Quine was explicitly called into question by Noah \cite{purdyTrans:noah80}, and the subject then taken up by 
W.C.~Purdy~\cite{purdyTrans:purdy96a}, who gave his own definition of `fluted formulas', proving decidability.
It is questionable whether Purdy's reconstruction is faithful to Quine's intentions: the matter is clouded by differences between the definitions of predicate functors in Noah's and Quine's respective papers~\cite{purdyTrans:noah80}
and~\cite{purdyTrans:quine76b}, both of which Purdy cites. In fact, Quine's original definition of 
`fluted' quantificational schemata appears to coincide with a logic introduced---apparently independently---by A.~Herzig~\cite{purdyTrans:herzig90}.
Rightly or wrongly, however, the name `fluted fragment' 
has now attached itself to Purdy's definition in~\cite{purdyTrans:purdy96a}; and we shall continue to use it in that way in the present article. See Sec.~\ref{sec:prelim} for a formal definition. 

To complicate matters further, Purdy claimed in~\cite{purdyTrans:purdy02} that  
$\Fl$ (i.e.~the fluted fragment, in our sense, and his) has the exponential-sized model property: if a fluted formula $\phi$ is satisfiable, then it is satisfiable over a domain of size bounded by an exponential function of the number of symbols in $\phi$. Purdy concluded that the satisfiability problem for $\Fl$ is \NExpTime-complete. These latter claims are false. It was shown in~\cite{purdyTrans:P-HST16} that, although
$\Fl$ has the finite model property, there
is no elementary bound on the sizes of the models required, and the satisfiability problem for $\Fl$ is non-elementary.
More precisely, define $\FlVars{m}$ to be the subfragment of
$\Fl$ in which at most $m$ variables (free or bound) appear. Then the satisfiability problem 
for $\FlVars{m}$ is $\lfloor m/2 \rfloor$-\NExpTime-hard for all $m \geq 2$ and in
$(m-2)$-\NExpTime{} for all $m \geq 3$~\cite{P-HST-FLrev}.
It follows that the satisfiability problem 
for $\Fl$ is \Tower-complete, in the framework of~\cite{purdyTrans:schmitz16}.
These results fix  the exact complexity 
of satisfiability of $\FlVars{m}$ for small values of $m$.
Indeed, the satisfiability problem for $\FOt$, 
the two-variable fragment of first-order logic, is known to be \NExpTime-complete~\cite{purdyTrans:gkv97}, whence the corresponding problem for 
$\FlVars{2}$ is certainly in \NExpTime. Moreover, 
for $0 \leq m \leq 1$, $\FlVars{m}$ coincides with the $m$-variable fragment of first-order logic,
whence its satisfiability problem is \NPTime-complete. Thus,
taking $0$-\NExpTime{} to mean \NPTime, we see that the satisfiability problem 
for $\FlVars{m}$ is $\lfloor m/2 \rfloor$-\NExpTime-complete, at least for $m\leq 4$.

The focus of the present paper is what happens when we add to the fluted fragment the 
ability to stipulate that certain designated binary relations are \textit{transitive}, or are
{\em equivalence relations}. The motivation comes from analogous results obtained for other
decidable fragments of first-order logic.  Consider basic propositional modal logic K.
Under the standard translation into first-order logic (yielded by Kripke semantics), we can regard K
as a fragment of first-order logic---indeed as a fragment of $\FlVars{2}$.
From basic modal logic K, we obtain the logic K4 under the supposition that the accessibility relation
on possible worlds is transitive, and the logic S5 under the supposition that it is an equivalence relation: it is well-known that the satisfiability problems for K and K4 are \PSpace-complete, whereas that for
S5 is \NPTime-complete~\cite{purdyTrans:ladner77}. (For analogous results on {\em graded} modal logic, see~\cite{purdyTrans:kph09}.)
Closely related are also description logics (cf.~\cite{purdyTrans:2003handbook}) with {\em role hierarchies} and {\em transitive roles}. In particular, the description logic  $\mathcal{SH}$, which has the finite model property, is an $\ExpTime$-complete fragment of $\Fl$ with transitivity. 
Similar investigations have been carried out in respect of \FOt{}, which has the finite model property and
whose satisfiability problem, as just mentioned, is \NExpTime-complete. The finite model property is lost when one transitive relation or two equivalence relations are allowed. 
For equivalence, everything is known: the (finite) satisfiability problem for \FOt{} in the presence of a single equivalence relation 
remains \NExpTime-complete, but this increases to \TwoNExpTime-complete in the presence of two equivalence relations~\cite{purdyTrans:kmp-ht,purdyTrans:KO12}, and becomes undecidable with three. For transitivity, we
have an incomplete picture:
the {\em finite} satisfiability problem for $\FOt$ in the presence with a single transitive relation is decidable in \ThreeNExpTime~\cite{purdyTrans:ph18}, while the decidability of the satisfiability problem remains open (cf.~\cite{ST-FO2T}); the corresponding problems with two transitive relations 
are both undecidable~\cite{purdyTrans:KT09}. 

Adding equivalence relations to the fluted fragment poses no new problems. Existing results on
of $\FOt$ with two equivalence relations can be used to show that the satisfiability and finite
satisfiability problems for $\Fl$ (not just $\FlVars{2}$) with \textit{two} equivalence relations are decidable. Furthermore,
the proof that the corresponding problems for
$\FOt$ in the presence of \textit{three} equivalence relations are undecidable can easily be seen to apply also to $\FlVars{2}$. 
On the other hand, the situation with transitivity is less straightforward. 
We 
show in the sequel that the satisfiability and finite satisfiability problems for 
$\Fl$ remain decidable in the presence of a single transitive relation and equality. (This logic lacks the finite model property.) On the other hand, the satisfiability and the finite satisfiability problems for 
$\Fl$ in the presence of two transitive relations and equality, or indeed, 
in the presence of three transitive relations (but without equality) are all undecidable.
For the fluted fragment with two transitive relations but \textit{without} equality, the situation is not fully resolved. 
We show in the sequel that this fragment lacks the finite model property; this contrasts with the situation in description logics, where not only $\mathcal{SH}$ but also its extension $\mathcal{SHI}$ 
retain the finite model property, independently of the number of transitive relations \cite{DucL10}. However, the decidability of both satisfiability and finite satisfiability for this fragment remain open. Table~\ref{tab:overview} gives an overview of these results in comparison with known results on $\FOt$.\nb{I: sentenced moved from end of next para.}

Some indication that flutedness interacts in interesting ways with transitivity is given by known complexity results on various extensions of
guarded two-variable fragment with transitive relations. 
The {\em guarded fragment}, denoted \GF{}, is that fragment
of first-order logic in which all quantification
is of either of the forms 
$\forall \bar{v}(\alpha \rightarrow \psi)$ or $\exists \bar{v}(\alpha \wedge \psi)$,
where $\alpha$ is an atomic formula (a so-called {\em guard}) 
featuring all free variables of $\psi$. The
{\em guarded two-variable fragment}, denoted \GFt, is the intersection of
\GF{} and \FOt. It is straightforward to show that the addition
of two transitive relations to \GFt{} yields a logic whose satisfiability 
problem is undecidable. However, as long as the distinguished 
transitive relations appear only in guards, we can extend the whole
of \GF{} with any number of transitive relations, yielding the so-called 
{\em guarded fragment with transitive guards}, whose satisfiability 
problem is in \TwoExpTime~\cite{purdyTrans:st04}. Intriguingly, in the
two-variable case, we obtain a reduction in complexity if we
require transitive relations in guards to point {\em forward}---i.e.~allowing only $\forall v(t(u,v) \rightarrow \psi)$
rather than $\forall v(t(v,u) \rightarrow \psi)$, and similarly
for existential quantification. These restrictions resemble flutedness, of course,
except that they prescribe the order of variables only in \textit{guards}, rather than in the whole formula.
Thus, the extension of
\GFt{} with (any number of) transitive guards has a \TwoExpTime-complete
satisfiability problem; however, the corresponding problem under 
the  restriction to one-way transitive guards is \ExpSpace-complete~\cite{purdyTrans:kieronski06}. Since the above-mentioned extensions of \GFt{} lack the finite model property, their satisfiability and the finite satisfiability problems do not coincide. 
Decidability and complexity bounds for the finite satisfiability problems are established in \cite{purdyTrans:KT09,purdyTrans:KT18}. 

\begin{center}
	\begin{table}[htb]
		\begin{center}
			{
				{
					\begin{tabular}{|c|c|c|}
						\hline
						{\bf Special symbols} & \multicolumn{2}{c|}{
							{\bf Decidability and Complexity}}\\
						& $\FlVars{m}$ ($m\geq 2$) & $\FOt$ \\
						\hline\hline
no transitive r.  &  $\lfloor m/2\rfloor$-\NExpTime-hard  &FMP \\
& in $(m-2)$-\NExpTime $^{*)}$ &  \NExpTime{}-compl. \\
& \cite{purdyTrans:P-HST16,P-HST-FLrev}	& \cite{purdyTrans:gkv97}\\

\hline
1 transitive r.   & FMP \cite{P-HT19} & \Sat{}: ?  \\
\hline
1 transitive r. & \Sat: in $m$-\NExpTime & \Sat{}: ? \\
with =		&  {\bf Theorem~\ref{theo:FLotransmEqUpper}}	& \FinSat: \\
& \FinSat: in $(m+1)$-\NExpTime &  in \ThreeNExpTime\\
& \bf{Corollary \ref{cor:finsat1T}} &  \cite{purdyTrans:ph18}\\
						\hline
						2 transitive r.  &  \Sat{}: ? & undecidable  \\
						& \FinSat{}: ? & \cite{Kie05,Kaz06}\\
						\hline
						2 transitive r.  & undecidable &  \\
						with =				&  {\bf Theorem~\ref{theo:twoEqUndecidable}} & \color{gray}{undecidable}\\
						\hline
						1 trans.\& 1 equiv. & undecidable &  \\
						with =					& {\bf  Corollary~\ref{cor:twoEqUndecidable}} &  \color{gray}{undecidable}\\
						\hline
						3 transitive r. &  undecidable & \\
						& \Sat: {\bf Theorem~\ref{th:three}} & \color{gray}{undecidable}\\
						& \FinSat{}: {\bf Theorem~\ref{th:threeFinsat}} & \\
						\hline
						3 equivalence r. & undecidable & \\
						& {\bf Corollary~\ref{cor:three}} & \color{gray}{undecidable} \\
						\hline
						
						%
					\end{tabular} 
			}}
			
			\caption{Overview of $\FlVars{m}$ and $\FOt$ over restricted classes of structures. $^{*)}$ in case $m>2$, and $\NExpTime$-complete for $\FlVars{2}$. 
				Undecidability of extensions of \FOt{} shown in grey were known earlier, but now can be inherited from remaining results of the Table.} 
			\label{tab:overview}
		\end{center}
	\end{table}
\end{center}

%% file: preliminaries.tex
\section{Preliminaries}
\label{sec:prelim}
All signatures in this paper are purely relational, i.e.,~ there are no individual constants or function symbols. We
do, however, allow 0-ary relations (proposition letters). We use the notation $\phi \dot{\vee} \psi$ 
to denote the exclusive disjunction of $\phi$ and $\psi$.

Let $\bar{x}_\omega= x_1, x_2, \ldots$ be a fixed sequence of variables.
%
%
%
We define the sets of formulas $\Fl^{[m]}$ (for $m \geq 0$) by structural induction as follows:
(i) any non-equality
atom $\alpha(x_\ell, \ldots, x_m)$, where $x_\ell, \dots, x_m$ is a contiguous (possibly empty) subsequence of $\bar{x}_\omega$,
is in $\Fl^{[m]}$;
(ii) $\Fl^{[m]}$ is closed under boolean combinations;
(iii) if $\phi$ is in $\Fl^{[m+1]}$, then $\exists x_{m+1} \phi$ and $\forall x_{m+1} \phi$
are in $\Fl^{[m]}$.
The set of \textit{fluted formulas} is defined as \smash{$\FL = \bigcup_{m\geq 0} \Fl^{[m]}$}. A \textit{fluted sentence} is a fluted formula with no free variables.
Thus, when forming Boolean combinations in the fluted fragment, all the combined formulas must have as
their free variables some suffix of some prefix $x_1, \dots, x_m$ of $\bar{x}_\omega$; and, when quantifying, only the last variable in this prefix may be bound. 
Note also that proposition letters (0-ary predicates)
may, according to the above definitions, be combined freely with formulas: if $\phi$ is in 
$\Fl^{[m]}$, then so, for example, is $\phi \wedge P$, where $P$ is a proposition letter. 
For 
$m \geq 0$, denote by $\FlVars{m}$  the \textit{$m$-variable sub-fragment} of $\FL$, i.e.~the set of 
formulas of $\Fl$ featuring at most $m$ variables, free or bound.
Do not confuse $\FlVars{m}$ with $\Fl^{[m]}$.
For example, \eqref{eq:eg1} is in $\FlVars{m}$ just in case
$m \geq 2$, and \eqref{eq:eg2} is in $\FlVars{m}$ just in case $m \geq 3$; but they  
are both in $\Fl^{[0]}$. Note that $\FlVars{m}$-formulas cannot, by force of syntax, feature predicates of arity greater than $m$. 
The fragments $\FlEqVars{[m]}$, $\FlEq$ and $\FlEqVars{m}$ are defined analogously, except that equality atoms $x_{m-1} = x_m$ are
allowed in $\FlEqVars{[m]}$ for $m \geq 2$.

We denote by $\FlTrans{k}$ the extension of $\Fl$ with $k$ distinguished binary predicates 
assumed to be interpreted as transitive relations; and we denote by $\FlEqTrans{k}$ the corresponding
extension of $\FlEq$. We denote
their $m$-variable sub-fragments ($m \geq 2$) by $\FlTransVars{k}{m}$, respectively $\FlEqTransVars{k}{m}$.
A predicate is called {\em ordinary} if it is neither the equality predicate nor one of the distinguished predicates. 
Finally, we denote by $\FlEqTransVars{1}{2}^u$
the sub-fragment of
$\FlEqTransVars{1}{2}$ in which no binary predicates occur except equality and the distinguished predicate, i.e.,
where the non-logical signature consists purely of nullary and unary predicates, together with one distinguished binary predicate.

If ${\cal L}$ is any logic, we denote its satisfiability problem by $\Sat({\cal L})$ and its finite 
satisfiability problem by $\FinSat({\cal L})$, understood in the usual way.

\subsection{Variable-free syntax for fluted formulas}
Assuming,  as we shall, that the arity of every predicate is fixed in advance, variables in fluted formulas carry no
information, and therefore can be omitted. Thus, for example, sentences~\eqref{eq:eg1} and~\eqref{eq:eg2} can be written as follows
\begin{align}
& \mbox{
	\begin{minipage}{10cm}
	\begin{tabbing}
	No student admires every professor\\
	$\forall (\mbox{student}\rightarrow \neg \forall (\mbox{prof} \rightarrow \mbox{admires}))$
	\end{tabbing}
	\end{minipage}
}
\label{eq:eg1:free}\\
& \mbox{
	\begin{minipage}{8.5cm}
	\begin{tabbing}
	No lecturer introduces any professor to every student\\
	$\forall ($\=$\mbox{lecturer} \rightarrow$
	$\neg \exists ($\=$\mbox{prof}
	\wedge \forall  (\mbox{student} \rightarrow \mbox{intro})))$,
	\end{tabbing}
	\end{minipage}
}
\label{eq:eg2:free}
\end{align}
As an exercise, try converting \eqref{eq:eg2:free}
back into~\eqref{eq:eg2}.
The only ambiguity here comes from the choice of the highest-indexed variable; for example, the notation
$\forall (\mbox{prof} \rightarrow \mbox{admires})$ can mean $\forall x_{m+1} (\mbox{prof}(x_{m+1}) \rightarrow \mbox{admires}(x_m,x_{m+1}))$ for any $m \geq 1$. However, such ambiguity is perfectly harmless, and in 
fact---as the present authors have found---rather convenient. 
Variable-free syntax for fluted formulas takes a little getting used to, but makes for a compact presentation; 
we shall standardly employ it in the sequel. We write $\forall^m$ to denote a block of $m$ universal quantifiers; 
thus, if $\phi \in \Fl^{[m]}$, then $\forall^m  \phi \in \Fl^{[0]}$.
The elimination of variables seems to have been part of Quine's original motivation for introducing the fluted fragment
(or at least one of its close relatives).

\subsection{Loss of the finite model property}
The logic $\FLotrans$ possesses the finite model property (see Table~\ref{tab:overview}).  
However, this is no longer true if we add either equality or a second transitive relation, as shown by the examples below.
\begin{example}
	Consider the $\FLotranstEq$-sentence $\phi_1=\forall \exists . T_1 \wedge \forall \forall (T_1 \rightarrow \neg =)$,
	where $T_1$ is a distinguished binary predicate denoting a transitive relation. This sentence is satisfiable, but not finitely satisfiable. 
\end{example} 
\begin{proof}
	In standard first-order syntax, $\phi_1$ reads as follows:
	\begin{equation*}
	\phi_1=\forall x \exists y . T_1(x,y) \wedge \forall x \forall y (T_1(x,y) \rightarrow x\neq y).
	\end{equation*}
	It is obvious that $\phi_1$ is satisfiable 	
	(for example by the structure $\N$ with $T_1$ interpreted as $<$), but not finitely satisfiable.
\end{proof}

\begin{example}	Consider the $\FLtwotrans$-sentence
\begin{multline*}
\phi_2 =	 \exists p_0 \wedge  \forall (p_0 \dot{\vee} p_1 \dot{\vee} p_2) \wedge	\forall\forall \neg(T_1\wedge T_2) \wedge \\
\bigwedge_{i=0,1,2} \forall \big(p_i \rightarrow (\exists (p_{i+1} \wedge \neg (T_1 \vee T_2)) \; \wedge\;
	\forall (p_{i+2}\rightarrow T_1\vee T_2))\big), 
\end{multline*}
where the $p_i$ \textup{(}$0\leq i\leq 2$\textup{)} are unary predicates \textup{(}addition in subscripts interpreted modulo $3$\textup{)}, and
$T_1$, $T_2$ are distinguished binary predicates denoting transitive relations. This sentence is satisfiable, but not finitely satisfiable. 
\label{ex:two} 
\end{example}
\begin{proof}
For readers still getting used to variable-free notation, we again restore the variables in $\phi_2$:
\begin{multline*}
\exists x_1. p_0(x_1) \wedge  \forall x_1 (p_0(x_1) \dot{\vee} p_1(x_1) \dot{\vee} p_2(x_1)) \wedge	\forall x_1 \forall _2 \neg(T_1(x_1, x_2) \wedge T_2(x_1, x_2)) \wedge \\
\bigwedge_{i=0,1,2} \forall x_1 \big(p_i(x_1) \rightarrow (\exists x_2 (p_{i+1}(x_2) \wedge \neg (T_1(x_1, x_2) \vee T_2(x_1, x_2))) \; \wedge\; \\
	\forall x_1 (p_{i+2}(x_1) \rightarrow T_1(x_1, x_2)\vee T_2(x_1, x_2)))\big).
\end{multline*}
One can easily check that the structure $\N$ with the following interpretation of the predicate letters
	\begin{align*}
	p_i(n) \quad & \text{iff} \quad n \mod 3 =i\\
	T_1(n,m) \quad & \text{iff} \quad  n + 1 < m\\
		T_2(n,m) \quad & \text{iff}  \quad n > m
	\end{align*}
	is a model of $\phi_2$. 
	
	%
	
	\begin{figure}[hbt]
		\begin{center}
			
			\resizebox{12.2cm}{!}{
				\begin{tikzpicture}
				
				\clip (-0.8,-0.8) rectangle (8.7,1.3);
				
				\foreach \x in {0,1,2,3,4,5,6,7}
				\coordinate [label=center:$a_\x$] (A) at (\x, -0.5);
				
				\coordinate [label=center:$\ldots$] (A) at (7.8, 0);
				
				\foreach \x in {0,3,6}
				{
					\filldraw[fill=black] (\x, 0) circle (0.1); 
				}
				
				\foreach \x in {1,4,7}
				{
					\filldraw[fill=gray] (\x, 0) circle (0.1); 
				}
				
				\foreach \x in {2,5}
				{
					\filldraw[fill=white] (\x, 0) circle (0.1); 
				}

				\foreach \x in {0,2,4} \foreach \y in {0}
				\foreach \s in {0.15}
				{
					\draw [->, rounded corners,  thick, blue]  (\x+0.05,\y+\s) to [out=30,in=120] (\x+1.95,\y+\s);
					\draw [->, rounded corners,  thick, blue]  (\x+1+\s,\y-\s) to [out=-10,in=-170] (\x+3-\s,\y-\s);
				}
				
				\foreach \x in {0,1,2,3,4,5,6} \foreach \y in {0}
				\foreach \s in {0.1}
				{
					\draw [->, rounded corners, thick, red]  (\x+1-\s,\y+\s) to [out=150,in=20] (\x+\s,\y+\s);
					\draw [-, rounded corners, thick, dotted, black]  (\x+\s,\y) -- (\x+1-\s,\y);
				}
				
				\end{tikzpicture}
			}
		\end{center}
		
		\caption{Infinite chain in models of $\phi_2$ from Example~\ref{ex:two}. Pairs $(a_{i},a_{i+1})$ are neither in $\color{blue}{T_1}$ nor in $\color{red}{T_2}$; depicted by dotted lines. Blue and red arrows depict pairs belonging to the transitive relations $\color{blue}{T_1}$ and $\color{red}{T_2}$.}
		\label{fig:example2}
	\end{figure}
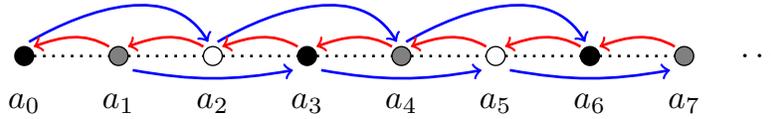
	
	To see that $\phi_2$ is not finitely satisfiable, suppose $\fA\models \Psi$. By the existential conjuncts of $\phi_2$, there exist distinct elements  $a_0,a_1,a_2\in A$ such that $a_i\in p_i$ and $(a_0,a_1), (a_1,a_2)\not\in T_1 \cup T_2$  (cf.~Figure~\ref{fig:example2}). 
	The universal  conjuncts of $\phi_2$ imply that $(a_0,a_2)$, $(a_1,a_0)$ and $(a_2,a_1)$ belong to $T_1\cup T_2$ but not to $T_1\cap T_2$. One can check that with transitive $T_1$ and $T_2$  this allows for only two options: 
	(i) $(a_1,a_0), (a_2,a_1) \in T_1$ and $(a_0,a_2)\in T_2$, or 
	(ii)    $(a_1,a_0), (a_2,a_1) \in T_2$ and $(a_0,a_2)\in T_1$. 
	In both cases applying transitivity of $T_1$ or of $T_2$ we have $(a_2,a_0)\in T_1\cup T_2$. But then the existential conjuncts require a new witness, say $a_3$, for $a_2$ such that $(a_2,a_3)\not\in T_1\cup T_2$.
	Again, taking the universal conjuncts into considerations, we get $(a_3,a_1)\in T_1 \cup T_2$. So, the situation repeats, and indeed $\fA$ embeds an infinite chain of elements such that, for each consecutive pair, $(a_i,a_{i+1})\not\in T_1\cup T_2$.
\end{proof}

\subsection{Fluted types and cliques}
\label{sec:ftc}
Suppose $\fA$ is a structure interpreting the distinguished binary predicate $T$ as a transitive relation.
A {\em clique} of $\fA$ is a maximal subset $B \subseteq A$
with the property that, for all distinct $a, b \in B$, $\fA \models T[a,b]$. Every element $a \in A$ is a member of exactly one clique,
and if that clique has size greater than 1, then, necessarily $\fA \models T[a,a]$.
Furthermore, if $B_1$ and $B_2$ are cliques, then
either every element of $B_1$ is related to every element of $B_2$ by $T$, or no element of $B_1$ is related to any element of $B_2$ by $T$. In this way, $T^\fA$ induces a strict partial order on the set of cliques.
If a singleton $\set{a}$ is a clique, then it may or may not be the case that $\fA \models T[a,a]$.
If $\fA \models \neg T[a,a]$, then we call $a$ (or sometimes $\set{a}$)  a {\em soliton}. 

In this paper, we adapt the familiar notions of atom, literal, $m$-type and clause to the fluted environment. A \textit{fluted $m$-atom} is an 
atomic formula of $\FL_=^{[m]}$. 
Remembering that we are using variable-free syntax, we see that 
a fluted $m$-atom is simply a predicate $p$ having arity at most $m$. A \textit{fluted $m$-literal} is a fluted $m$-atom or its negation;  a \textit{fluted $m$-type} is a maximal consistent conjunction of fluted $m$-literals. If $\bar{a}= a_1, \dots, a_m$ is a tuple of 
elements in some structure $\fA$, then $\bar{a}$ satisfies a unique fluted $m$-type over $\Sigma$, denoted $\ftp^\fA[\bar{a}]$.  We silently identify fluted $m$-types with their conjunctions where appropriate; thus, 
any fluted $m$-type may be regarded as a (quantifer-free) $\Fl^{[m]}$-formula.
Finally, 
a \textit{fluted $m$-clause}
is a disjunction of fluted $m$-literals. We allow the
empty clause $\bot$. We silently identify a finite set of clauses
$\Gamma$ 
with its conjunction where convenient, thus writing $\Gamma$ in formulas instead of the (technically more correct)
$\bigwedge \Gamma$.
A fluted $m$-atom/literal/clause is automatically a fluted $m'$-atom/literal/clause for all $m' > m$; the same is not true of 
fluted $m$-types for signatures containing predicates of arity greater than $m$. In any case, reference to $m$ is suppressed if inessential or clear from context.

At various points in Sec.~\ref{sec:flTransEq}, it will be convenient to appeal to the technique of resolution theorem-proving
in order to simplify formulas. 
If $\gamma= \gamma' \vee A$ and $\delta= \delta' \vee \neg A$ are both fluted $m$-clauses, where $A$ is a fluted atom, then so is the clause $\gamma' \vee \delta'$, called a {\em fluted resolvent} of $\gamma$ and $\delta$. 
If the predicate in $A$ is ordinary and has maximum arity both among the predicates
of $\gamma$ and among those of $\delta$, then we say that $\gamma' \vee \delta'$
is the {\em maximal ordinary resolvent} (or {\em mo-resolvent}) of $\gamma$ and $\delta$.
(Recall that a\nb{I: 'a' inserted.} predicate is called {\em ordinary} if it is neither the equality predicate nor one of the distinguished predicates.\nb{I: $T_k$ removed})
Thus mo-resolution is simply a restricted version of resolution.
By regarding fluted $m$-clauses as shorthand for their universal closures,
{resolution}---and in particular mo-resolution---can be seen as a valid inference rule: 
from $\forall^m(\gamma' \vee A)$ and $\forall^m(\delta' \vee \neg A)$, infer
$\forall^m(\gamma' \vee \delta')$. 
We remark that, if $A$ is the only literal of $\gamma$ involving an $m$-ary predicate, and similarly for $\neg A$ in $\delta$, then
the mo-resolvent $\gamma' \vee \delta'$ will be a fluted $(m-1)$-clause (and therefore also a fluted $m$-clause). This is will prove
important when dealing with the fragments $\FlEqTransVars{k}{m}$ for $m > 2$.\nb{I: extra explanation added.}


If $\Gamma$ is a set of fluted clauses, denote by $\Gamma^*$ the smallest
set of fluted clauses including $\Gamma$ and closed under mo-resolution, in the sense that if
$\gamma, \delta \in \Gamma^*$ \mbox{mo-resolve} to form $\epsilon$, then $\epsilon \in \Gamma^*$. Clearly, $\Gamma^*$ is finite if $\Gamma$ is. Further, if
$\Gamma$ is a set of $m$-clauses, for $m \geq 2$, and taking $m$ to be clear from context, we denote 
by $\Gamma^\circ$ the result of removing from $\Gamma^*$ any
clauses involving any ordinary predicates of arity $m$. If $m > 2$, then
$\Gamma^\circ$ is necessarily a set of fluted ($m-1$)-clauses; and if $m=2$,
then $\Gamma^\circ$ is a set of fluted $2$-clauses involving no binary predicates other than (possibly) $=$ or 
the distinguished predicates $T_k$. 

The following lemma is, in effect, nothing more than the familiar completeness theorem for (ordered) propositional resolution.
\begin{lemma}
	Let $\Gamma$ be a set of fluted $m$-clauses, and $\tau$ a fluted $m$-type over the signature of $\Gamma^\circ$.  If $\Gamma^\circ 
	\cup \{\tau\}$ is consistent, then there exists a fluted type $\tau^+$ 
    over the signature of $\Gamma$ such that
	$\tau^+ \supseteq \tau$ and $\Gamma \cup \{\tau^+\}$ is consistent.
	\label{lma:resolution}
\end{lemma}
\begin{proof}
	Enumerate the ordinary
    $m$-ary predicates occurring in $\Gamma$ as\linebreak
    $p_1, \dots, p_n$. Note that none of these predicates occurs in $\tau$.
	Define a \textit{level-$i$ extension}
	of $\tau$ inductively as follows: (i) $\tau$ is an level-0 extension of $\tau$; (ii) if $\tau'$ is a level-$i$ extension of $\tau$ ($0 \leq i < n$),
    then $\tau' \cup \{p_{i+1}\}$  and $\tau' \cup \{\neg p_{i+1}\}$ are 
	level-$(i+1)$ extensions of $\tau$. Thus, the level-$n$ extensions of $\tau$ are exactly 
	the fluted $m$-types over the signature of $\Gamma$ extending $\tau$. 
	If $\tau'$ is a level-$i$ extension of $\tau$ ($0 \leq i < n$), we say that $\tau'$ {\em violates} a clause $\delta$
	if, for every literal in $\gamma$, the opposite literal is in $\tau'$; we say that $\tau'$ {\em violates} a set of clauses
	$\Delta$ if $\tau'$ violates some $\delta \in \Delta$.
	Suppose now that $\tau'$ is a level-$i$ extension of $\tau$ ($0 \leq i < n$). We claim that, if both 
	$\tau' \cup \{p_{i+1}\}$  and $\tau' \cup \{\neg p_{i+1}\}$ violate $\Gamma^*$, then so does $\tau'$. For suppose
    otherwise. In that case, there must be a clause $\neg p_{i+1} \vee \gamma' \in \Gamma^*$ violated by $\tau' \cup \{p_{i+1}\}$ and a clause $p_{i+1} \vee \gamma' \in \Gamma^*$ violated by $\tau' \cup \{\neg p_{i+1}\}$. But then 
$\tau'$ violates the mo-resolvent $\gamma' \vee \delta'$,  contradicting the
	supposition that $\tau'$ does not violate $\Gamma^*$. This proves the claim. Now, since $\tau$ 
	is by hypothesis consistent with $\Gamma^\circ$, it certainly does not violate $\Gamma^\circ$. Moreover, since it involves no ordinary predicates of
	arity $m$, $\tau$ does not violate $\Gamma^*$ either. By the above claim, then,
	there must be at least one level-$n$ extension of $\tau$ which does not
	violate $\Gamma^* \supseteq \Gamma$. Since $\tau^+$ is a fluted $m$-type, this proves the lemma.
\end{proof}

%% file: 1-DecidableEq.tex
\section{The decidability of fluted logic with one transitive relation and equality}\label{sec:onetrans}
\label{sec:flTransEq}
In this section, we study the logic $\FlEqTrans{1}$, the fluted fragment with equality and a single distinguished transitive relation; we also consider 
its $m$-variable sub-fragment, $\FlEqTransVars{1}{m}= \FlEqTrans{1} \cap \FO^m$, for all $m \geq 2$.
As already mentioned, even the smallest of these fragments lacks the finite model property. Nevertheless, we show that the satisfiability problem for $\FlEqTrans{1}$ is decidable; indeed, 
$\Sat(\FlEqTransVars{1}{m})$ is in $(m+1)$-$\NExpTime$ for $m \geq 2$. Given known results on the fluted fragment, it follows that 
$\Sat(\FlEqTrans{1})$ is \Tower-complete, according to the framework of super-elementary complexity classes developed in Schmitz~\cite{purdyTrans:schmitz16}.
The structure of the proof is as follows.
Recall that $\FLotranstEqMinus$ is the sub-fragment of $\FlEqTransVars{1}{2}$ in which no binary predicates appear other than $T$ and $=$. In Sec.~\ref{sec:FLotranstEqMinus}, we prove an upper complexity bound of $\TwoNExpTime$ for $\Sat(\FLotranstEqMinus)$; in Sec.~\ref{sec:FLotranstEq}, we show that $\Sat(\FlEqTransVars{1}{2})$ is also in \TwoNExpTime, via a reduction to $\Sat(\FLotranstEqMinus)$; and in Sec.~\ref{sec:FLotransmEq}, we show that $\Sat(\FlEqTransVars{1}{m})$ is in $m$-\NExpTime, via a series of exponential-sized reductions to $\Sat(\FlEqTransVars{1}{2})$. In all these reductions, we take particular care of the sizes both of the formulas produced, and of their signatures.

We will be dealing here with logics featuring a single distinguished transitive relation, and we use the letter $T$ for the corresponding binary predicate. Thus, if $\fA$ is a structure, we always assume that $T^\fA$ is a transitive relation on $A$.
A formula of $\FlEqTransVars{1}{m}$ is said to be in {\em normal form} if it has the shape
\begin{equation}
{\bigwedge_{i\in S} }\forall^{m-1} (\mu_i \rightarrow \exists (\kappa_i \wedge \Gamma_i)) \wedge
{\bigwedge_{j\in T}} \forall^{m-1} (\nu_j  \rightarrow \forall \Delta_j) \wedge \forall^m \Omega,
\label{eq:flMnf}
\end{equation}
where $S$ and $T$ are finite sets of indices, such that, for $i \in S$ and $j \in T$, 
$\mu_i$ and $\nu_j$ are quantifier-free fluted formulas of arity at most $(m-1)$,
$\kappa_i$ is a formula of any of the four forms 
$(T \wedge =)$,
$(T \wedge \neq)$,
$(\neg T \wedge =)$,
$(\neg T \wedge \neq)$,
and $\Gamma_i$, $\Delta_j$ and $\Omega$ are sets of fluted clauses  in $\FlEqTransVars{1}{m}$.
(Here, of course, we are making use of our convention that finite sets of clauses are identified with their
conjunctions.) 
We refer to the formulas $\kappa_i$ as {\em control formulas}; observe in this regard that the binary predicates $T$ and $=$ count as
atomic formulas of $\Fl^{[m]}$ for all $m \geq 2$.
The following lemma is slightly modified from~\cite[Lemma~4.1]{P-HST-FLrev}, where it was proved for the sub-fragment without
equality. The proof, however, is virtually identical, and we may simply state:
\begin{lemma}
Let $\phi$ be an $\FlEqTransVars{1}{2}$-sentence. We can compute, in time bounded by 
a polynomial function of $\sizeof{\phi}$, a normal-form $\FlEqTransVars{1}{2}$-formula
$\psi$
such that: (i) $\models \psi \rightarrow \phi$; and (ii) any model of $\phi$ can be expanded to a model 
of $\psi$.
\label{lma:flMnf}
\end{lemma}
We show in Lemmas~\ref{lma:FLotranstEqToFLotranstEqSpread} and~\ref{lma:FLotranstEqToFLotranstEqMinus} how, in the
two-variable case, normal form formulas can be further massaged into a collection of extremely simple formulas for which
the satisfiability problem is easy to analyse. Since that analysis forms the core of the whole proof, 
that is where we shall begin.
 
\subsection{Basic formulas in $\FLotranstEqMinus$}
\label{sec:FLotranstEqMinus}
In
the logic $\FLotranstEqMinus$, the only binary predicates available are equality and the distinguished predicate, $T$. These
suffice, however, to state that an element is related by $T$ to itself, for example, using the unary formula $\exists(= \wedge\; T)$.
We may therefore suppose that we have available a 
distinguished {\em unary} predicate $\hat{T}$, which
we take to be satisfied, in any structure, by precisely those elements related to themselves by $T$: 
i.e.~$\fA \models \hat{T}[a] \Leftrightarrow \fA \models T[a,a]$; this 
constitutes no essential increase in the expressive power of $\FLotranstEqMinus$.
In this section (\ref{sec:FLotranstEqMinus}), then, all signatures are implicitly assumed to 
contain both $T$ and $\hat{T}$, interpreted as described.
Under this assumption, a {soliton} is a clique consisting of a single element $a$ such that
$\fA \not \models \hat{T}[a]$.\nb{I: duplicated paragraph removed.} 

Our goal is to establish
that the satisfiability problem for this fragment 
is in \TwoNExpTime. In fact, it suffices to confine our attention to 
conjunctions of so-called {\em basic} formulas of this fragment (defined below). Our strategy is to 
show that any satisfiable, finite set $\Psi$ of basic formulas has a \textit{certificate}, of
size bounded by a doubly exponential function of $\sizeof{\Psi}$, which guarantees the existence of a (possibly infinite) model.



Let $\Sigma$ be a signature for $\FLotranstEqMinus$. Call an $\FLotranstMinus$-formula over $\Sigma$ 
{\em basic} if it is of one of the following forms, where $\pi$ and $\pi'$ are fluted 1-types over $\Sigma$ 
and $\mu$ a quantifier-free formula over $\Sigma$ of arity 1:
\begin{multicols}{2}
\begin{enumerate}[(B1)]
\item[(B1)] $\forall (\pi \rightarrow \exists (\mu \wedge T\; \wedge  \neq))$
\item[(B2)] $\forall (\pi \rightarrow \exists (\mu \wedge \neg T\; \wedge \neq))$
\item[(B3)] $\forall (\pi \rightarrow \forall (\pi' \rightarrow T))$ \qquad ($\pi \neq \pi'$)
\item[(B4)] $\forall (\pi \rightarrow \forall (\pi' \rightarrow \neg T))$ \ \quad ($\pi \neq \pi'$)
\item[(B5)] $\forall (\pi \rightarrow \forall (\pi \rightarrow (= \vee\; T))$
\item[(B6)] $\forall (\pi \rightarrow \forall (\pi \rightarrow (= \vee\; \neg T))$ 
\item[(B7)] $\forall \mu$
\item[(B8)] $\exists \mu$.
\end{enumerate}	
\end{multicols}

Suppose $\fA$ is a structure, $B$ a clique of $\fA$,
and $\pi$, $\pi'$ fluted 1-types. Say that $B$ is {\em determined} by the pair $\set{\pi, \pi'}$
if it is the unique clique of $\fA$ in which $\pi$ and $\pi'$ are both realized.
We call $\fA$ {\em quadratic} if, for any clique $B$ determined by some pair of fluted 1-types $\set{\pi, \pi'}$, 
there exists a fluted 1-type $\pi^*$ such that $B$ is the unique clique of $\fA$ in which $\pi^*$ is
realized. That is, in a quadratic structure, any
clique which can be uniquely identified as the only clique containing a given pair of fluted 1-types, $\pi$ and $\pi'$, can be 
uniquely identified as the only clique containing some (possibly different) fluted 1-type $\pi^*$.

Let $\Phi$ be a set of basic formulas over some signature $\Sigma$, and write 
$\ell= |\Sigma|$. Now let $\Sigma^*$ be $\Sigma$ 
together with the fresh unary predicates $p_0, \dots p_{2\ell-1}$, 
let $\bar{p}_0$ be the formula \mbox{$\neg p_0 \wedge \cdots \wedge \neg p_{2\ell-1}$}, and let
$\Phi^* = \set{\phi^* \mid \phi \in \Phi \cup \set{\exists \top}}$, where
\begin{equation*}
\phi^* :=
\begin{cases}
\forall(\pi \wedge \bar{p}_0 \rightarrow \exists(\mu \wedge \bar{p}_0 \wedge \chi) & \text{if $\phi= \forall(\pi \rightarrow \exists (\mu \wedge \chi))$}\\
\forall(\pi \wedge \bar{p}_0 \rightarrow \forall(\pi' \wedge \bar{p}_0 \rightarrow \chi) & \text{if $\phi= \forall(\pi \rightarrow \forall (\pi' \rightarrow \chi))$}\\
\forall(\bar{p}_0 \rightarrow \mu) & \text{if $\psi= \forall \mu$}\\
\exists(\mu \wedge \bar{p}_0) & \text{if $\psi= \exists \mu$.}\\
\end{cases}
\end{equation*}
Modulo trivial logical manipulation,  $\Phi^*$ is a set of basic formulas over $\Sigma^*$.
Call any fluted 1-type $\pi$ over $\Sigma^*$ such that $\models \pi \rightarrow \bar{p}_0$ {\em proper}. Clearly, the proper
fluted 1-types over $\Sigma^*$ are in natural 1--1 correspondence with the fluted 1-types over $\Sigma$. 
\begin{lemma}
	Suppose $\Phi$ is a set of basic formulas.
	The following are equivalent: \textup{(i)} $\Phi$ is satisfiable; \textup{(ii)} $\Phi^* \cup \set{\forall \bar{p}_0}$ is satisfiable; \textup{(iii)}
	$\Phi^*$ is satisfied in a quadratic structure; \textup{(iv)}
	$\Phi^*$ is satisfiable.
	\label{lma:quadratic}
\end{lemma}
\begin{proof}
	(i) $\Rightarrow$ (ii): If $\fA \models \Phi$, let $\fB$ be the expansion of $\fA$ obtained by taking every element of $A$
	to satisfy $\bar{p}_0$. It is obvious that $\fB \models \Phi^* \cup \set{\forall \bar{p}_0}$. 
	(ii) $\Rightarrow$ (iii): Suppose $\fA \models \Phi^* \cup \set{\forall \bar{p}_0}$. For each (unordered) pair,
	$\pi$, $\pi'$ of
	distinct, proper fluted 1-types (over $\Sigma^*$) such that there is exactly one clique, $u$ of $\fA$ in which both are realized,
	choose a fresh, {\em improper} fluted 1-type over $\Sigma^*$, and simply add a new element with that fluted 1-type to $u$. 
	Because there are
	certainly $2^{2|\sigma|}-1$ improper fluted 1-types, we never run out of fresh, improper fluted 1-types, so let $\fB$ be the resulting structure. 
	Since the new elements do not satisfy $\bar{p}_0$, we have $\fB \models \Phi^*$. And since all the
	newly realized fluted 1-types occur only in single cliques, $\fB$ is quadratic. 
	(iii) $\Rightarrow$ (iv) is trivial.
	(iv) $\Rightarrow$ (i): 
	Suppose $\fA \models \Phi^*$, and let $\fB$ be
	restriction of $\fA$ to the (necessarily non-empty) set of elements satisfying $\bar{p}_0$. It is obvious that $\fB \models \Phi$. 
\end{proof}
Lemma~\ref{lma:quadratic} tells us that any set $\Phi$ of basic formulas over $\Sigma$ can be transformed, in polynomial
time, to a set $\Phi^*$ of basic formulas over a larger signature $\Sigma^*$ such that $\Phi$ has a model if and only if $\Phi^*$ has a quadratic model. In the following lemmas, therefore, we may assume this conversion has been carried out, and concern ourselves with establishing conditions for a set of basic formulas $\Phi$ to have a \textit{quadratic} model.

For the remainder of Sec.~\ref{sec:FLotranstEqMinus},  we 
fix a signature $\Sigma$ of unary predicates.
All fluted 1-types are assumed to be over the signature $\Sigma$,
and are, as usual, identified with their conjunctions where convenient.
We denote by $\Pi_\Sigma$ the set of these fluted 1-types.
We always use the (possibly decorated) letters $\pi$ to range over fluted 1-types, 
and $\mu$ to range over quantifier-free formulas of arity 1 in the signature $\Sigma$. Thus, all such 
$\pi$ and $\mu$ are $\FLotranstEqMinus$-formulas. We use $\Pi$ to range over sets of fluted 1-types.

A \textit{clique-type} is a function $\xi: \Pi_\Sigma \rightarrow \set{0,1,2}$. 
If $\fA$ is a structure interpreting $\Sigma$,
$B$ is a clique of $\fA$, and $a \in B$, then the {\em clique-type of} $B$ is the function $\ctp^\fA[a]: \Pi_\Sigma \rightarrow \set{0,1,2}$ given by
\begin{equation*}
\ctp^\fA[a](\pi) = 
\begin{cases}
2 & \text{if $\pi$ is realized in $\fA$ by at least two elements of $B$}\\
1 & \text{if $\pi$ is realized in $\fA$ by exactly one element of $B$}\\
0 & \text{otherwise.}
\end{cases}
\end{equation*}
Intuitively, we should think of a clique type as a multi-set of fluted 1-types, with counting truncated at 2. 
We write $\pi \in \xi$ to mean that $\xi(\pi) \geq 1$, and treat $\xi$ as the set of fluted 1-types $\set{\pi \mid \pi \in \xi}$ 
where convenient, thus writing, for example $\xi \cup \Pi$ for  $\set{\pi \mid \pi \in \xi \mbox{ or } \pi \in \Pi}$,
and so on.  A {\em soliton clique-type} $\xi$ is one such that $\neg \hat{T} \in \bigcup \xi$. A  
\textit{clique-super-type} is a pair $(\xi, \Pi)$, where $\xi$ is a clique-type and $\Pi$ a set of fluted 1-types.
The {\em clique-super-type of} $a$ is the pair $\cstp^\fA[a]= (\ctp^\fA[a], \Pi)$, where
\begin{align*}
& \Pi = \set{\ftp^\fA[b] \mid \text{$\fA \models T[a,b]$ and $\fA \not \models T[b,a]$ for some $b \in A$}}.
\end{align*}
Intuitively, a clique-super-type is the type of some clique together with a specification of which 
fluted 1-types outside that clique can be reached via the predicate $T$.\nb{I: tiny change} 
If $B$ is a clique, then all elements of $B$ obviously have the same clique-type and the same clique-super-type,
denoted by $\ctp^\fA[B]$ and $\cstp^\fA[B]$, respectively.

We  now describe the principal data-structure used to test satisfiability of sets of basic $\FLotranstEqMinus$-formulas.
A {\em certificate} is a triple $\cC= \langle \Omega, \ll, V \rangle$, where $\Omega$ is a set of clique super-types,
$\ll$ a strict partial order on $\Pi_\Sigma$, and $V \subseteq \Pi_\Sigma$, subject to the following conditions:
\begin{enumerate}[(C1)]
\item [(C1)] if $\langle \xi, \Pi \rangle \in \Omega$ and $\pi' \in \Pi$, then there exists $\langle \xi', \Pi' \rangle \in \Omega$ 
such that\newline
(i) $\pi' \in \xi'$, (ii)  $\Pi' \cup \xi' \subseteq \Pi$, and (iii) $\xi \cap V \cap \Pi' = \emptyset$; 
\item [(C2)] if $\langle \xi, \Pi \rangle, \langle \xi', \Pi' \rangle \in \Omega$ are distinct, $\pi \in \xi$, $\pi' \in \xi'$ and $\pi \ll \pi'$, then
$\xi' \cup \Pi' \subseteq \Pi$;
\item [(C3)] if $\langle \xi, \Pi \rangle, \langle \xi', \Pi' \rangle \in \Omega$ and $\xi \cap \xi' \cap V \neq \emptyset$, then $\xi = \xi'$ and $\Pi = \Pi'$;
\item [(C4)] if $\langle \xi, \Pi \rangle \in \Omega$ and $\xi$ is a soliton clique-type, then there exists $\pi \in \Pi_\Sigma$ such that $\xi(\pi)= 1$ and
              $\xi(\pi') = 0$ for all $\pi' \in \Pi_\Sigma \setminus \set{\pi}$;
\item [(C5)] if $\langle \xi, \Pi \rangle \in \Omega$, $\pi' \in \xi$ and $\pi \ll \pi'$,  then $\pi \not \in \Pi$;
\item [(C6)] if $\langle \xi, \Pi \rangle \in \Omega$,  $\pi, \pi' \in \xi$ and $\pi \ll \pi'$ then $\xi \cap V \neq \emptyset$.
\end{enumerate}
If $\fA$ is a structure, then the {\em certificate of} $\fA$ is the tuple $\cC(\fA)= \langle \Omega, \ll, V \rangle$, 
where: $\Omega = \set{\cstp^\fA[a] \mid a \in A}$
is the set of clique-super-types realized in $\fA$; $\pi \ll \pi'$ if and only
if $\pi$ and $\pi'$ are realized in $\fA$, $\fA \models \forall(\pi \rightarrow \forall(\pi' \rightarrow T))$ and
$\fA \not \models \forall (\pi' \rightarrow \forall (\pi \rightarrow T))$; and
$V$ is the set of fluted 1-types realized in exactly one clique of $\fA$.
\begin{lemma}
	The relation $\ll$ in the construction of $\cC(\fA)$ is a strict partial order on $\Pi_\Sigma$.
	\label{lma:isaPOEq}
\end{lemma}
\begin{proof}
We need only check transitivity. Suppose, $\pi \ll \pi'$ and $\pi' \ll \pi''$. Trivially, $\fA \models \forall(\pi \rightarrow \forall(\pi'' \rightarrow T))$.
On the other hand, if we also have $\fA \models \forall(\pi'' \rightarrow \forall(\pi \rightarrow T))$, then $\fA \models \forall(\pi'' \rightarrow \forall(\pi' \rightarrow T))$,
contradicting $\pi' \ll \pi''$. Hence $\pi \ll \pi''$.
\end{proof}

\begin{lemma}
	If $\fA$ is any quadratic structure interpreting $\Sigma$, then $\cC(\fA)$ is a certificate.
	\label{lma:isaCerificateEq}
\end{lemma}
\begin{proof}
	Write $\cC(\fA)= \langle \Omega, \ll, V \rangle$. By Lemma~\ref{lma:isaPOEq}, $\ll$ is  a strict partial order on $\Pi_\Sigma$. 
	We must check conditions ({C1})--({C6}). 
\medskip

\noindent
({C1}): Suppose $\langle \xi, \Pi \rangle \in \Omega$ and $\pi' \in \Pi$. Let $a$ be
such that $\cstp^\fA[a] = \langle \xi, \Pi \rangle$. Then there exists $b \in A$
such that $\ftp^\fA[b] = \pi'$ and $\fA \models T[a,b]$, but with 
$a$ and $b$ lying in different cliques. Let 
$\cstp^\fA[b] = \langle \xi', \Pi' \rangle$. Then: (i) $\langle \xi', \Pi' \rangle \in \Omega$ by construction of $\Omega$;
(ii) $\xi' \cup \Pi' \subseteq \Pi$ by transitivity of $T^\fA$; and (iii) if $\pi'' \in \xi \cap V$, then all
elements with fluted 1-type $\pi''$ lie in the same clique as $a$. Since $a$ and $b$ are not in the
same clique, $b$ cannot be related by $T$ to any of these elements, which is to say $\pi'' \not \in \Pi'$. 

\noindent
({C2}): Suppose  $\langle \pi, \Pi \rangle, \langle \pi', \Pi' \rangle \in \Omega$ are distinct, $\pi \in \xi$, $\pi' \in \xi'$ and $\pi \ll \pi'$.
Let $a, b \in A$ be such that 
$\cstp^\fA[a] = \langle \xi, \Pi \rangle$ and $\cstp^\fA[b] = \langle \xi', \Pi' \rangle$. If $\pi \ll \pi'$, then
$\fA \models T[a,b]$. Moreover, if $a$ and $b$ belong to different cliques, then $\xi' \cup \Pi' \subseteq \Pi$, by the transitivity of $T$.

\noindent
({C3}): 
Suppose $\langle \xi, \Pi \rangle, \langle \xi', \Pi' \rangle \in \Omega$ and $\xi \cap \xi' \cap V \neq \emptyset$. 
Let $a, b \in A$ be such that 
$\cstp^\fA[a] = \langle \xi, \Pi \rangle$ and $\cstp^\fA[b] = \langle \xi', \Pi' \rangle$. If there exists a fluted 1-type $\pi''$
realized both in the clique of $a$ and in the clique of $b$, and, moreover, in just one clique of $\fA$, then $a$ and $b$ are
in the same clique.

\noindent
({C4}):
Suppose $\langle \xi, \Pi \rangle \in \Omega$ and $\neg \hat{T} \in \bigcup \xi$. 
By construction, there exists $b \in A$ such that $\ctp^\fA[b] = \xi$, and $\fA \not \models \hat{T}[b]$.
But then $b$ is the only element of its clique, and we may set $\pi = \ftp^\fA[b]$.

\noindent
({C5}):
Suppose $\langle \xi, \Pi \rangle \in \Omega$, $\pi' \in \xi$ and $\pi \ll \pi'$. 
Let $a, a' \in A$ be such that 
$\cstp^\fA[a] = \langle \xi, \Pi \rangle$, $\ftp^\fA[a'] = \pi'$, and 
$a'$ is in the same clique as $a$. To show that $\pi \not \in \Pi$, 
we must show that, for all $b \in A$ such that $\ftp^\fA[b] = \pi$, either $\fA \not \models T[a,b]$
or $b$ is in the same clique as $a$. But this follows immediately from
$\pi \ll \pi'$.

\noindent
({C6}): 
Suppose $\langle \xi, \Pi \rangle \in \Omega$,  $\pi, \pi' \in \xi$ and $\pi \ll \pi'$. It follows that there is exactly one
clique of $\fA$, say $u$, in which $\pi$ and $\pi'$ are both realized, and that 
$\cstp^\fA[u] = \langle \xi, \Pi \rangle$. Since $\fA$ is, by assumption,  quadratic, there exists a fluted 1-type $\pi^* \in \xi$ realized
only in $u$. Thus $\xi \cap V \neq \emptyset$. 
\end{proof}

Now suppose $\cC= \langle \Omega, \ll, V \rangle$ is a certificate. We proceed to define a structure $\fA$.
As an aide to intuition, we give an informal sketch first. The domain $A$ is the disjoint union of sets $A_{\xi, \Pi}$,
where $(\xi, \Pi)$ ranges over $\Omega$; the elements of $A_{\xi, \Pi}$ will all be assigned the clique-super-type $(\xi, \Pi)$.
If $\xi$ contains no fluted 1-type $\pi$ such that $\pi \in V$, then $A_{\xi, \Pi}$ will consist\nb{I: tiny change} of infinitely many sets 
$A_{\xi, \Pi, i}$ ($i \geq 0$), referred to in the construction as `cells'. (It will later turn out that the cells are exactly the $T$-cliques.)
If, on the other hand, $\xi$ contains a fluted 1-type $\pi$ such that $\pi \in V$, then $A_{\xi, \Pi}$ will consist of a single cell $A_{\xi,\Pi,0}$.
Note that, in the latter case, there will only ever be a single pair $(\xi, \Pi) \in \Omega$  such that $\pi \in \xi$, by ({C3}). Each cell
$A_{\xi,\Pi,i}$ is in turn the disjoint union of sets $A_{\pi,\xi,\Pi,i}$, where $\pi$ ranges over the fluted 1-types in $\xi$. Each element of the
set $A_{\pi,\xi,\Pi,i}$ will be given fluted 1-type $\pi$, and this set has cardinality equal to $\xi(\pi)$ (i.e.~either 1 or 2).
Fig.~\ref{fig:modelConstruction} gives a schematic representation of the domain $A$, showing some representative sets $A_{\xi, \Pi}$;
here, $\xi$ contains the fluted 1-types $\pi_1$, $\pi_2$ and $\pi_3$ with the indicated multiplicities.\nb{: New paragraph break
	and new diagram.}
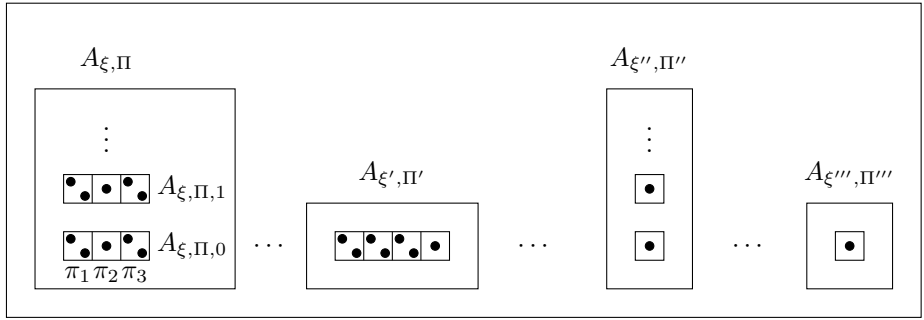
\begin{figure}
\begin{center}
\begin{tikzpicture}[scale=0.38]
\draw (-1,-1) rectangle (31,10);

\draw (0,0) rectangle (7,7);
\draw (2.5,8) node {$A_{\xi,\Pi}$};
\draw (1.5,0.5) node {$\pi_1$};
\draw (2.5,0.5) node {$\pi_2$};
\draw (3.5,0.5) node {$\pi_3$};
\draw (1,1) rectangle (4,2);
\draw (5.5,1.5) node {$A_{\xi,\Pi,0}$};
\filldraw (1.25,1.75) circle (0.15);
\filldraw (1.75,1.25) circle (0.15);
\draw (2,1) -- (2,2);
\filldraw (2.5,1.5) circle (0.15);
\draw (3,1) -- (3,2);
\filldraw (3.25,1.75) circle (0.15);
\filldraw (3.75,1.25) circle (0.15);
\draw (1,3) rectangle (4,4);
\draw (5.5,3.5) node {$A_{\xi,\Pi,1}$};
\filldraw (1.25,3.75) circle (0.15);
\filldraw (1.75,3.25) circle (0.15);
\draw (2,3) -- (2,4);
\filldraw (2.5,3.5) circle (0.15);
\draw (3,3) -- (3,4);
\filldraw (3.25,3.75) circle (0.15);
\filldraw (3.75,3.25) circle (0.15);
\draw (2.5,5.5) node {$\vdots$};

\draw (8.25,1.5) node {$\cdots$};

\draw (9.5,0) rectangle (15.5,3);
\draw (12.5,4) node {$A_{\xi',\Pi'}$};
\draw (10.5,1) rectangle (14.5,2);
\filldraw (10.75,1.75) circle (0.15);
\filldraw (11.25,1.25) circle (0.15);
\draw (11.5,1) -- (11.5,2);
\filldraw (11.75,1.75) circle (0.15);
\filldraw (12.25,1.25) circle (0.15);
\draw (12.5,1) -- (12.5,2);
\filldraw (12.75,1.75) circle (0.15);
\filldraw (13.25,1.25) circle (0.15);
\draw (13.5,1) -- (13.5,2);
\filldraw (14,1.5) circle (0.15);

\draw (17.5,1.5) node {$\cdots$};

\draw (20,0) rectangle (23,7);
\draw (21.5,8) node {$A_{\xi'',\Pi''}$};
\draw (21,1) rectangle (22,2);
\filldraw (21.5,1.5) circle (0.15);
\draw (21,3) rectangle (22,4);
\filldraw (21.5,3.5) circle (0.15);
\draw (21.5,5.5) node {$\vdots$};

\draw (25,1.5) node {$\cdots$};

\draw (27,0) rectangle (30,3);
\draw (28.5,4) node {$A_{\xi''',\Pi'''}$};
\draw (28,1) rectangle (29,2);
\filldraw (28.5,1.5) circle (0.15);
\end{tikzpicture}
\end{center}
\caption{Construction of the domain $A$ of $\fA(\cC)$ for $\cC$ a certificate.}
\label{fig:modelConstruction}
\end{figure}

The relation $T$
is defined as the transitive closure of the union of three relations, $t_0$, $t_1$ and $t_2$, each of which plays a specific role. The 
relation $t_0$ specifies $T$ within each cell, $A_{\xi,\Pi,i}$. As long as $\xi$ contains no fluted 1-type $\pi$ such that $\neg \hat{T} \in \pi$, we take
$t_0$ to be the total relation on $A_{\xi,\Pi,i}$. If, on the other hand, $\xi$ does contain a fluted 1-type $\pi$ such that $\neg \hat{T} \in \pi$, then we take
$t_0$ to be the empty relation on $A_{\xi,\Pi,i}$. Note that, in the latter case, $A_{\xi,\Pi,i}$ is in fact a singleton, by ({C4}). 
The relation $t_1$, in essence, secures the existential commitments required by the clique-super-types. Specifically, if
$a \in A_{\xi,\Pi,i}$ and $\pi' \in \Pi$, we select some $(\xi', \Pi') \in \Omega$ such that $\xi' \cup \Pi' \subseteq \Pi$ (possible by
({C1})), and choose cells included in $A_{\xi',\Pi'}$ whose elements will act as `witnesses' for the fact that $a$ has to be related 
by $T$ to something of type $\pi'$. We need to be careful which cells we choose, however, because there is a danger of creating loops in the
resulting graph of $t_1$-links, which would result in the merging of more than one cell into a single clique. To avoid such loops,
if $a \in A_{\xi,\Pi,i}$, we generally pick witnesses in $A_{\xi',\Pi',i+2}$ (so that, in particular, the last index increases). In one case, however,
we must break this rule: if $\xi'$ contains a fluted 1-type $\pi''$ such that $\pi'' \in V$, then $A_{\xi', \Pi',j}$ exists only for the value $j=0$, and we need to do some
work to ensure that unwanted loops do not arise here. Finally, the relation $t_2$ deals with the $T$-relations mandated by $\ll$. If $A_{\xi,\Pi,i}$ and $A_{\xi',\Pi',j}$ are distinct cells with $\pi \in \xi$ and $\pi' \in \xi'$, where $\pi \ll \pi'$, we take all elements of the former cell to be 
related by $t_2$ to all elements of the latter. Again, we need to do some work to ensure that this does not generate unwanted loops in the graph of
$t_1$- and $t_2$-links.

Turning to the formal definition of $\fA$, we begin with the construction of the domain, $A$. 
For all $(\xi,\Pi) \in \Omega$, 
all $\pi \in \xi$ and all $i \in \N$, let $a^+_{\pi, \xi, \Pi,i}$ and $a^-_{\pi, \xi, \Pi,i}$ be fresh objects. Set
\begin{align*}
A_{\pi, \xi, \Pi,i}  = & 
\begin{cases}
\set{a^+_{\pi, \xi, \Pi,i}, a^-_{\pi, \xi, \Pi,i}} & \text{if $\xi(\pi) = 2$}\\
\set{a^+_{\pi, \xi, \Pi,i}} & \text{otherwise (i.e. if $\xi(\pi) = 1$)}
\end{cases}\\
A_{\xi, \Pi,i}  = & \bigcup_{\pi \in \xi} A_{\pi, \xi, \Pi,i}\\
A_{\xi, \Pi}  = & 
\begin{cases}
\bigcup_{i \in \N} A_{\xi, \Pi,i} & \text{if $\xi \cap V = \emptyset$}\\
A_{\xi, \Pi,0}  & \text{otherwise}
\end{cases}\\
A  = &  \bigcup_{(\xi, \Pi) \in \Omega} A_{\xi, \Pi}.
\end{align*}
The sets $A_{\xi,\Pi,i}$ will be called {\em cells}. If
$\xi$ is a soliton clique-type, we call the cell $A_{\xi,\Pi,i}$ a {\em soliton-cell}.
It follows from ({C4}) that, in this case, $A_{\xi,\Pi,i} = \set{a^+_{\pi, \xi, \Pi,i}}$ for some
fluted 1-type $\pi$. 
Note that the converse does not hold: it is perfectly feasible for the cell $A_{\xi,\Pi,i}$ to consist of
the single element $a^+_{\pi, \xi, \Pi,i}$ even though $\hat{T} \in \pi$. 
Having defined $A$, we may set the extensions of all the ordinary (unary) predicates by
stipulating $\ftp^\fA[a]= \pi$ for all $a= a^p_{\pi, \xi, \Pi,i} \in A$, where $p \in \set{+,-}$. 

It remains only to set the extension of the distinguished predicate $T$. 
To this end, we define three binary relations, $t_0$, $t_1$ and $t_2$. 
Let $a= a^p_{\pi, \xi, \Pi,i}$ and $a'= a^{p'}_{\pi', \xi', \Pi',j}$; and let
$u= A_{\xi,\Pi,i}$ and $v= A_{\xi',\Pi',j}$ be the respective cells of $a$ and $a'$.
We declare
$t_0(a,a')$ if and only if $u=v$ (i.e.~$\xi= \xi'$, $\Pi = \Pi'$, and $i= j$), and $\xi$ is not a soliton clique-type. 
That is:
$t_0$ holds between pairs of elements in the same non-soliton cell. 
Now declare $t_1(a,a')$ if (a) $\xi' \cup \Pi' \subseteq \Pi$; (b) $\xi' \cap V = \emptyset \Rightarrow j \geq i+2$; and
(c) $\xi \cap V \cap \Pi' = \emptyset$. Note that the relation $t_1$ depends only on the \textit{cells} of its 
relata: that is to say,
if $b \in u$  and $b' \in v$,
then $t_1(a,a')$ implies $t_1(b,b')$. There being no ambiguity, we shall write, in this case, $t_1(u,v)$. 
Finally, declare
$t_2(a,a')$ if $u \neq v$ and, for some fluted 1-types $\pi \in \xi$ and $\pi' \in \xi'$, we have $\pi \ll \pi'$.
Again, we write in this case $t_2(u,u')$, since this relation depends only on the cells of its relata. Having defined the relations
$t_0$, $t_1$ and $t_2$, we let $T^\fA$ be the transitive closure of $t_0 \cup t_1 \cup t_2$.
We denote the structure $\fA$, constructed from the certificate $\cC$ as just described, by $\fA(\cC)$. Notice that $\fA(\cC)$ will in general be infinite.

We must check that $\fA(\cC)$ interprets the predicates $T$ and $\hat{T}$ consistently. Lemmas~\ref{lma:noLoopt1}--\ref{lma:tIsOK} do precisely this.
\begin{lemma}
If $t_1(a,a')$, then $a$ and $a'$ occupy different cells of $A$.
\label{lma:noLoopt1}
\end{lemma}
\begin{proof}
Suppose for contradiction that $t_1(a,a')$ with $a= a^p_{\pi, \xi, \Pi,i}$ and $a'= a^{p'}_{\pi', \xi, \Pi,i}$. 
By condition (a) in the definition of
$t_1$, we have $\xi \subseteq \Pi$, and, by condition (b), we have,
$\xi \cap V \neq \emptyset$, whence $\xi \cap V \cap \Pi \neq \emptyset$, contradicting condition (c).
\end{proof}
 
Now consider the directed graph on the set of cells of $A$ defined by the relation $t_1 \cup t_2$. We show that this graph is acyclic. It follows 
that the cells (both soliton and non-soliton) are the cliques of the relation $T^\fA$, and hence that
$T^\fA$ induces a strict partial order on these cells.
\begin{lemma}
Suppose $u_0, \dots, u_k$ \textup{(}$k \geq 1$\textup{)} is a sequence of cells such that, for all $h$ \textup{(}$0 \leq h < k$\textup{)} either $t_1(u_h,u_{h+1})$ or
$t_2(u_h,u_{h+1})$. Writing $u_h= A_{\xi_h,\Pi_h,i_h}$ for all $h$ \textup{(}$0 \leq h \leq k$\textup{)}, we have $\xi_k \cup \Pi_k \subseteq \Pi_0$.
\label{lma:accumulate}
\end{lemma}
\begin{proof}
We proceed by induction on $k$. For the base case ($k=1$) if $t_1(u_0,u_1)$, then the result is immediate by (a) in the definition of $t_1$. If $t_2(u_0,u_1)$, then there exist
$\pi_0 \in \xi_0$ and $\pi_1 \in \xi_1$ such that $\pi_0 \ll \pi_1$. The result then follows from ({C2}). For the inductive case ($k>1$), we have by inductive hypothesis,
$\xi_{k-1} \cup \Pi_{k-1} \subseteq \Pi_0$; and from the base case applied to the sequence $u_{k-1}, u_k$, we have $\xi_k \cup \Pi_k \subseteq \Pi_{k-1}$.
\end{proof}

\begin{lemma}
There exists no sequence of cells
$u_0, \dots, u_k= u_0$ \textup{(}$k \geq 2$\textup{)} such that, for all $h$ \textup{(}$0 \leq h < k$\textup{)} either $t_1(u_h,u_{h+1})$ or
$t_2(u_h,u_{h+1})$. 
\label{lma:noLoop12}
\end{lemma}
\begin{proof}
Suppose for contradiction that such a sequence exists, again writing $u_h= A_{\xi_h,\Pi_h,i_h}$ for all $h$ ($0 \leq h \leq k$). By Lemma~\ref{lma:accumulate}, 
$\Pi_0 = \cdots = \Pi_{k} = \Pi$, say, and $\xi_h \in \Pi$ for all $h$ ($0 \leq h \leq k$). It follows that we cannot have $t_2(u_h,u_{h+1})$ for any $h$ ($0 \leq h < k$),
since, if there exist $\pi_h \in \xi_h$ and $\pi_{h+1} \in \xi_{h+1}$ with $\pi_h \ll \pi_{h+1}$, then, by ({C5}), $\pi_{h+1} \not \in \Pi_h = \Pi$, contradicting $\xi_{h+1} \subseteq \Pi$.
Thus, we may assume that $t_1(u_h,u_{h+1})$ for all $h$ ($0 \leq h < k$). Necessarily, $i_{h+1} \leq i_h$ for some $h$ in the same range; 
indeed, by rotating the original sequence if necessary, we may assume without loss of generality that $h < k-1$.
By (b) in the definition of $t_1$, $\xi_{h+1} \cap V \neq \emptyset$, and by (c), $\xi_{h+1} \cap V \cap \Pi_{h+2} = \emptyset$. But we have just argued that
$\xi_{h+1} \subseteq \Pi$ and  $\Pi_{h+2}=  \Pi$. This is a contradiction.
\end{proof}
\begin{lemma}
In the structure $\fA= \fA(\cC)$, we have $\hat{T}^\fA = \set{a \in A \mid \fA \models T[a,a]}$.
\label{lma:tIsOK}
\end{lemma}
\begin{proof}
Fix $a \in A_{\pi,\xi,\Pi,i}$. If $\fA \models \hat{T}[a]$, then
$\hat{T} \in \pi$, whence, by ({C4}), $\xi$ is not a soliton clique type. Hence $t_0(a,a)$, and $\fA \models T[a,a]$.
Conversely, if $\fA \not \models \hat{T}[a]$, then $\neg \hat{T} \in \pi$, so that $\xi$ is certainly a soliton type, and $a$ is not related to itself by $t_0$. On the other hand, by Lemma~\ref{lma:noLoop12}, there is no sequence of cells
$u_0, \dots, u_k$ \textup{(}$k \geq 2$\textup{)} with $a \in u_0 = u_k$,
such that, for all $h$ \textup{(}$0 \leq h < k$\textup{)}, either $t_1(u_h,u_{h+1})$ or
$t_2(u_h,u_{h+1})$. Since $T^\fA$ is the transitive closure of $t_0 \cup t_1 \cup t_2$, we see that $\fA \not \models T[a,a]$, as required.
\end{proof}

Thus, from a quadratic structure $\fA$, we can define a certificate $\cC(\fA)$, and from a certificate $\cC$, we can define a structure $\fA(\cC)$. (It is easy to see that
$\fA$ will in fact be quadratic, though this is inessential.) 
Let $\cC = \langle \Omega, \ll, V \rangle$ be a certificate and $\psi$ a basic formula. We next
define a relation $\models$ of \textit{satisfaction} between these relata. In this definition,
for any fluted 1-type $\pi$, we say that $\pi$ {\em occurs} in $\cC$ if, there exists $(\xi, \Pi) \in \Omega$ such that $\pi \in \xi$.

\begin{enumerate}
\item $\psi$ is $\forall (\pi \rightarrow \exists(\mu \wedge T\; \wedge \neq))$: $\cC \models \psi$ if and only if, for all $(\xi,\Omega) \in \Omega$, with $\pi \in \xi$,
either (i) $\models \pi \rightarrow \mu$ and $\xi(\pi)=2$; or (ii) there exists 
$\pi' \in \xi$ such that $\pi' \neq \pi$ and $\models \pi' \rightarrow \mu$; or (iii) there exists $\pi' \in \Pi$ such that
$\models \pi' \rightarrow \mu$.
\item $\psi$ is $\forall (\pi \rightarrow \exists (\mu \wedge \neg T\; \wedge \neq))$: $\cC \models \psi$ if and only if, for all $\langle \xi, \Pi \rangle \in 
\Omega$ with $\pi \in \xi$, there exists $\langle \xi', \Pi' \rangle \in \Omega$ such that
(i) $\models \pi' \rightarrow \mu$;
(ii) there exist no $\pi'' \in \Pi$ and $\pi''' \in \xi'$ such that $\pi'' \ll  \pi'''$; (iii)
$\xi' \cap \Pi \cap V = \emptyset$; and (iv) $(\xi,\Pi) = (\xi',\Pi') \Rightarrow \xi \cap V = \emptyset$. 
\item $\psi$ is $\forall (\pi \rightarrow \forall (\pi' \rightarrow T))$, where $\pi \neq \pi'$:  $\cC \models \psi$ if and only if one of the following obtains:
(i) one of $\pi$ or $\pi'$ does not occur in $\cC$; (ii) 
$\pi \ll \pi'$; or (iii) for all $(\xi,\Pi), (\xi',\Pi') \in \Omega$ such that $\pi \in \Pi$ and $\pi' \in \xi'$, we have $\xi = \xi'$, $\Pi = \Pi'$ and $\xi \cap V \neq \emptyset$.
\item $\psi$ is $\forall (\pi \rightarrow \forall(\pi' \rightarrow \neg T))$, where $\pi \neq \pi'$: $\cC \models \psi$ if and only if
for all $\langle \xi, \Pi \rangle \in \Omega$ such that $\pi \in \xi$, $\pi' \not \in \xi \cup \Pi$.
\item $\psi$ is $\forall (\pi \rightarrow \forall(\pi \rightarrow (= \vee\; T)))$: $\cC \models \psi$ if and only if
there is at most one $\langle \xi, \Pi \rangle \in \Omega$ such that $\pi \in \xi$, and, 
if such a $\langle \xi, \Pi \rangle$ exists, then $\xi \cap V \neq \emptyset$. 
\item $\psi$ is $\forall (\pi \rightarrow \forall(\pi \rightarrow (= \vee\; \neg T)))$: $\cC \models \psi$ if and only if
for all $\langle \xi, \Pi \rangle \in \Omega$, $\pi \not \in \xi \cap \Pi$, and $\xi(\pi) \leq 1$.
\item $\psi$ is $\forall \mu$: $\cC \models \psi$ if and only if, for all
$\langle \xi, \Pi \rangle \in \Omega$ and $\pi \in \xi$, $\models \pi \rightarrow \mu$.
\item $\psi$ is $\exists \mu$: $\cC \models \psi$ if and only if there exist
$\langle \xi, \Pi \rangle \in \Omega$ and $\pi \in \xi$ such that $\models \pi \rightarrow \mu$.
\end{enumerate}

Finally, we show that satisfaction of formulas by certificates corresponds to satisfaction of formulas by structures in the sense captured
by the following two lemmas.
\begin{lemma}
Let $\psi$ be a basic formula, and suppose $\fA \models \psi$ for some quadratic structure $\fA$. Then $\cC(\fA) \models \psi$.
\label{lma:correctnessEq}
\end{lemma}
\begin{proof}
Write $\cC(\fA) = \langle \Omega, \ll, V \rangle$. We consider the forms of $\psi$ in turn.
\begin{enumerate}[1.]
\item $\psi$ is $\forall(\pi \rightarrow \exists(\mu \wedge T\; \wedge \neq))$:  Suppose $\fA \models \psi$ and
$(\xi,\Pi) \in \Omega$ with $\pi \in \xi$. Let $a \in A$ be such that $\cstp^\fA[a]= (\xi,\Pi)$ and $\ftp^\fA[a]= \pi$.
Pick $b \in A \setminus \set{a}$ such that $\fA \models \mu[b]$ and $\fA \models T[a,b]$, and let
$\ftp^\fA[b] = \pi'$. Thus, $\models \pi' \rightarrow \mu$. (i) If $a$ and $b$ are in the
same clique of $\fA$ and $\pi = \pi'$, then $\models \pi \rightarrow \mu$, and $\xi(\pi)= 2$. 
(ii) If $a$ and $b$ are in the same clique, but $\pi' \neq \pi$, then $\pi' \in \xi$.
(iii) If $a$ and $b$ are not in the same clique, then $\pi \in \Pi$.  
\item $\psi$ is $\forall(\pi \rightarrow \exists(\mu \wedge \neg T\; \wedge \neq))$: 
Suppose $\fA \models \psi$ and
$(\xi,\Pi) \in \Omega$ with $\pi \in \xi$. Let $a \in A$ be such that $\cstp^\fA[a]= (\xi,\Pi)$ and $\ftp^\fA[a]= \pi$.
Pick $b \in A \setminus \set{a}$ such that $\fA \models \mu[b]$ and $\fA \not \models T[a,b]$, and let
$\cstp^\fA[b]= (\xi',\Pi')$, and $\ftp^\fA[b] = \pi'$. (i) Thus, $\models \pi' \rightarrow \mu$. (ii) Suppose, for contradiction, 
that there exist $\pi'' \in \Pi$ and $\pi''' \in \xi'$ such that $\pi'' \ll \pi'''$. Then there exist $b'', b''' \in A$ such that
$\fA \models T[a,b'']$, $\fA \models T[b'',b''']$, with $b'''$ in the same clique as $b$, contradicting the assumption that 
$\fA \not \models T[a,b]$. (iii) Suppose, for contradiction, 
that $\pi'' \in \xi' \cap \Pi \cap V$. Then there exists $b'' \in A$ with $\ftp^\fA[b'']= \pi''$, realized in just one
clique (namely, the clique of $b$) and an element $b'''$ with $\ftp^\fA[b''']= \pi''$ and $\fA \models T[a, b''']$. This
contradicts the supposition that $\fA \not \models T[a,b]$. (iv) Suppose, for contradiction, that $(\xi,\Pi)= (\xi',\Pi')$ and
$\pi'' \in \xi \cap  V$.
Then the cliques of both $a$ and $b$ contain
elements of fluted 1-type $\pi''$, with such elements realized in just one
clique. Thus $a$ and $b$ are in the same clique, which
contradicts the supposition that $\fA \not \models T[a,b]$.
\item $\psi$ is $\forall (\pi \rightarrow \forall (\pi' \rightarrow T))$, where $\pi \neq \pi'$: 
Suppose $\fA \models \psi$. (i) If $\pi$ and $\pi'$ are not both realized in $\fA$, then they
do not both occur in $\cC$.
If $\pi$ and $\pi'$ are both realized in $\fA$, and 
$\fA \not \models \forall (\pi' \rightarrow \forall (\pi \rightarrow T))$, then 
$\pi \ll \pi'$. (iii) Otherwise, $\pi$ and $\pi'$ are realized in $\fA$, but there is a  
clique, say $u$, containing all these realizing elements. 
Hence, if $(\xi, \Pi), (\xi', \Pi') \in \Omega$ with $\pi \in \xi$ and $\pi' \in \xi'$, 
then $(\xi, \Pi)$= $(\xi', \Pi')$, and $\pi \in V$, whence $\xi \cap V \neq \emptyset$. 
\item $\psi$ is $\forall (\pi \rightarrow \forall(\pi' \rightarrow \neg T))$, where $\pi \neq \pi'$: 
Suppose $\fA \models \psi$ and $(\xi, \Pi) \in \Omega$ with $\pi \in \xi$. Then there exist $a \in A$ such that
$\cstp^\fA[a]= (\xi, \Pi)$. By the definition of $\cstp^\fA[a]$, $\pi' \not \in \xi \cup \Pi$.
\item $\psi$ is $\forall (\pi \rightarrow \forall(\pi \rightarrow (= \vee\; T)))$: 
Suppose $\fA \models \psi$. Then all elements $a \in A$ such that $\ftp^\fA[a]= \pi$ lie
in a single clique, so let their common clique-super-type be $(\xi, \Pi)$. Thus, $(\xi, \Pi)$
is the only element of $\Omega$ such that $\pi \in \xi$; moreover, if this element exists,
we have $\pi \in V$, and hence $\xi \cap V \neq \emptyset$.
\item $\psi$ is $\forall (\pi \rightarrow \forall(\pi \rightarrow (= \vee\; \neg T)))$:
Suppose $\fA \models \psi$ and $(\xi, \Pi) \in \Omega$ with $\pi \in \xi$. Let $a \in A$ be such 
that $\cstp^\fA[a]= (\xi,\Pi)$ and $\ftp^\fA[a]= \pi$, and let $u$ be the clique of $a$ in $\fA$. Since $\fA \models \psi$, there is certainly
no element $b \in A \setminus u$ such that $\ftp^\fA[b]= \pi$ and $\fA \models T[b,a]$, whence
$\pi \not \in \Pi$. One the other hand, there is no element $b \in u \setminus \set{a}$ such that $\ftp^\fA[b]= \pi$, whence $\xi(\pi)= 1$.
\end{enumerate}
The cases $\forall \mu$ and $\exists \mu$ are routine.
\end{proof}

\begin{lemma}
Let $\psi$ be a basic formula, and suppose $\cC \models \psi$ for some certificate $\cC$. Then $\fA(\cC) \models \psi$.
\label{lma:completenessEq}
\end{lemma}
\begin{proof}
Write $\cC = \langle \Omega, \ll, V \rangle$ and $\fA = \fA(\cC)$. 
We consider the forms of $\psi$ in turn.
\begin{enumerate}[1.]
	
\item $\psi$ is $\forall(\pi \rightarrow \exists(\mu \wedge T\; \wedge \neq))$: 
Suppose $\cC \models \psi$ and $a \in A$ with $\ftp^\fA[a]= \pi$. We may write
$a= a^p_{\pi,\xi,\Pi,i}$, for $(\xi, \Pi) \in \Omega$ with $\pi \in \xi$. 
We must show that there exists $b \in A \setminus \set{a}$ such that $\fA \models \mu[b]$ and $\fA \models T[a,b]$.
(i) If $\models \pi \rightarrow \mu$ and $\xi(\pi) =2$, then, by construction of $\fA$, there
exists $b= a^{p'}_{\pi,\xi,\Pi,i}$ with $p' \neq p$. Thus, $\ftp^\fA[b]= \pi$ and $t_0(a,b)$, whence $\fA \models T[a,b]$.
(ii) If there exists $\pi' \in \xi$ such that $\pi' \neq \pi$ and $\models \pi' \rightarrow \mu$,  there
exists $b= a^{p}_{\pi',\xi,\Pi,i}$. Thus, $\ftp^\fA[b]= \pi'$ and $t_0(a,b)$, whence $\fA \models T[a,b]$.
(iii) If there exists $\pi' \in \Pi$ such that $\models \pi' \rightarrow \mu$, then, by ({C1}), choose 
$(\xi', \Pi') \in \Omega$ with $\pi' \in \xi'$, $\xi' \cup \Pi' \subseteq \Pi$
and $\xi \cap \Pi' \cap V = \emptyset$.
Suppose on the one hand that $\xi' \cap V = \emptyset$. Then we may let $b= a^{+}_{\pi',\xi,\Pi,i+2}$.
Certainly, $\ftp^\fA[b]= \pi'$. It suffices to prove that $t_1(a,b)$, whence $\fA \models T[a,b]$. 
We consider conditions (a)--(c) in the definition of $t_1$. (a) 
We have already established that $\xi' \cup \Pi' \subseteq \Pi$.
(b) Trivially, $i+2 \geq i+2$. (c) {\em A fortiori}, $\xi' \cap V \cap \Pi= \emptyset$. 
Suppose on the other hand that $\xi' \cap V \neq \emptyset$.  Then we may let $b= a^{+}_{\pi',\xi,\Pi,0}$. Since 
$\xi \cap \Pi' \cap V = \emptyset$, we have $\xi \neq \xi'$, so that $b \neq a$. Again, consider 
conditions (b) and (c) in the definition of $t_1$. For (b), we are supposing anyway that 
 $\xi' \cap V \neq \emptyset$, and for (c), we have already established that $\xi \cap \Pi' \cap V = \emptyset$.
Thus, in all cases, we have $\fA \models \mu[b]$ and $\fA \models T[a,b]$, as required.
\item $\psi$ is $\forall(\pi \rightarrow \exists(\mu \wedge \neg T\; \wedge \neq))$: 
Suppose $\cC \models \psi$ and $a \in A$ with $\ftp^\fA[a]= \pi$. We may write
$a= a^p_{\pi,\xi,\Pi,i}$, for $(\xi, \Pi) \in \Omega$ with $\pi \in \xi$. 
Then we may select $(\xi',\Pi') \in \Omega$ with $\pi' \in \xi'$ such that:
(i) $\models \pi' \rightarrow \mu$; 
(ii) there exists no $\pi'' \in \Pi$ and $\pi''' \in \xi'$ such that $\pi'' \ll  \pi'''$; (iii)
$\xi' \cap \Pi \cap V = \emptyset$; and (iv) $(\xi,\Pi) = (\xi', \Pi') \Rightarrow \xi \cap V = \emptyset$. 
Suppose on the one hand that $(\xi',\Pi') \neq (\xi, \Pi)$. Let $b= a^{+}_{\pi',\xi',\Pi',0}$, so that, by
construction of $\fA$, $\ftp^\fA[b]= \pi'$. We must show that $a \neq b$ and $\fA \not \models T[a,b]$.
Let $u$ be the cell containing $a$ and $u'$ the cell containing $b$. Since
$(\xi',\Pi') \neq (\xi, \Pi)$, we have $u \neq u'$, whence, certainly $a \neq b$. So suppose
for contradiction that there is a sequence of $(t_1 \cup t_2)$-links from $u$ to $u'$. Let $u'' \in A_{\xi'', \Pi''}$, say, be
the penultimate element of this sequence. Certainly, there is no $t_2$-link from $u''$ to $u'$, since this would
require $\pi'' \in \xi''$ and $\pi''' \in \xi'$ with $\pi'' \ll \pi'''$. But 
by Lemma~\ref{lma:accumulate}, we would then have $\pi'' \in \Pi$, which is ruled out by (ii). On the other hand,
if there were a $t_1$-link from $u''$ to $u'$, then we would have $\xi' \cap V \neq \emptyset$, and again
by Lemma~\ref{lma:accumulate}, $\xi' \subseteq \Pi$, whence $\xi' \cap V \cap \Pi \neq \emptyset$, which is ruled out by (iii).
Suppose on the other hand that $(\xi',\Pi') = (\xi, \Pi)$. But then (iv) implies $\xi \cap V = \emptyset$, so that we may select
$b= a^{+}_{\pi',\xi,\Pi,j}$, where $j= 1$ if $i=0$ and $j= 0$ otherwise. Again, 
let $u$ be the cell containing $a$ and $u'$ the cell containing $b$. Thus $u \neq u'$, whence certainly $a \neq b$.
Moreover, $\fA \models \mu[b]$. Again, it remains to show that $\fA \not \models T[a,b]$. Suppose there is a chain
$u= u_0, \dots, u_k= u'$ of $(t_1, \cup t_2)$-links. By construction of $t_1$, we must have $t_2(u_{k-1},u_k)$, since
$j \leq 1$. Then there exists $\pi'' \in \xi_{k-1}$ and $\pi''' \in \xi_k$ such that $\pi'' \ll \pi'''$. Then, certainly,
$k >1$ since, otherwise, $\xi_0 = \xi_k = \xi$ contains both $\pi''$ and $\pi'''$ with $\pi'' \ll \pi'''$
and $\xi \cap V = \emptyset$, which contravenes ({C6}).  But if $k >1$,  then $\pi \in \Pi$ by Lemma~\ref{lma:accumulate},
which contravenes ({C5}). Thus, we have shown that $\fA \not \models T[a,b]$ as required.
\item $\psi$ is $\forall (\pi \rightarrow \forall (\pi' \rightarrow T))$, where $\pi \neq \pi'$: 
Suppose $\cC \models \psi$, and that $a, a' \in A$ with $\ftp^\fA[a]= \pi$ and $\ftp^\fA[a']= \pi'$.
Write $a= a^p_{\pi,\xi,\Pi,i}$ and $a'= a^{p'}_{\pi',\xi',\Pi',j}$.
We must show that $\fA \not \models T[a,a']$. We consider the three possibilities in the definition of $\cC \models \psi$.
(i) By construction of $\fA$, $\pi$ and $\pi'$ both occur in $\cC$, so the first possibility does not arise.
(ii) Suppose that $\pi \ll \pi'$. If $a$ and $a'$ are in different cells, then 
then we immediately have $t_2(a,a')$.
If, on the other hand, $a$ and $a'$ are in the same cell, then since $\pi \neq \pi'$, by ({C4}), 
$\neg \hat{T} \not \in \bigcup \xi$, whence $t_0(a,a')$. 
(iii) Suppose that there is a single clique-super-type $(\xi,\Pi) \in \Omega$ such that $\xi$ contains
either $\pi$ or $\pi'$ and that $\xi \cap V \neq \emptyset$. By the construction of $\fA$, $a$ and $a'$
belong to the same cell $A_{\xi,\Pi,0}$, and again by ({C4}), 
$\neg \hat{T} \not \in \bigcup \xi$, whence $t_0(a,a')$. In all cases, then, $\fA \models T[a,a']$, as required. 
\item $\psi$ is $\forall (\pi \rightarrow \forall(\pi' \rightarrow \neg T))$, where $\pi \neq \pi'$: 
Suppose $\cC \models \psi$, and that $a, a' \in A$ with $\ftp^\fA[a]= \pi$ and $\ftp^\fA[a']= \pi'$.
Write $a= a^p_{\pi,\xi,\Pi,i}$ and $a'= a^{p'}_{\pi',\xi',\Pi',j}$. From the definition of $\cC \models \psi$,
we have $\pi' \not \in \xi \cup \Pi$, whence $\xi \neq \xi'$. Thus, $a$ and $a'$ occupy different cells, say, $u$ and
$u'$, respectively. By Lemma~\ref{lma:accumulate}, there is no chain
$u= u_0, \dots, u_k= u'$ of $(t_1 \cup t_2)$-links. Therefore, $\fA \not \models T[a,a']$, as required.
\item $\psi$ is $\forall (\pi \rightarrow \forall(\pi \rightarrow (= \vee\; T)))$: 
Suppose $\cC \models \psi$, and that $a, a' \in A$ with $\ftp^\fA[a]= \ftp^\fA[a']= \pi$ and $a \neq a'$.
From the definition of $\cC \models \psi$ and the 
construction of $\fA$, $a$ and $a'$ belong to the same set $A_{\xi,\Pi}$ and, moreover,
$\xi \cap V \neq \emptyset$. It follows that 
$a$ and $a'$ belong to the cell $A_{\pi,\xi,\Pi,0}$.  
Since $a \neq a'$, by the construction of $A$, $\xi(\pi) = 2$, whence by ({C4}), 
$\xi$ is not a soliton clique-type,
whence $t_0(a,a')$. Thus $\fA \models T[a,a']$, as required.
\item $\psi$ is $\forall (\pi \rightarrow \forall(\pi \rightarrow (= \vee\; \neg T)))$: 
Suppose $\cC \models \psi$. It follows immediately by construction of $\fA$ that
no set $A_{\pi,\xi,\Pi,i}$ can have cardinality greater than 1.
Now suppose $a, a' \in A$ with $\ftp^\fA[a]= \ftp^\fA[a']= \pi$ and $a \neq a'$. Thus,
$a$ and $a'$ are not in the same cell, and hence by Lemma~\ref{lma:accumulate}, 
$\fA \models T[a,a']$ implies $\xi' \subseteq \Pi$, whence $\pi \in \Pi$, contradicting
the definition of  $\cC \models \psi$. Thus, $\fA \not \models T[a,a']$, as required.
\end{enumerate}
The cases $\forall \mu$ and $\exists \mu$ are routine.
\end{proof}

\begin{lemma}
There exists a non-deterministic procedure which, when given
a set $\Phi$ of basic $\FLotranstEqMinus$-formulas over a signature $\Sigma$, will terminate
in time bounded by $g(2^{2^{g(|\Sigma|)}} + \sizeof{\Phi})$, for some fixed polynomial $g$, and which
has an accepting run if and only if $\Phi$ is satisfiable.
\label{lma:FLotransmEqMinus}
\end{lemma}
\begin{proof}
Let $\Phi$ be given. By Lemma~\ref{lma:quadratic}, the following are equivalent: $\Phi$ is satisfiable; 
$\Phi^*$ is satisfied in a quadratic structure; $\Phi^*$ is satisfiable.
Observe that $\Phi^*$ (a set of basic formulas over some signature $\Sigma^* \supseteq \Sigma$) can be computed in time bounded by a polynomial function of $\Phi$.
By Lemma~\ref{lma:correctnessEq}, if $\Phi^*$ is satisfiable over a quadratic structure 
then there exists a certificate $\cC$, interpreting $\Sigma^*$, such that $\cC \models \Phi^*$. By Lemma~\ref{lma:completenessEq}, 
if there exists a certificate $\cC$ over $\Sigma^*$, 
such that $\cC \models \Phi^*$, then $\Phi^*$ is satisfiable. 
Evidently $\sizeof{\cC}$ is bounded by a doubly exponential function of $|\Sigma^*|$, and the condition
$\cC \models \Phi^*$ may be checked in time bounded by a polynomial function of $\sizeof{\Phi^*} + \sizeof{\cC}$.
\end{proof}

%

\subsection{The logic $\FlEqTransVars{1}{2}$}
\label{sec:FLotranstEq}
The next step is to allow arbitrary (non-distinguished) binary predicates; that is, we consider the logic
$\FlEqTransVars{1}{2}$, the 2-variable fluted fragment with equality and a single, distinguished, transitive relation $T$.

In the context of a structure interpreting a relational signature, a {\em king} is an element whose fluted 1-type is not realized by any other element 
in that structure. The fluted 1-types of kings are called {\em royal}.
We make use of the well-known fact that, in two-variable logic, parts of structures may be duplicated as long as they 
contain no king.\nb{I: Removed refs to KMP-HT and Kiero+Otto12; I'm not sure that these do mention this lemma}
We use the formulation appearing in~\cite[Lemma~4.1]{purdyTrans:ph18}. The proof given there 
concerns two-variable first-order logic with a single distinguished predicate interpreted as a {\em partial order}; however, the proof for the (present) case in which it is interpreted as a {\em transitive relation} is identical, and we need not repeat it here.
\begin{lemma}
Let $\fA_1$ be a structure over domain $A_1$,
$A_0$ the set of kings of $\fA_1$, $\fA_0$ the restriction of $\fA$ to $A_0$, and $B_1 = A_1 \setminus A_0$. 
There exists a family of sets $\set{B_i}_{i \geq 2}$, pairwise disjoint and disjoint from $A_1$, 
a family of bijections
$\set{f_i}_{i \geq 1}$, where $f_i: B_i \rightarrow B_1$,
and a sequence of structures $\set{\fA_i}_{i \geq 2}$, where
$\fA_i$ has domain $A_i = A_0 \cup B_1 \cup B_2 \cup \cdots \cup B_i$, such that, for all $i \geq 1$:
\begin{enumerate}[\textup{(}i\textup{)}]
\item $\fA_{i-1} \subseteq \fA_i$, and all 2-types realized in $\fA_i$ are realized in $\fA_{1}$;
\item for all $a \in B_i$ and all $b \in A_1$, if $f_i(a) \neq b$, then
$\ftp^{\fA_{i}}[a,b] = \ftp^{\fA_{1}}[f_i(a),b]$;
\item for all $a \in B_i$, all $j$ \textup{(}$1 \leq j \leq i$\textup{)} and all $b \in B_j$, if $f_i(a) \neq f_j(b)$, then 
$\ftp^{\fA_{i}}[a,b] = \ftp^{\fA_1}[f_i(a),f_j(b)]$;
\item $T^{\fA_i}$ is a transitive relation.
\end{enumerate}
\label{lma:theGreatInflate}
\end{lemma}
Intuitively, the sets $B_2, \dots, B_i$ are copies of $B_1$. This 
copying process may be continued indefinitely (or even infinitely); we require only finitely many
iterations in this paper.

Recall the notion of normal form for $\FlEqTransVars{1}{m}$ given in~\eqref{eq:flMnf}. For $m=2$, we obtain the special case
\begin{equation}
{\bigwedge_{i \in S}}\forall (\mu_i \rightarrow \exists (\kappa_i \wedge \Gamma_i)) \wedge
{\bigwedge_{j \in T}} \forall (\nu_j  \rightarrow \forall \Delta_j) \wedge \forall \forall  \Omega,
\label{eq:flTWOnf}
\end{equation}
where $S$ and $T$ are finite sets of indices, such that, for $i \in S$ and $j \in T$,
$\mu_i$ and $\nu_j$ are quantifier-free formulas of arity 1,
$\kappa_i$ is a control formula, 
and $\Gamma_i$, $\Delta_j$, and
$\Omega$ are sets of fluted 2-clauses.\nb{I: added `fluted'}
Our strategy will be to reduce the satisfiability problem for formulas of the form~\eqref{eq:flTWOnf}
to that of sets of basic $\FLotranstEqMinus$-formulas.
We begin by introducing a variant of the  normal form
for $\FlEqTransVars{1}{2}$. A formula of this logic is in {\em spread normal form} if it has the shape
\begin{multline}
\bigwedge_{h \in R} \exists \lambda_h \wedge 
{\bigwedge_{i \in S}}\forall (\mu_i  \rightarrow \exists (o_i \wedge \kappa_i \wedge \Gamma_i)) \wedge\\
{\bigwedge_{j \in T}} \forall (\nu_j  \rightarrow \forall \Delta_j) \wedge
\forall \forall \Omega  \wedge
\bigwedge_{i,i' \in S}^{i \neq i'} \forall (o_i \rightarrow \neg o_{i'}),
\label{eq:flTWOnfSpread}
\end{multline}
where $R$ is an index set, the $\lambda_h$ ($h \in R$) are quantifier-free, unary formulas,
the $o_i$ ($i \in S$) are unary predicates, and
$S$, $T$, $\Omega$, $\mu_i$, $\nu_j$, $\kappa_i$, $\Gamma_i$, $\Delta_j$ are as before.
The essential change 
here is the insertion of the atoms $o_i$ into the conjuncts $\forall (\mu_i  \rightarrow \exists (\kappa_i \wedge \Gamma_i))$
of~\eqref{eq:flTWOnf} together with the addition of the conjuncts $\forall (o_i \rightarrow \neg o_{i'})$ for distinct
indices $i$ and $i'$. The point is that, if an object satisfies $\mu_i$ for several indices $i$, the
corresponding witnesses of the formula $\exists (o_i \wedge \kappa_i \wedge \Gamma_i)$ for that element are all distinct. As we might say, the 
witness requirements are `spread' over different objects. The other change in~\eqref{eq:flTWOnfSpread}
is the addition of the conjuncts $\exists \lambda_h$. These are required if we are going to convert normal-form
$\FlEqTransVars{1}{2}$-sentences into spread normal form without incurring an unacceptable inflation in the size of the
signature, as promised by the next lemma.
\begin{lemma}
Let $\phi$ be a normal-form $\FlEqTransVars{1}{2}$-formula and
$\Pi = \set{\pi_1, \dots, \pi_L}$ a set of fluted 1-types over the signature of $\phi$.  We can compute, in time bounded by 
an exponential function of $\sizeof{\phi}$, a formula $\psi$ in spread normal form,
such that: \textup{(}i\textup{)} the signature of $\psi$ is bounded in size by a polynomial function of $\sizeof{\phi}$;
\textup{(}ii\textup{)} $\models \psi \rightarrow \phi$; and \textup{(}iii\textup{)} if $\phi$ 
has a \textup{(}finite\textup{)} model in which $\Pi$ is the set of royal fluted 1-types, then $\psi$ has a \textup{(}finite\textup{)} model.
\label{lma:FLotranstEqToFLotranstEqSpread}
\end{lemma}
\begin{proof}
Let $\phi$ have the shape
\begin{equation*}
{\bigwedge_{i \in S} }\forall (\mu_i \rightarrow \exists (\kappa_i \wedge \Gamma_i)) \wedge
{\bigwedge_{j \in T}} \forall (\nu_j  \rightarrow \forall \Delta_j) \wedge \forall\forall \Omega.
\end{equation*}
Let 
$o_i$ be a fresh unary predicate for each $i \in S$, let $k= \lceil \log ((L+1) \cdot |S|)\rceil$, 
and let $w_0, \dots, w_{k-1}$
be a  collection of fresh unary predicates. Observe that $k$ is polynomially bounded as a function of
$\sizeof{\phi}$.
For each $i \in S$ and each $\ell$ ($0 \leq \ell \leq L$), 
we take 
$\bar{w}\langle i, \ell \rangle$ to be a distinct formula of the form 
$\pm w_0 \wedge \cdots \wedge \pm w_{k}$. As a guide to intuition, 
read the (1-place) formula $\bar{w}\langle i, 0 \rangle$ as 
characterizing those elements $a$ such that there exists a non-royal $b$ with $a,b$ satisfying
$\Gamma_i$, read the formulas
$\bar{w}\langle i, \ell \rangle$ ($1 \leq \ell \leq L$) as 
characterizing those elements $a$ such that $a, b_\ell$ satisfies
$\Gamma_i$, where $b_\ell$ is the king with fluted 1-type $\pi_\ell$, and finally
take the predicates $o_i$ to pick out pairwise disjoint collections of non-royal elements (we will say more
presently about how these sets are chosen).

We define $\psi$ to be the conjunction of the following formulas.
\begin{align}
&\bigwedge_{\ell=1}^{L} \exists \pi_\ell
\label{eq:flTWO0}\\
& \bigwedge_{i \in S} 
        \forall (\bar{w}\langle i, 0\rangle \rightarrow 
            \exists(o_i \wedge \kappa_i \wedge \Gamma_i))
\label{eq:flTWO1}\\
& \bigwedge_{i \in S} \bigwedge_{\ell= 1}^{L}
        \forall (\bar{w}\langle i, \ell\rangle \rightarrow 
            \forall(\pi_\ell \rightarrow (\kappa_i \wedge \Gamma_i)))
\label{eq:flTWO4}\\
& \bigwedge_{i\in S} \forall \big(\mu_i \rightarrow \bigvee_{\ell = 0}^L \bar{w}\langle i,\ell \rangle\big)
\label{eq:flTWO6}\\
& \bigwedge_{j \in T} 
        \forall\left(\nu_j \rightarrow \forall \Delta_j \right)
\label{eq:flTWO3}\\
& \bigwedge_{i, i' \in S}^{i \neq i'} \forall \neg (o_i \wedge o_i') \wedge \forall \forall \Omega.
\label{eq:flTWO7}
\end{align}
The conjuncts~\eqref{eq:flTWO0}--\eqref{eq:flTWO6} clearly entail
$\bigwedge_{i \in S} \forall (\mu_i \rightarrow \exists (\kappa_i \wedge \Gamma_i))$. Thus, $\psi \rightarrow \phi$. Suppose,
on the other hand, $\fA_1 \models \phi$ with the set of royal types in $\fA_1$ equal to $\Pi$.
We may assume without loss of generality that $S = \set{1, \dots, s}$.
Let the set of kings in $\fA_1$ be $A_0$, and apply the construction  of
Lemma~\ref{lma:theGreatInflate} to obtain the structures
$\fA_2, \dots, \fA_s$.
Let
$\fB = \fA_s$, a model of $\phi$ with domain $B= A_0 \cup B_1 \cup \cdots \cup B_s$. 
For all $a \in B$, if there exists a non-royal element $b$ such that $\fB \models \Gamma_i[a,b]$, then
there exists such a $b$ in each of the sets $B_1, \dotsm B_s$. Now expand $\fB$ to a model $\fB^+$
by setting $o_i^{\fB^+} = B_i$, and interpreting the predicates $w_0, \dots, w_{k-1}$ so that 
the formulas
$\bar{w}\langle i, \ell \rangle$ ($0 \leq \ell \leq L$) have the interpretations suggested above. 
It is then simple to check that $\fB^+ \models \psi$. We  remark finally that the consequents $(\kappa_i \wedge \Gamma_i)$ occurring in~\eqref{eq:flTWO4} can of course be written as a set of fluted clauses since $=$, $\neq$, $T$ and $\neg T$
are fluted literals. Hence $\psi$ is in spread normal form.
\end{proof}

\begin{lemma}
Let
$\phi$ be a spread normal-form $\FlEqTransVars{1}{2}$-formula.
We can compute, in time bounded by 
an exponential function of $\sizeof{\phi}$, a set $\Phi$ of basic formulas, 
such that: \textup{(}i\textup{)} the signature of $\Phi$ consists of the unary predicates occurring in $\phi$ together with the
distinguished predicate $T$; \textup{(}ii\textup{)} 
$\models \phi \rightarrow \bigwedge \Phi$; and \textup{(}iii\textup{)} any model of 
$\Phi$ can be expanded to a model of $\phi$.
\label{lma:FLotranstEqToFLotranstEqMinus}
\end{lemma}
\begin{proof}
Let $\phi$ be given, having the shape
\begin{multline*}
\bigwedge_{h \in R} \exists \lambda_h \wedge 
{\bigwedge_{i \in S}}\forall (\mu_i  \rightarrow \exists (o_i \wedge \kappa_i \wedge \Gamma_i)) \wedge\\
{\bigwedge_{j \in T}} \forall (\nu_j  \rightarrow \forall \Delta_j) \wedge
\forall \forall \Omega  \wedge
\bigwedge_{i,i' \in S}^{i \neq i'} \forall (o_i \rightarrow \neg o_{i'}),
\end{multline*}
and recall from Sec.~\ref{sec:ftc} that, for any fluted clause set $\Gamma$,
$\Gamma^\circ$ denotes the result of saturating under mo-resolution,
and then removing any clauses involving ordinary predicates of maximal arity (in this case 2). 
Noting that each $o_i$ is a (1-literal) clause, and regarding each control formula $\kappa_i$ as a
pair of (1-literal) clauses,
let $\psi$ be the corresponding conjunction
\begin{align}
& \bigwedge_{h \in R} \exists \lambda_h
\label{eq:flT1}\\
& \bigwedge_{i \in S} \bigwedge_{J \subseteq T} \forall \big(\big(\mu_i \wedge \bigwedge_{j \in J} \nu_j\big) \rightarrow \exists\big(\kappa_i \cup \set{o_i} \cup \Gamma_i \cup \bigcup_{j \in J} \Delta_j \cup \Omega \big)^{\hspace{-0.5mm} \circ}\big)
\label{eq:flT2}\\
& \bigwedge_{J \subseteq T} 
        \forall \big(\bigwedge_{j \in J} q_j \rightarrow \forall\big(\bigcup_{j \in J} \Delta_j \cup \Omega \big)^{\hspace{-0.5mm} \circ}\big)
\label{eq:flT3}\\
& \bigwedge_{i, i' \in S}^{i \neq i'} \forall \neg (o_i \wedge o_i').
\label{eq:flT4}
\end{align}
It is immediate by the validity of mo-resolution that $\models \phi \rightarrow \psi$. We claim that any model of $\psi$ may be expanded to a model of $\phi$. For suppose
$\fB \models \psi$; we expand to a structure $\fB^+$ interpreting a signature $\Sigma^+$ which additionally features the
non-distinguished binary predicates occurring in $\phi$ as follows. Fix any $a \in B$, and let 
$J = \set{j \in J \mid \fB \models \nu_j[a]}$. For each $i \in S$, if $\fB \models \mu_i[a]$, by~\eqref{eq:flT2},
there exists $b_i \in B$ such that $\fB \models \big(\kappa_i \cup \set{o_i} \cup \Gamma_i \cup \bigcup_{j \in J} \Delta_j \cup \Omega \big)^{\hspace{-0.5mm} \circ}\big)[a,b_i]$,
and by~\eqref{eq:flT4}, these $b_i$ are all distinct. For each $i \in S$,
let $\tau_i = \ftp^\fB[a,b_i]$. Obviously, $\models \tau_i \rightarrow \kappa_i$. By Lemma~\ref{lma:resolution}, there exists a fluted type $\tau^+_i$ in the signature $\Sigma^+$, 
such that $\tau^+_i \supseteq \tau_i$ and $\tau^+_i \models \set{o_i} \cup \Gamma_i \cup \bigcup_{j \in J} \Delta_j  \cup \Omega$. Therefore,
we may assign the pair $a,b_i$ to the extensions of the predicates in $\Sigma^+ \setminus \Sigma$ in such a way that its fluted 2-type
is $\tau^+_i$. Since the various $b_i$ are distinct, no clashes arise. Performing this operation for every element $a$, we have a partially
defined structure $\fB^+$ such that, however it is completed, $\fB^+ \models {\bigwedge_{i \in S} }\forall (\mu_i \rightarrow \exists (\kappa_i \wedge \Gamma_i))$, and,
moreover, the conjuncts $\forall (\nu_j  \rightarrow \forall \Delta_j)$ (for $j \in T$) and $\forall \forall \Omega$ have not been violated.
To  complete the definition of $\fB^+$, consider any ordered pair $a,b$ for which the predicates in $\Sigma^+ \setminus \Sigma$ have not been assigned.
Let $J = \set{j \in J \mid \fB \models \nu_j[a]}$, so that, by~\eqref{eq:flT3}, $\fB \models \left(\bigcup_{j \in J} \Delta_j  \cup \Omega\right)^\circ$.
Letting $\tau = \ftp^\fB[a,b]$, by Lemma~\ref{lma:resolution} there exists a fluted type $\tau^+$ in the signature $\Sigma^+$, 
such that $\tau^+ \supseteq \tau$ and $\tau^+ \models \bigcup_{j \in J} \Delta_j  \cup \Omega$. Therefore,
we may assign the pair $a,b$ to the extensions of the predicates in $\Sigma^+ \setminus \Sigma$ in such a way that its fluted 2-type
is $\tau^+$. At the end of this process, we have $\fB^+ \models \bigwedge_{j \in T} \forall (\nu_j  \rightarrow \forall \Delta_j) \wedge \forall\forall\Omega$.
Since $\fB^+$ is an expansion of $\fB$, we certainly have $\fB^+ \models \bigwedge_{h \in R} \exists \lambda_h$, and
also $\fB^+ \models \bigwedge_{i, i' \in S}^{i \neq i'} \forall \neg (o_i \wedge o_i')$ by~\eqref{eq:flT1} and~\eqref{eq:flT4}.  

The desired set of basic formulas $\Phi$ can now be obtained from $\psi$ by simple manipulation.
The conjuncts in~\eqref{eq:flT1} and~\eqref{eq:flT4} are already basic. 
The conjuncts in~\eqref{eq:flT2} are all of the
forms $\forall (\eta \rightarrow \exists (\neq \wedge \pm T \wedge \theta ))$ or
$\forall (\eta \rightarrow \exists (= \wedge \pm T \wedge \theta))$. 
Clearly, we may eliminate occurrences of $T$ and $=$ from $\theta$, so assume that this has been done, and
$\theta$ is a quantifier-free formula of arity 1. In the former case, we replace the conjunct
with a collection of conjuncts $\forall (\pi \rightarrow \exists (\neq \wedge \theta \wedge \pm T))$, where
$\pi$ ranges over all fluted 1-types consistent with $\eta$; such conjuncts are basic, 
of the forms
{(B1)} or {(B2)}. In the latter case, remembering that
$\hat{T}$ is a unary predicate interpreted as the diagonal of $T$, we see that all such formulas are
either trivial or 
logically equivalent to a formula of the form {(B7)} ruling out a certain collection of fluted 1-types, and 
which can be computed in time bounded by an exponential function of $\sizeof{\phi}$. 
The conjuncts in~\eqref{eq:flT3} are all of the
forms $\forall (\eta \rightarrow \forall \theta)$. Conversion to a conjunction of basic formulas of the forms
{(B3)}--{(B7)} in time bounded by an exponential function of $\sizeof{\phi}$ is then completely routine,
using similar considerations. 
\end{proof}

Thus, 
we have the promised upper bound for the problem $\Sat{\FlEqTransVars{1}{2}}$.
\begin{lemma}
	The satisfiability problem for $\FlEqTransVars{1}{2}$ is in \TwoNExpTime. 
	\label{lma:FLotranstEqUpper}
\end{lemma}
\begin{proof}
	Let an $\FlEqTransVars{1}{2}$-sentence $\phi$ be given. By Lemma~\ref{lma:flMnf}, we may assume without loss of 
	generality that $\phi$ is in normal form.  
	Guess a set $\Pi$ of fluted 1-types over the signature of $\phi$ and 
	apply the procedure guaranteed by Lemma~\ref{lma:FLotranstEqToFLotranstEqSpread} to obtain, in time
	bounded by an exponential function of $\sizeof{\phi}$, a spread normal-form formula $\psi$, over a signature bounded by a polynomial function of
	$\sizeof{\phi}$, such that $\models \psi \rightarrow \psi$, and, if $\phi$ has a (finite) model 
	in which the set of royal fluted 1-types is $\Pi$, then $\psi$ has a such a model too. 
    By Lemma~\ref{lma:FLotranstEqToFLotranstEqMinus} we may then 
	obtain, in time
	bounded by an exponential function of $\sizeof{\psi}$, a set $\Phi_\Pi$ of basic $\FLotranstEqMinus$ sentences,
	over the signature consisting of the unary predicates of $\psi$ together with the distinguished predicate $T$, 
	such that $\Phi_\Pi$ is satisfiable over the same domains as $\phi$, assuming that the set of royal 1-types is $\Pi$. 
	Hence it suffices to check the satisfiability of each such $\Phi_\Pi$, non-deterministically, in time  
	bounded by a doubly exponential function of $\sizeof{\phi}$.
	But this we can do by Lemma~\ref{lma:FLotransmEqMinus}.
\end{proof}
\subsection{The logic $\FlEqTransVars{1}{m}$}
\label{sec:FLotransmEq}
Finally, we show how the satisfiability problem for $\FlEqTransVars{1}{m+1}$ can be reduced to the corresponding
problem for  $\FlEqTransVars{1}{m}$, but with exponential blow-up. The following notion will be useful. Let $I$ be a finite set. A {\em cover} of $I$ is a
set $M = \set{C_1, \dots, C_\ell}$ of subsets of $I$ such that $C_1 \cup \cdots \cup C_\ell = I$; the
elements of $M$ will be referred to as {\em cells}. A {\em minimal cover} of $I$ is a
cover $M$ of $I$ such that no proper subset of $M$ is a cover of $I$. Since no minimal cover of $I$ can have
more than $I$ cells, we have $|MC(I)| \leq 2^{|I|^2}$. Denote by $MC(I)$ the set of minimal covers of $I$.
If $I$ is a set of integers, and $M$ is a minimal cover of $I$, we may assume the cells of $M$ to be enumerated
in some standard way as $C_1, \dots, C_\ell$. 

\begin{lemma}
Let $\phi$ be a normal-form $\FlEqTransVars{1}{m+1}$-formula \textup{(}$m \geq 2$\textup{)}. We can compute, in time bounded by 
an exponential function of $\sizeof{\phi}$, a 
 normal-form $\FlEqTransVars{1}{m}$-formula $\psi$ such that $\phi$ and $\psi$ are satisfiable over the same domains.
\label{lma:FLotransmEqToFLotransmEq}
\end{lemma}
\begin{proof}
Let $\phi$ be as given in \eqref{eq:flMnf}. The control formulas $\kappa_i$ occurring there
involve only binary predicates, and thus will be---so far as this proof is concerned---inert. Indeed, 
since each $\kappa_i$ is a conjunction
of two (1-literal) clauses, we can harmlessly absorb it into the respective clause set $\Gamma_i$.
Thus, we may take $\phi$ to have the shape:
\begin{equation*}
\bigwedge_{i \in S}\forall^m (\mu_i \rightarrow \exists \Gamma_i) \wedge
\bigwedge_{j \in T} \forall^m (\nu_j  \rightarrow \forall \Delta_j)  \wedge \forall^{m+1} \Omega.
\end{equation*}
We may also assume without loss of generality that the indices in the sets $S$ and $T$ are integers.
For each $I$ ($I \subseteq S$)
and each $J$ ($J \subseteq T$), let $p_{I,J}$ and $q_J$ be fresh ($m-1$)-ary predicates.
Further, for each minimal cover $M = \set{C_1, \dots, C_\ell}$ of $I$ (enumerated in the standard way),   
let $p_{I,J,M}$ be a fresh ($m-1$)-ary predicate, and for each $h$ ($1 \leq h \leq \ell$), let 
$p_{I,J,M,h}$ be a fresh $m$-ary predicate. 

Remembering that the $\Gamma_i$, $\Delta_j$, and $\Omega$ occurring in $\phi$ are sets of fluted $(m+1)$-clauses,
let $\psi$ be the conjunction of formulas
\begin{align}
& \bigwedge_{I \subseteq S} \bigwedge_{J \subseteq T} 
        \forall^m\big(\bigwedge_{i \in I} \mu_i \wedge \bigwedge_{j \in J} \nu_j\rightarrow p_{I,J}\big)
\label{eq:flM1}\\
& \bigwedge_{J \subseteq T} 
        \forall^m\big(\bigwedge_{j \in J} \nu_j\rightarrow q_{J}\big)
\label{eq:flM2}\\
& \bigwedge_{I \subseteq S} \bigwedge_{J \subseteq T} 
        \forall^m\big(p_{I,J} \rightarrow \bigvee_{M \in MC(I)} p_{I,J,M}\big)
\label{eq:flM3}\\ 
&\bigwedge_{I \subseteq S} \bigwedge_{J \subseteq T} \bigwedge_{M \in MC(I)\hspace{-2mm}}
      \hspace{-4mm}  \forall^{m-1}  \big(p_{I,J,M} \rightarrow 
\hspace{-1mm} \bigwedge_{h=1}^{|M|} \exists \big(p_{I,J,M,h} \wedge 
           \big(\hspace{-1mm} \bigcup_{i \in C_h} \Gamma_i \cup \bigcup_{j \in J} \Delta_j \cup \Omega \big)^{\hspace{-0.5mm} \circ}\big)\big)
\label{eq:flM4}\\
& \bigwedge_{I \subseteq S} \bigwedge_{J \subseteq T} 
        \forall^{m-1}\big(q_{J} \rightarrow \bigwedge_{h=1}^{|M|} \forall \big( \bigcup_{j \in J} \Delta_j \cup \Omega \big)^{\hspace{-0.5mm} \circ} \big)
\label{eq:flM5}\\
& \bigwedge_{I \subseteq S} \bigwedge_{J \subseteq T} \bigwedge_{M \in MC(I)} \bigwedge_{1 \leq h < h' \leq |M|} \hspace{-4mm}\forall^m \neg (p_{I,J,M,h} \wedge p_{I,J,M,h'}).
\label{eq:flM6}
\end{align}
Modulo re-arrangement of conjuncts, $\psi$ is a normal-form formula of $\FlEqTransVars{1}{m}$. 
It suffices therefore to show that $\phi$ and $\psi$ are satisfiable over the same domains. 

Suppose $\fA \models \phi$. We expand to a model $\fA^+ \models \psi$ as follows. 
For any ($m-1$)-tuple $\bar{a}$ and any $I \subseteq S$ and $J \subseteq T$, if there exists $a \in A$ such that
$\fA \models \mu_i[a,\bar{a}]$ for all $i \in I$  and $\fA \models \nu_j[a,\bar{a}]$ for all $j \in J$,
assign $\bar{a}$ to the extension of $p_{I,J}$, and pick some particular $a$ for which this condition
is satisfied. Since $\fA \models \phi$,
there exists a collection of distinct individuals $b_1, \dots, b_\ell$ and a minimal cover
$M= \set{C_1, \dots, C_\ell}$ of $I$ such that, for all $h$ ($1 \leq h \leq \ell$), 
$\fA \models \left(\bigcup_{i \in C_h} \Gamma_i \cup \bigcup_{j \in J} \Delta_J \cup \Omega \right)[a,\bar{a},b_h]$.
Now assign $\bar{a}$ to the extension of $p_{I,J,M}$, and for all $h$ ($1 \leq h \leq \ell$), 
assign $\bar{a},b_h$ to the extension of $\fA^+ \models p_{I,J,M,h}$.
It follows by the validity of the resolution rule, that, for all $h$ ($1 \leq h \leq \ell$), 
$\fA \models \left(\bigcup_{i \in C_h} \Gamma_i \cup \bigcup_{j \in J} \Delta_J \cup \Omega \right)^\circ[\bar{a},b_h]$,
and, by construction, 
$\fA^+ \models p_{I,J,M,h}[\bar{a},b_h]$.
Carrying out this process for all possible ($m-1$)-tuples $\bar{a}$, we see that 
$\fA^+$ verifies the formulas~\eqref{eq:flM1}, \eqref{eq:flM3} and~\eqref{eq:flM4}. 
Moreover, since the individuals $b_1, \dots, b_\ell$ are by hypothesis distinct, 
no tuple $\bar{a},b_h$ satisfies both $p_{I,J,M,h}$ and $p_{I,J,M,h'}$ for $h' \neq h$,
so that $\fA^+$ verifies the formulas~\eqref{eq:flM6}.
Similarly, 
for any ($m-1$)-tuple $\bar{a}$ and any $J \subseteq T$, if there exists $a \in A$ such that
$\fA \models \nu_j[a,\bar{a}]$ for all $j \in J$, 
assign $\bar{a}$ to the extension of $q_{J}$, and pick some particular $a$ for which this condition
is satisfied. Since $\fA \models \phi$,
for any individual $b \in A$, we have 
$\fA \models \left(\bigcup_{j \in J} \Delta_J \cup \Omega \right)[a,\bar{a},b]$, whence,
by the validity of the resolution rule, 
$\fA \models \left(\bigcup_{j \in J} \Delta_J \cup \Omega \right)^\circ[\bar{a},b]$. 
Thus, $\fA^+$ verifies the formulas~\eqref{eq:flM2} and~\eqref{eq:flM5}, whence
$\fA^+ \models \psi$, as required.

Suppose, conversely, that $\fB \models \psi$. We expand to a structure $\fB^+$ interpreting the $(m+1)$-ary predicates
of $\phi$ in such a way that $\fB^+ \models \phi$. Fix for the moment some
element $a$ and ($m-1$)-tuple of elements $\bar{a}$, and define 
$I  = \set{i \in S \mid \fB \models \mu_i[a,\bar{a}]}$ and 
$J  = \set{j \in T \mid \fB \models \nu_j[a,\bar{a}]}$. It follows from~\eqref{eq:flM1}
that $\fB \models p_{I,J}[\bar{a}]$. Indeed, from~\eqref{eq:flM3}, there 
exists a minimal cover $M = \set{C_1, \dots, C_\ell}$ of $I$ such that $\fB \models p_{I,J,M}[\bar{a}]$, whence from~\eqref{eq:flM4}, we can find elements $b_1, \dots, b_\ell$ such that, for all $h$ ($1 \leq h \leq \ell$), 
$\fB \models \left(\bigcup_{i \in C_h} \Gamma_i \cup \bigcup_{j \in J} \Delta_J \cup \Omega \right)^\circ[\bar{a},b_h]$,
and $\fB \models p_{I,J,M,h}[\bar{a},b_h]$. Letting $\tau_h= \ftp^\fB[\bar{a},b_h]$, it follows from Lemma~\ref{lma:resolution} that there exists a fluted ($m+1$)-type $\tau^+_h \supseteq \tau_h$ such
that $\models \tau^+_h \rightarrow \left(\bigcup_{i \in C_h} \Gamma_i \cup \bigcup_{j \in J} \Delta_J \cup \Omega \right)$.
From~\eqref{eq:flM6}, the $b_h$ are all distinct, so we may interpret the ($m+1$)-ary predicates of $\phi$ in $\fB^+$ 
so that $\ftp^{\fB^+}[a,\bar{a},b_h]= \tau^+_h$. Carrying out this process
for all $m$-tuples $(a,\bar{a})$, we thus ensure that $\fB^+ \models \bigwedge_{i \in S} \forall^m (\mu_i \rightarrow \exists \Gamma_i)$. Note that we have not assigned any ($m+1$)-tuples in such a way as to violate the constraints $\bigwedge_{j \in T} \forall^m (\nu_j \rightarrow \forall \Delta_j)$ or $\forall^{m+1} \Omega$. To complete
the definition of $\fB^+$, let $a,\bar{a},b$ be an ($m+1$)-tuple for which the extensions of the ($m+1$)-ary predicates have not been fixed. Again, let 
$J  = \set{j \in T \mid \fB \models \nu_j[a,\bar{a}]}$. It follows from~\eqref{eq:flM2}
that  $\fB \models q_{J}[\bar{a}]$, and thence from~\eqref{eq:flM5} that
$\fB \models \left(\bigcup_{j \in J} \Delta_J \cup \Omega \right)^\circ[\bar{a},b]$. Now let 
$\tau= \ftp^\fB[\bar{a},b]$, so that, from Lemma~\ref{lma:resolution}, there exists a fluted ($m+1$)-type $\tau^+ \supseteq \tau$ such
that $\models \tau^+ \rightarrow \left(\bigcup_{j \in J} \Delta_J \cup \Omega \right)$. 
Hence we may interpret the ($m+1$)-ary predicates of $\phi$  in $\fB^+$ 
so that $\ftp^{\fB^+}[a,\bar{a},b]= \tau^+$. 
At the end of this process, 
$\fB^+ \models \bigwedge_{j \in T} \forall^m (\nu_j \rightarrow \forall \Delta_j) \wedge \forall^{m+1} \Omega$. Thus, $\fB^+ \models \phi$.
\end{proof}

We have finally reached the goal of this section.
\begin{theorem}
The satisfiability problem for $\FlEqTransVars{1}{m}$ is in $m$-\NExpTime. 
\label{theo:FLotransmEqUpper}
\end{theorem}
\begin{proof}
Let an $\FlEqTransVars{1}{m}$-sentence $\phi$ be given. By Lemma~\ref{lma:flMnf}, we may assume without loss of 
generality that $\phi$ is in normal form.  
We proceed by induction, starting with $m=2$. (The cases $m=0$ and $m=1$ are trivial.) The base case
is Lemma~\ref{lma:FLotranstEqUpper}. For the recursive case, 
Lemma~\ref{lma:FLotransmEqToFLotransmEq} reduces the original problem to the corresponding problem for $m-1$,
but with an exponential blow-up.
\end{proof}
Before moving the the next section we obtain a corollary concerning the finite satisfiability problem. 
\begin{corollary}
The finite satisfiability problem for $\FlEqTransVars{1}{m}$ is in\nb{I: parens inserted in $m+1$}\newline
$(m+1)$-\NExpTime.\label{cor:finsat1T}
\end{corollary}
\begin{proof}
The proof differs from the proof of Theorem~\ref{theo:FLotransmEqUpper} only in the base case, where we apply the fact that the finite satisfiability problem for $\FOt$ with one transitive relation and equality is decidable in \ThreeNExpTime~\cite{purdyTrans:ph18}; this complexity bound obviously applies to $\FlEqTrans{1}$. 
\end{proof}

%% file: 2-Undecidable-2T-Eq.tex
\section{Fluted Logic with more Transitive Relations}\label{sec:undecidable}
In the previous section, we considered $\FLt$ extended with a single transitive relation and equality.
In this section we consider $\FLt$ extended with more transitive relations.
Specifically, we show that the satisfiability and finite satisfiability problems for
$\FLt_=2T$
(two-variable fluted logic with two transitive relations and equality) or for $\FLt{}3T$
(two-variable fluted logic with three transitive relations but without equality), are all undecidable.

A \emph{tiling system} is a tuple $\boldsymbol{\mathcal{C}}=({\mathcal C}, H, V)$, where
$\mathcal C$ is a finite set of {\em tiles},
and
$H$, $V \subseteq {\mathcal C} \times {\mathcal C}$ are the  {\em horizontal} and {\em vertical} constraints.
A \textit{tiling} of $\N^2$ for $\boldsymbol{\mathcal{C}}$ is a function $f: \N^2 \rightarrow {\mathcal C}$, such that 
for all $X,Y \in \N$, 
$(f(X,Y), f(X+1,Y)) \in  H$ and $(f(X,Y), f(X,Y+1)) \in  V$. Intuitively, we think of $f$ as assigning 
(a copy of) some tile in $\cC$ to each point with integer coordinates in the upper-right quadrant of the plane: this assignment
must respect the horizontal and vertical constraints, understood as a list
of which tiles may be placed immediately to the right of (respectively: immediately above) which others. A tiling
is {\em periodic} if there exist $m$, $n$ such that, for all $X$ and  $Y$, $f(X+m,Y) = f(X,Y+n) = f(X,Y)$. 
Denote by 
$\N^2_{m,n}$ the finite initial segment
$[0,m-1] \times [0,n-1]$ of $\N^2$. A \textit{tiling} of $\N^2_{m,n}$ is a function
$f: \N^2_{m,n} \rightarrow {\mathcal C}$, such that 
for all $X$, $Y$ ($0 \leq X < m-1$, $0 \leq Y \leq n-1$),
$(f(X,Y), f(X+1,Y)) \in  H$ and 
for all $X$, $Y$ ($0 \leq X \leq  m-1$, $0 \leq Y < n-1$), $(f(X,Y), f(X,Y+1)) \in  V$. If $f$ is a tiling (of either $\N^2_{m,n}$ or $\N^2$), we
call the value $f(0,0)$ the {\em initial condition}, and, if $f$ is a tiling of $\N^2_{m,n}$,
we call the value $f(m-1,n-1)$ the {\em final condition}.

There are many undecidability results concerning tiling systems.
The {\em infinite tiling problem with initial condition} is the following: given a tiling system $\boldsymbol{\mathcal{C}}$ and
a tile $C_0 \in \cC$, does there exist a tiling of $\N^2$ for $\boldsymbol{\mathcal{C}}$ with initial condition $C_0$?
The {\em finite tiling problem with initial and final conditions} is the following: given a tiling system $\boldsymbol{\mathcal{C}}$ and
tiles $C_0, C_1 \in \cC$, do there exist positive $m$, $n$ and a tiling of $\N^2_{m,n}$ for $\boldsymbol{\mathcal{C}}$ with initial condition $C_0$ and final condition $C_1$? It is straightforward to show:
\begin{proposition}
The infinite tiling problem with initial condition and the finite tiling problem with initial and final conditions are both undecidable.
\label{prop:noddyTiling}
\end{proposition}
The following result, by contrast, is deep (see~e.g.~\cite[p.~90]{purdyTrans:BGG97}). Recall that sets $A$ and $B$ are \textit{recursively inseparable} if there exists no recursive (=decidable) set $S$ such that $A \subseteq S$ and $B \cap S = \emptyset$.
\begin{proposition}
The set of tiling systems for which there exists a periodic tiling of $\N^2$ is recursively inseparable from the set 
of tiling systems for which there exists no tiling of $\N^2$.
\label{prop:insep}
\end{proposition}

\subsection{The case of two transitive relations}\label{sec:twotrans}

In this section we show that both the satisfiability and the finite satisfiability problems for $\FLtwotranstEq$ are undecidable. 
(Recall from Example~\ref{ex:two} that $\FLtwotrans$ admits infinity axioms.)
	\begin{figure}[thb]
		\begin{center}
			\begin{minipage}{13cm}
				\begin{minipage}{12.5cm}
					\begin{center}
						\begin{tikzpicture}[xscale=1,yscale=1]
						\clip (-0.15,-0.15) rectangle (7.3,4.3);
						\foreach \x in {-1,1,3,5,7}
						\foreach \y in {-1,1,3,5,7} \foreach \z in {0.1} 
						{
							\draw[color=red,thick,-, >=latex] (\x+\z, \y+\z) -- (\x+1-\z, \y+\z) --
							(\x+1-\z, \y+1-\z) -- (\x+\z, \y+1-\z) -- (\x+\z, \y+\z);
							\draw[color=red,thick,-, >=latex] (\x+\z, \y+\z) --
							(\x+1-\z, \y+1-\z) -- (\x+\z, \y+1-\z) -- (\x+1-\z, \y+\z);
							
						}	
						
						\foreach \x in {0,2,4,6}
						\foreach \y in {0,2,4,6} \foreach \z in {0.1} 
						{
							\draw[color=blue,thick,-, >=latex] (\x+\z, \y+\z) -- (\x+1-\z, \y+\z) --
							(\x+1-\z, \y+1-\z) -- (\x+\z, \y+1-\z) -- (\x+\z, \y+\z);
							\draw[color=blue,thick,-, >=latex] (\x+\z, \y+\z) --
							(\x+1-\z, \y+1-\z) -- (\x+\z, \y+1-\z) -- (\x+1-\z, \y+\z);
						}	
						
						\foreach \x in {1,5}
						\foreach \y in {1,5} \foreach \z in {0.15} 
						{
							\draw[color=blue,thick,->, >=latex] (\x, \y) -- (\x+1-\z, \y);
							\draw[color=blue,thick,->, >=latex] (\x, \y) -- (\x, \y+1-\z);
							\draw[color=blue,thick,->, >=latex] (\x-1, \y) -- (\x-1, \y+1-\z);
							\draw[color=blue,thick,->, >=latex] (\x, \y-1) -- (\x+1-\z, \y-1);						
						}	
						
						\foreach \x in {3,7}
						\foreach \y in {3,7} \foreach \z in {0.15} 
						{
							\draw[color=blue,thick,->, >=latex] (\x, \y) -- (\x+1-\z, \y);
							\draw[color=blue,thick,->, >=latex] (\x, \y) -- (\x, \y+1-\z);
							\draw[color=blue,thick,->, >=latex] (\x-1, \y) -- (\x-1, \y+1-\z);
							\draw[color=blue,thick,->, >=latex] (\x, \y-1) -- (\x+1-\z, \y-1);						
						}

						\foreach \x in {1,5}
						\foreach \y in {3,7} \foreach \z in {0.15} 
						{
							\draw[color=blue,thick,->, >=latex] (\x+1-\z, \y) -- (\x, \y) ;
							\draw[color=blue,thick,->, >=latex] (\x, \y+1-\z) -- (\x, \y);
							\draw[color=blue,thick,->, >=latex] (\x-1, \y+1-\z) -- (\x-1, \y);
							\draw[color=blue,thick,->, >=latex] (\x+1-\z, \y-1) -- (\x, \y-1);						
						}	
						
						\foreach \x in {3,7}
						\foreach \y in {1,5} \foreach \z in {0.15} 
						{
							\draw[color=blue,thick,->, >=latex] (\x+1-\z, \y) -- (\x+\z, \y) ;
							\draw[color=blue,thick,->, >=latex] (\x, \y+1-\z) -- (\x, \y+\z);
							\draw[color=blue,thick,->, >=latex] (\x-1, \y+1-\z) -- (\x-1, \y+\z);
							\draw[color=blue,thick,->, >=latex] (\x+1-\z, \y-1) -- (\x+\z, \y-1);						
						}	
						
						\foreach \x in {1,5}
						\foreach \y in {3,7} \foreach \z in {0.15} 
						{
							\draw[color=blue,thick,->, >=latex] (\x+1-\z, \y) -- (\x+\z, \y) ;
							\draw[color=blue,thick,->, >=latex] (\x, \y+1-\z) -- (\x, \y+\z);
							\draw[color=blue,thick,->, >=latex] (\x-1, \y+1-\z) -- (\x-1, \y+\z);
							\draw[color=blue,thick,->, >=latex] (\x+1-\z, \y-1) -- (\x+\z, \y-1);						
						}

						\foreach \x in {0,1,2,3,4,5,6,7}
						\foreach \y in {0,1,2,3,4,5,6,7}
						{
							\filldraw[fill=white] (\x, \y) circle (0.2); 
							\pgfmathtruncatemacro{\abc}{mod(\y,4)} 
							\pgfmathtruncatemacro{\bcd}{mod(\x,4)} 
							\coordinate [label=center:$_{\bcd\abc}$] (A) at (\x,\y);
						}

						
						\draw[ ->, dotted, very thick] (7.5,0) -- (8,0);
						\draw[ ->, dotted, very thick] (0,7.5) -- (0,8);
						\draw[ ->, dotted, very thick] (7.5,7.5) -- (8,8);
						
						%
						\end{tikzpicture}\end{center}
				\end{minipage}
				
			\end{minipage}		
		\end{center}
		
		\caption{Intended expansion of the $\N\times\N$ grid with two transitive relations {\color{blue}$T_1$}   and  {\color{red}$T_2$}. Edges without arrows represent connections in both direction. Nodes are marked by the indices of  the $c_{ij}$s they satisfy.  
			\label{fig:grid-two-transitive}}
	\end{figure}
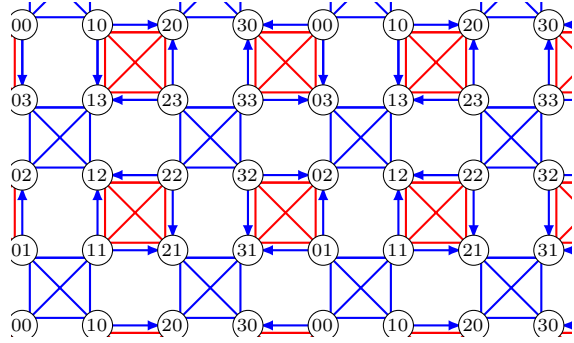
	
		Suppose the signature contains
		 two transitive relations $T_1$ and $T_2$, and additional unary predicates $c_{i,j}$ ($0\leq i,j\leq 3$) called {\em local address predicates}. We write a formula $\phi_{grid}$  capturing several properties of the intended expansion of the $\N^2$ grid  as shown in Fig.~\ref{fig:grid-two-transitive}.  
		 There, each element with coordinates $(X,Y)$ satisfies $c_{i,j}$, where $i=X\mod 4$ and $j=Y\mod 4$, and the transitive relations connect only some elements that are close in  the grid.   The formula $\phi_{grid}$ is a conjunction of the following statements~\eqref{tiling:initial}-\eqref{IfRed10}. 
	
	\medskip
\noindent		There is an initial element:
			\begin{equation}
			\exists c_{0,0}.\label{tiling:initial}
			\end{equation}

\noindent		The predicates $c_{i,j}$ enforce a partition of the universe: 
			\begin{eqnarray}
			\forall \big(\dot{\bigvee}_{0\leq i \leq 3} \dot{\bigvee}_{0\leq j \leq 3} c_{i,j}\big). 
				\label{tiling:partition}
			\end{eqnarray}
		
\noindent		Transitive paths do not connect distinct elements with the same local address:
		\begin{eqnarray}
			\bigwedge_{0\leq i,j\leq 3} \forall  (c_{i,j}\rightarrow \forall  ((T_1\vee  T_2)\wedge c_{i,j} \rightarrow =)).\label{Clique}
		\end{eqnarray}
		
\noindent Each element belongs to a 4-element $T_1$-clique:
		\begin{align}
\bigwedge_{i,j \in \set{0,2}}
	\forall \quad \big( (c_{i,j} &\rightarrow \exists  (T_1\wedge c_{i+1,j}))\wedge 	
	(c_{i+1,j} \rightarrow \exists  (T_1\wedge c_{i+1,j+1})) \wedge\nonumber\\
	(c_{i+1,j+1}& \rightarrow \exists  (T_1\wedge c_{i,j+1})) \wedge  
	(c_{i,j+1} \rightarrow \exists  (T_1\wedge c_{i,j})) \big).
	\label{CliqueBlue}
	\end{align}
	
\noindent Each element belongs to a 4-element $T_2$-clique: 
	\begin{align}
		%
\bigwedge_{i,j \in \set{1,3}}\forall \quad \big( (c_{i,j} &\rightarrow \exists  (T_2\wedge c_{i+1,j}))\wedge 	
		(c_{i+1,j} \rightarrow \exists  (T_2\wedge c_{i+1,j+1})) \wedge	\nonumber\\
	(c_{i+1,j+1}& \rightarrow \exists  (T_2\wedge c_{i,j+1})) \wedge  
		(c_{i,j+1} \rightarrow \exists  (T_2\wedge c_{ij})) \big).\label{CliqueRed}
		\end{align}
		%
		
		
\noindent	Certain pairs of elements connected by one transitive relation are also connected by the other one, specifically: 
		\begin{align}
		&\bigwedge_{i=0,2} \forall  (c_{i,i} \rightarrow \forall  ((T_1\vee  T_2)\wedge (c_{i,i-1}\vee c_{i-1,i})\rightarrow (T_1\wedge  T_2))
		\label{IfRed00}\\
		&\bigwedge_{i=1,3} \forall  (c_{i,i} \rightarrow \forall  ((T_1 \vee  T_2)\wedge (c_{i,i+1}\vee c_{i+1,i})\rightarrow (T_1\wedge  T_2))
		\label{IfRed11}\\	
		&\bigwedge_{i=0,2} \forall  (c_{i,i+1} \rightarrow \forall  ((T_1 \vee  T_2)\wedge (c_{i,i+2}\vee c_{i-1,i+1})\rightarrow (T_1\wedge  T_2))
		\label{IfRed01}\\		%
		&\bigwedge_{i=1,3} \forall  (c_{i,i-1} \rightarrow \forall  ((T_1 \vee  T_2)\wedge (c_{i,i}\vee c_{i,i-2} )\rightarrow (T_1\wedge  T_2)).
		\label{IfRed10}
		\end{align}

		%
	A model of $\phi_{grid}$ is shown in Fig.~\ref{fig:grid-two-transitive}.  Observe\nb{L explanation added} that the formulas~\eqref{CliqueBlue} and \eqref{CliqueRed} work in tandem with~\eqref{Clique}. Namely, both \eqref{CliqueBlue} and \eqref{CliqueRed} generate, for a given element $a$ of some local address $c_{i,j}$ in any model of $\phi_{grid}$ four new elements of certain local addresses such that the fourth element, say $a'$, has the same local address as the element $a$. Formula~\eqref{Clique} then implies $a=a'$, hence the element $a$ is a member of a 4-element $T_1$-clique and a member of a (distinct) 4-element $T_2$-clique; members of these cliques can be uniquely identified by their local addresses (cf.~Fig.~\ref{fig:grid-two-transitive}).
	One can also obtain finite models over a toroidal grid structure $\Z_{4m}\times \Z_{4m}$  ($m>0$) by identifying elements from columns 0 and $4m$ and from rows 0 and $4m$. 
	
We show that any model of $\phi_{grid}$ embeds the standard grid $\N^2$ in a natural way. To see this,
for all $i$, $j$ in the range $0 \leq i, j <4$, define the formulas $h_{i,j}$ and $v_{i,j}$ as follows:
\begin{align*}
h_{i,j} := & 
\begin{cases}
		T_1 \wedge  c_{i+1,j} & \text{if $i$ is even}\\
		T_2 \wedge  c_{i+1,j} & \text{otherwise}
\end{cases} & &
v_{i,j} := & 
\begin{cases}
T_1 \wedge  c_{i,j+1} & \text{if $j$ is even}\\
T_2 \wedge  c_{i,j+1} & \text{otherwise.}
\end{cases}
\end{align*}
The intuition is that, for any element $a$ satisfying $c_{i,j}$, $h_{i,j}$ will be satisfied by the pair $[a,b]$ just in case
$b$ is immediately to the right of $a$, and $v_{i,j}$ will be satisfied by the pair $[a,b]$ just in case
$b$ is immediately above $a$. (See Fig.~\ref{fig:grid-two-transitive}.) 

\begin{lemma}
In any model $\fA$, of $\phi_{grid}$, the following hold for any $i,j$ in the range $0\leq i,j <4$:
\begin{align}
& \fA \models c_{i,j}[a] \; \Rightarrow \text{ there exists $b$ s.t.~$\fA \models h_{i,j}[a,b]$ and
	                                      $a'$ s.t.~$\fA \models v_{i,j}[a,a']$}
\label{eq:propagate}\\
& \fA \models c_{i,j}[a]\wedge h_{i,j}[a,b]\wedge v_{i,j}[a,a'] \wedge v_{i+1,j}[b,b']\quad \Rightarrow \quad \fA \models h_{i,j+1}[a',b'].
\label{eq:confluence}
\end{align}
\label{lma:grid}
\end{lemma}
\begin{proof}
Let $a\in A$ and $\fA \models c_{i,j}[a]$.  
The existence of $b$ in~\eqref{eq:propagate} is immediate from~\eqref{CliqueBlue} for $i,j$ even, and from~\eqref{CliqueRed} for $i,j$ odd.  
Suppose $i$ is even and $j$ is odd. By the last conjunct of~\eqref{CliqueBlue}, there is  $a_1\in A$ such that
$T_1[a,a_1]\wedge c_{i,j-1}[a_1]$. By~\eqref{CliqueBlue} again, there are $a_2,a_3,a_4\in A$ such that $\fA \models   T_1[a_1,a_2]\wedge c_{i+1,j-1}[a_2]\wedge  T_1[a_2,a_3]\wedge c_{i+1,j}[a_3] \wedge  T_1[a_3,a_4]\wedge c_{i,j}[a_4]$. 
By transitivity of $T_1$, $\fA\models T_1[a,a_4]$ and by~\eqref{Clique}, $a=a_4$, so the elements $a,a_1,a_2,a_3$ form a $T_1$-clique in $\fA$, hence $T_1[a,a_3]$ holds and, indeed, $\fA\models h_{i,j}[a,a_3]$.
In the same way we show the existence of $b$ when $i$ is odd and $j$ even, and, also, the existence of $a'$. 
%
We should regard the witnesses for the formulas $\exists h_{i,j}$ and $\exists v_{i,j}$ with respect to any element $a$ are the 
horizontal and vertical neighbours, respectively, of $a$.

We now establish~\eqref{eq:confluence} proceeding separately for the possible indices $i$ and $j$.
Consider first the case $i = j= 0$, and suppose $a,a',b$ and $b'$ are elements of $\fA$ such that $\fA\models c_{0,0}[a] \wedge T_1[a,b]\wedge c_{1,0}[b]\wedge T_1[a,a'] \wedge c_{0,1}[a']\wedge  T_1[b,b']\wedge  c_{1,1}[b']$. By \eqref{CliqueBlue} $b'$ is a member of a $T_1$-clique consisting of elements of local addresses $c_{1,1}, c_{0,1}, c_{0,0}, c_{1,0}$. Since by \eqref{Clique} the relation $T_1$ does not connect distinct elements of the same local address, $a'$ belongs to the $T_1$-clique of $b'$, so $\fA \models  T_1[a',b']$, and the claim follows. 

Consider now the case $i = 3$, $j= 0$, and suppose
 $a,a',b$ and $b'$ are elements  such that
$\fA \models c_{3,0}[a] \wedge  T_1[a,a']\wedge c_{3,1}[a']\wedge  T_2[a,b] \wedge c_{0,0}[b]\wedge  T_1[b,b'] \wedge c_{0,1}[b'] $. Applying~\eqref{CliqueRed} together with~\eqref{Clique} to $b$, we see that $b$ is a member of a 4-element $T_2$-clique consisting of elements of local addresses $c_{0,0}, c_{3,0}, c_{3,3}, c_{0,3}$. By~\eqref{Clique}, $a$ is a member of this clique,\nb{I: I'm struggling to see this. L: explanation added here and a general observation added after formula (32).} whence $\fA\models  T_2[b,a]$. 
By \eqref{IfRed00}, $\fA \models  T_1[b,a]$. Moreover, $b'$ is in a $T_1$-clique of $b$, and so $\fA \models  T_1[b',b]$. 
By transitivity of $T_1$,  $\fA \models  T_1[b',a']$.  Now, by \eqref{IfRed01}, $\fA\models  T_2[b',a']$. 
By \eqref{CliqueRed}, $a'$ is a member of a $T_2$-clique that, by \eqref{Clique}, must contain $b'$. Hence $h_{3,0}[a',b']$ holds and the claim follows.
		
		
	%
The remaining cases are dealt with similarly.
\end{proof}

Lemma~\ref{lma:grid} shows that any model $\fA$ of $\phi_{grid}$ contains, in effect, a homomorphic embedding of the infinite 
grid $\N^2$. Specifically, we define a function $\iota: \N^2 \rightarrow A$ as follows.
Set $\iota(0,0)$ to be some witness for (1). By~\eqref{eq:propagate},
we may choose $\iota(1,0), \iota(2,0), \dots$ such that, for all $X \geq 0$, setting $i=X \mod 4$, we have $\fA \models h_{i,0}[\iota(X,0), \iota(X+1,0)]$; and
then, for every $X \geq 0$, we may choose $\iota(X,1), \iota(X,2), \dots$ such that for every $Y \geq 0$, setting $j=Y\mod 4$, we have
$\fA \models v_{i,j}[\iota(X,Y), \iota(X,Y+1)]$.  A simple induction on $Y$ using~\eqref{eq:confluence} then shows that, for all $X$ and $Y$,
$\fA \models h_{i,j}[\iota(X,Y), \iota(X+1,Y)]$. 

We can now map any tiling system $\boldsymbol{\mathcal{C}}$ to an $\mathcal{FL}^2_=2\mbox{\rm T}$-formula $\eta_\cC$ in such 
a way that  $\cC$ has a tiling if any only if $\eta_\cC$ is satisfiable. We simply let $\eta_\cC$ be the conjunction of
$\phi_{grid}$ with the following formulas. 

\medskip 

\noindent Each node encodes precisely one tile: 
%
\begin{eqnarray}
\forall \big(\bigvee_{C \in {\cal C}} C  \wedge \bigwedge_{C \neq D} (\neg C   \vee \neg D )\big).\label{eq:2T-tiling1}
\end{eqnarray}

\noindent Adjacent tiles respect $ H$ and $ V$:
	\begin{eqnarray}
		%
\bigwedge_{C \in {\cal C}}	\bigwedge_{0 \leq i,j < 4} \forall \Big(C \wedge c_{i,j}  \rightarrow  \forall \big((h_{i,j}  \rightarrow\!\!\! \bigvee_{C': (C,C') \in  H}\!\!  C') \wedge (v_{i,j}  \rightarrow\!\!\! \bigvee_{C': (C,C') \in  V}\!\!  C')\big)\Big).
		\label{eq:2T-tiling2}
	\end{eqnarray}
	
If $f$ is a tiling of $\N^2$ for $\boldsymbol{\cal C}$, we expand the standard grid model of $\phi_{grid}$ by taking 
any predicate $C \in \cC$ to be satisfied by $(X,Y) \in \N^2$ just in case $f(X,Y) = C$. It is a simple matter to check that 
$\eta_\cC$ is true in the resulting structure. Conversely, 
if $\fA \models \eta_{\boldsymbol{\cal C}}$, then $\fA \models \phi_{grid}$, and so 
there exists a grid embedding $\iota: \N^2 \rightarrow A$. We then define a function
$f: \N^2 \rightarrow {\cal C}$ by setting $f(X,Y)$ to be the unique tile $C \in \cC$ such that $\fA \models C[\iota(X,Y)]$,
which is well-defined by~\eqref{eq:2T-tiling1}. 
By~\eqref{eq:2T-tiling2}, 
$f$ is a tiling for $\boldsymbol{\cal C}$.

Indeed, the same argument shows that, $\eta_\cC$ has a {\em finite} model if and only if there is a {\em periodic} tiling 
of $\N^2$ for $\boldsymbol{\cal C}$. Since, as remarked above, the set of tiling systems for which there exists no tiling of the plane is 
recursively inseparable from the set of tiling systems for which there exists a periodic tiling of the plane, we obtain:
\begin{theorem}
	The satisfiability problem and the finite satisfiability problems for $\mathcal{FL}^2_=2\mbox{\rm T}$ are both undecidable.
	\label{theo:twoEqUndecidable}
\end{theorem}
A quick check reveals that the formula $\eta_\cC$ lies in the guarded fragment of first-order logic. Moreover,
the proof of Lemma~\ref{lma:grid}  remains valid even if $T_2$ is required to be an equivalence relation. Thus we have:
\begin{corollary}
	The satisfiability problem and the finite satisfiability problems for the intersection of $\mathcal{FL}^2_=2\mbox{\rm T}$  with the guarded fragment are both undecidable. This result continues to hold if, in place of $\mathcal{FL}^2_=2\mbox{\rm T}$,
	we have $\mathcal{FL}^2_=1\mbox{\rm T}1\mbox{\rm E}$, the two-variable fluted fragment together with identity, one transitive relation and
	one equivalence relation.
\label{cor:twoEqUndecidable}
\end{corollary}
We conclude the section by remarking that decidability of the satisfiability and the finite satisfiability problems for $\FlTransVars{2}{m}$ remains open for every $m\geq 2$. We showed in Example~\ref{ex:two} that these two problems are distinct.

%% file: 2-Undecidable-3T.tex
\subsection{The case of three transitive relations}\label{sec:threetrans}

In this section we show that the satisfiability problem and the finite satisfiability problem for $\FLthreetrans$ are both undecidable. 
(Note that equality is not available in this logic.)
We start by reducing the infinite tiling problem to the satisfiability problem.

%


%


\input{newProof}

	
Equipped with Lemma~\ref{lem:3Tinfinity} we can now define a natural embedding $\iota$ of $\N^2$ into any model $\fA$ of $\phi_{grid}$ setting  $\iota(X,Y)=a_t$, where $a_t$ is the element of the infinite sequence as defined above such that $\varsigma(t)=(X,Y)$. In view of the above discussion it is easy to see that $\iota$ has the following properties: 
\begin{enumerate}[(H1)]
	\item  If $X\geq Y$ then $T_\diamond[\iota(X,Y),\iota(X+1,Y)]$. Moreover, if $X$ is even then 
	$T_\diamond[\iota(X,Y),\iota(X,Y+1)]$, and if $X$ is odd then $T_\diamond[\iota(X,Y+1),\iota(X,Y)]$.
	\item  If $X< Y$ then $T_\diamond[\iota(X,Y),\iota(X,Y+1)]$. Moreover, if $Y$ is even then 
	\mbox{$T_\diamond[\iota(X+1,Y),\iota(X,Y)]$}, and if $Y$ is odd then $T_\diamond[\iota(X,Y),\iota(X+1,Y)]$.
\end{enumerate}
The above observation allows us to write formulas that properly assign tiles to elements of the model of $\phi_{grid}$. We do this with a formula  $\phi_{tile}$,
which again features several conjuncts. 
The first conjunct is straightforward. We require that each node encodes precisely one tile and the initial element satisfies the initial tiling condition by adding to $\phi_{tile}$ the formula:
	\begin{equation}
	\forall \big(\bigvee_{C \in {\cal C}} C  \wedge \bigwedge_{C \neq D} (\neg C   \vee \neg D )\wedge (\mbox{lf}\wedge \mbox{dg}\rightarrow C_0) \big).\label{tiling3T:partition}
	\end{equation}
The next formulas ensure that adjacent tiles respect the constraints $H$ and $V$. To ensure that the horizontal constraints  are satisfied 
we add to $\phi_{tile}$ the following conjuncts for every $C\in\cal C$:
\begin{align}
\bigwedge_{0\leq i,j\leq 5}\forall (C \wedge d_{ij} \rightarrow 
\forall (T_\diamond\wedge d_{\moduSix{i+1},j}  &\rightarrow \bigvee_{C': (C,C') \in H}\!\!\!\!\!  C'))\label{tiling3T:horD}\\
\bigwedge_{0\leq i\leq 5}\bigwedge_{j=1,3,5}\forall (C \wedge c_{ij}  \rightarrow 
\forall (T_\diamond\wedge (c_{\moduSix{i+1},j}\vee d_{\moduSix{i+1},j}) & \rightarrow \bigvee_{C': (C,C') \in H}\!\!\!\!\!  C'))\label{tiling3T:horCRight}\\
\bigwedge_{0\leq i\leq 5}\bigwedge_{j=0,2,4}\forall (C \wedge (c_{i,j}\vee d_{i,j} )  \rightarrow 
\forall (T_\diamond\wedge c_{\moduSix{i-1},j} & \rightarrow \bigvee_{C': (C',C) \in H}\!\!\!\!\!  C')).\label{tiling3T:horCLeft}
\end{align}	
A similar group of conjuncts is added to handle the vertical constraints. Again, we add to $\phi_{tile}$ the following conjuncts for every $C\in\cal C$: 
\begin{align}
\bigwedge_{0\leq i,j\leq 5}\forall (C \wedge (c_{i,j}\vee d_{i,j}) & \rightarrow 
\forall (T_\diamond\wedge c_{i,\moduSix{j+1}}  \rightarrow \bigvee_{C': (C,C') \in V}\!\!\!\!\!  C'))\label{tiling3T:verC}\\
\bigwedge_{i=0,2,4}\bigwedge_{0\leq j\leq 5}\forall (C \wedge d_{i,j} & \rightarrow 
\forall (T_\diamond\wedge d_{i,\moduSix{j+1}}  \rightarrow  \bigvee_{C': (C,C') \in V} \!\!\!\!\! C'))\label{tiling3T:verDup}\\
\bigwedge_{i=1,3,5}\bigwedge_{0\leq j\leq 5}\forall (C \wedge d_{i,j} & \rightarrow 
\forall (T_\diamond\wedge d_{i,\moduSix{j-1}}  \rightarrow  \bigvee_{C': (C',C) \in V}\!\!\!\!\!  C')).\label{tiling3T:verDdown}
\end{align}	
This completes the definition of the formula $\phi_{tile}$. 
Finally, let $\eta_\cC$ be the conjunction of $\phi_{grid}$ and $\phi_{tile}$. We show that 
		$\eta_\cC$ is satisfiable  iff  $\boldsymbol{\cal C}$ tiles  $\N^2$.
Namely, if  $\boldsymbol{\cal C}$ tiles  $\N^2$ then  to show that $\eta_\cC$ is satisfiable we  expand our intended model $\fG$ for $\phi_{grid}$ assigning to every element of the grid a unique $C\in {\cal C}$ given by the tiling.   
		
		Now, let $\fA\models \eta_{\boldsymbol{\cal C}}$. Since $\fA\models \phi_{grid}$ consider 
		the embedding $\iota$  of the standard $\N^2$ grid into $\fA$ defined above. 
		We define a tiling of the $\N^2$ grid assigning to every node $(X,Y)\in \N^2$ the unique tile $C$ such that $\fA \models C(\iota(X,Y))$. Formula~\eqref{tiling3T:partition} ensures that this is well defined and satisfies the initial condition. Formulas~\eqref{tiling3T:horD}-\eqref{tiling3T:horCLeft}
ensure that the horizontal  constraints are satisfied and formulas~\eqref{tiling3T:verC}-\eqref{tiling3T:verDdown}
ensure that the vertical  constraints are satisfied.
Hence, we have the following

\begin{theorem}
	The satisfiability problem for $\FLthreetrans$, the two-variable fluted fragment with three transitive relations, is undecidable. 
	\label{th:three}
\end{theorem}

We remark that since the formula $\phi_{grid}$ is an axiom of infinity,  we cannot get simultaneously undecidability of the finite satisfiability problem applying Proposition~\ref{prop:insep}. To prove the latter we reduce from the finite tiling problem. We proceed as follows. First, we 
modify the formula $\phi_{grid}$ so that it no longer constructs an infinite chain of witnesses but the process is allowed to stop whenever the boustrophedon meets an element on the bottom row. In other words, the chain of witnesses corresponds to a square domain $\N^2_{2n,2n}$, for some $n\geq 1$.

Denote the modified formula $\phi_{sgrid}$. It contains some conjuncts taken directly $\phi_{grid}$, some that are modified versions
of conjuncts in $\phi_{grid}$, and some that are new. First of all, we employ an additional control predicate $\mbox{rt}$ intended to mark the rightmost  column of the square domain. This is secured by adding the following new conjunct to $\phi_{sgrid}$ (complementing the formula~\eqref{eq:bou:interact}):
\begin{equation}
\forall(\mbox{rt} \rightarrow \bigvee_{i=0,2,4}\bigvee_{j=0}^{5} d_{i,j})
\label{eq:bou:interact:right}
\end{equation}
and the following new control formula:
\begin{align}
& \bigwedge_{i=0}^5 \bigwedge_{j=0}^5 
\forall (d_{i,j} \wedge \pm \mbox{rt} \rightarrow \forall(T_\diamond \wedge d_{i,\moduSix{j-1}} \rightarrow \pm \mbox{rt})).
\label{eq:bou:ControlRight}
\end{align}
In $\phi_{sgrid}$ we modify the formula~\eqref{eq:bou:initial} by ensuring that the initial element does not satisfy $\mbox{rt}$ as follows:
\begin{equation}
\exists (d_{0,0}  \wedge \mbox{dg} \wedge \mbox{bt} \wedge \neg \mbox{rt}).
\label{eq:bou:initialRight}
\end{equation}
Finally, we modify the formula~\eqref{eq:bou:dGenE}; now we require a new witness only for bottom elements that are not on the rightmost column, writing: 
\begin{align}
& \bigwedge_{i=1,3,5} \forall  (d_{i,0} \wedge \mbox{bt} \wedge \neg \mbox{dg}  \wedge \neg \mbox{rt} \rightarrow \exists (d_{\moduSix{i+1},0} \wedge \mbox{bt} \wedge \neg \mbox{dg} \wedge T_{0}))
\label{eq:bou:dGenDRight}. 
\end{align}
Remaining conjuncts of $\phi_{grid}$ constitute conjuncts of $\phi_{sgrid}$ without modification.

Observe that $\phi_{sgrid}$ has finite models:  if a witness $a_t$ of the conjunct~\eqref{eq:bou:cGenE} happens to satisfy $\mbox{rt}$ then the following witnesses $a_{t'}$ with $t'>t$, corresponding to a downward column in the model, also satisfy $\mbox{rt}$ due to the control formula~\eqref{eq:bou:ControlRight}. As argued earlier, the chain of witnesses eventually reaches an element $a_{t''}$ satisfying $\mbox{bt}$, and this is where no new witnesses are required due to the modified conjunct~\eqref{eq:bou:dGenDRight}. 
Moreover in every finite model of $\phi_{sgrid}$ one can embed a square grid $\N^2_{2n,2n}$ similarly as we did before embedding the $\N^2$ grid in models of $\phi_{grid}$. 

In order to complete the reduction of  the finite tiling problem we need one more conjunct ensuring the final condition:
\begin{equation}
\forall (\mbox{dg}\wedge \mbox{rt}\rightarrow C_1).\label{tiling3T:partition:finite}
\end{equation}

It should be now straightforward to check that the conjunction of \eqref{tiling3T:partition:finite} with $\phi_{sgrid}\wedge \phi_{tile}$ is {\em finitely} satisfiable iff  $\boldsymbol{\cal C}$ tiles  $\N^2_{2n,2n}$, for some $n\geq 1$. Hence, we have the following:

\begin{theorem}
	The finite satisfiability problem for $\FLthreetrans$, the two-variable fluted fragment with three transitive relations, is undecidable. 
\label{th:threeFinsat}
\end{theorem}

We complete this section noticing that all formulas used in the proofs of Theorems \ref{th:three} and \ref{th:threeFinsat} are either guarded or can be rewritten as guarded. Furthermore, in the proof it would suffice to assume that $T_0$, $T_1$ and $T_2$ are interpreted as equivalence relations. Hence, we can strengthen the above theorem as follows.
\begin{corollary}
	The  (finite) satisfiability problem for the intersection of the fluted fragment with the two-variable guarded fragment is undecidable in the presence of three transitive relations \textup{(}or three equivalence relations\textup{)}. \label{cor:three}
\end{corollary}

%% file: newProof.tex
We write a formula $\phi_{grid}$  capturing several properties of the intended expansion of the $\N^2$ grid  as shown in Fig.~\ref{fig:zig-zag}. 
The formula $\phi_{grid}$ comprises a large number of conjuncts. To help give an overview of the construction, we have organized these conjuncts into groups, each of which secures a particular property (or collection of properties) exhibited by its models. We use the following notational conventions. If $i$ is an integer, $i/2$ indicates integer division without remainder (e.g.,~$5/2= 2$); moreover, $\modu{k}{i}$ denotes the remainder of $i$ on division by $k$, and 
$\moduSix{i}$ (i.e.,~without the subscript) denotes $\modu{6}{i}$.

The signature of $\phi_{grid}$ comprises the unary predicates $c_{i,j}$ and $d_{i,j}$ ($0 \leq i,j \leq 5$) and
$\mbox{bt}$, $\mbox{lf}$, $\mbox{dg}$ and $\mbox{dg}^+$, together with the distinguished 
binary predicates $T_0$, $T_1$ and $T_2$. We call the $c_{i,j}$ and $d_{i,j}$
{\em local address predicates}, and
require that they
partition the universe:
\begin{eqnarray}
\forall \big(\dot{\bigvee_{0\leq i,j \leq 5}}  c_{i,j}\; \dot{\vee}\; \dot{\bigvee_{0\leq i,j \leq 5}}  d_{i,j}\big).
\label{eq:bou:partition}
\end{eqnarray}
 
\input{rectangular-paths-with-yellow}

Informally, we think of an element in a structure interpreting these predicates
as having integer coordinates $(X,Y)$ in the plane such that if $Y >X$, then
its local address  is  $c_{i,j}$, where $i = \moduSix{X}$ and $j = \moduSix{Y}$,
and if $Y \leq X$, then  its local address is $d_{i,j}$, with $i$ and $j$ determined in the same way.
We call the unary predicates $\mbox{bt}$, $\mbox{lf}$, $\mbox{dg}$ and $\mbox{dg}^+$ {\em control predicates}, and
require them to interact with the local address predicates in certain ways:
\begin{equation}
\forall(\mbox{bt} \rightarrow \bigvee_{i=0}^{5} d_{i,0}) \wedge
\forall(\mbox{lf} \rightarrow \bigvee_{j=0}^{5} c_{0,j}) \wedge
\forall(\mbox{dg} \rightarrow \bigvee_{i=0}^{5} d_{i,i}) \wedge
\forall(\mbox{dg}^+ \rightarrow \bigvee_{j=0}^{5} c_{j,\moduSix{j+1}}).
\label{eq:bou:interact}
\end{equation}
Informally, we think of an element with coordinates $(X,Y)$ as satisfying $\mbox{bt}$ if $Y= 0$ (`bottom'), 
$\mbox{lf}$ if $X= 0$ and $Y >0$ (`left, but not bottom'), $\mbox{dg}$ if $Y = X$ (`diagonal'), and $\mbox{dg}^+$ if $Y = X+1$ (`super-diagonal').  
Finally, we call the binary predicates $T_0$, $T_1$ and $T_2$ {\em colours}. To aid visualization, we use
the respective synonyms  $\mbox{black}$, $\mbox{green}$ and $\mbox{red}$ for these predicates. 

We take $\phi_{grid}$ to contain conjuncts 
generating a sequence of elements $\set{a_t}_{t\geq 0}$ satisfying the $c_{i,j}$ and $d_{i,j}$ in a particular order. 
The intuition is that the elements of this sequence (each of which is assigned integer coordinates in the plane) follows the boustrophedon depicted in Fig.~\ref{fig:zig-zag} (thick grey arrow).  There is an `initial' element corresponding to the left bottom node:
\begin{equation}
\exists (d_{0,0}  \wedge \mbox{dg} \wedge \mbox{bt}).
\label{eq:bou:initial}
\end{equation}
This element has a $T_1$-successor satisfying $c_{0,1}$, $\mbox{dg}^+$ and $\mbox{lf}$:
\begin{align}
& \forall(\mbox{bt} \wedge \mbox{dg} 
\rightarrow \exists (c_{0,1} \wedge \mbox{dg}^+ \wedge \mbox{lf} \wedge T_{1}).
\label{eq:bou:dGenC}
\end{align}
Other elements satisfying $d_{i,j}$ in the sequence have successors given by the following conjuncts:
\begin{align}
& \bigwedge_{i=0,2,4} \bigwedge_{j=0}^{5} \forall  (d_{i,j}\wedge \neg \mbox{dg}  \rightarrow \exists (d_{i,\moduSix{j+1}} \wedge \neg \mbox{bt} \wedge T_{\modu{3}{j/2}} \wedge T_{\modu{3}{(j+1)/2+1}}))
\label{eq:bou:dGenA}\\
& \bigwedge_{i=1,3,5} \bigwedge_{j=0}^{5} \forall  (d_{i,j}\wedge \neg \mbox{bt}  \rightarrow
\exists (d_{i,\moduSix{j-1}} \wedge \neg \mbox{dg} 
\wedge T_{\modu{3}{j/2+1}} \wedge T_{\modu{3}{(j+1)/2-1}}))
\label{eq:bou:dGenB}\\
& \bigwedge_{i=1,3,5} \forall  (d_{i,0} \wedge \mbox{bt} \wedge \neg \mbox{dg}  \rightarrow \exists (d_{\moduSix{i+1},0} \wedge \mbox{bt} \wedge \neg \mbox{dg} \wedge T_{0}))
\label{eq:bou:dGenD}\\
& \bigwedge_{i=0,2,4} \forall(d_{i,i} \wedge \neg \mbox{bt} \wedge \mbox{dg} \rightarrow \exists
(c_{\moduSix{i-1},i} \wedge \mbox{dg}^+  \wedge \neg \mbox{lf}\wedge T_{\modu{3}{i/2-1}}  \wedge T_{\modu{3}{i/2}} ) ).
\label{eq:bou:dGenE}
\end{align}
Likewise, each element satisfying $c_{i,j}$ in the sequence has a successor given by the following conjuncts:
\begin{align}
& \bigwedge_{j=0,2,4} \bigwedge_{i=0}^{5} \forall  (c_{i,j}\wedge \neg \mbox{lf}  \rightarrow
\exists (c_{\moduSix{i-1},j} \wedge \neg \mbox{dg}^+ 
\wedge T_{\modu{3}{i/2-1}} \wedge T_{\modu{3}{(i+1)/2}}))
\label{eq:bou:cGenB}\\
& \bigwedge_{j=1,3,5} \bigwedge_{i=0}^{5} \forall  (c_{i,j}\wedge \neg \mbox{dg}^+  \rightarrow \exists (c_{\moduSix{i+1},j} \wedge \neg \mbox{lf} \wedge T_{\modu{3}{i/2+1}} \wedge T_{\modu{3}{(i+1)/2-1}}))
\label{eq:bou:cGenA}\\
& \bigwedge_{j=0,2,4} \forall  (c_{0,j} \wedge \mbox{lf} \rightarrow \exists (c_{0,j+1} \wedge \mbox{lf} 
 \wedge \neg \mbox{dg}^+ \wedge T_{1}))
\label{eq:bou:cGenD}\\
& \bigwedge_{j= 1,3,5} \forall(c_{{j-1},j} \wedge \mbox{dg}^+ \rightarrow \exists
(d_{j,j} \wedge \mbox{dg}  \wedge \neg \mbox{bt} \wedge T_{\modu{3}{(j+1)/2}} \wedge T_{\modu{3}{(j+3)/2}})).
\label{eq:bou:cGenE}
\end{align}

Starting with $a_0$ witnessing the formula~\eqref{eq:bou:initial}, we see that formulas~\eqref{eq:bou:dGenC}--\eqref{eq:bou:cGenE} generate, potentially, further elements. 
Accordingly, we call these conjuncts of $\phi_{grid}$ the {\em generation rules}.
Since the address predicates $c_{i,j}$
and 
$d_{i,j}$ 
form a partition,
at most one of these formulas has its preconditions satisfied, so we obtain a sequence $a_0, a_1, a_2, \dots$,  
satisfying the various predicates specified by those formulas. It is not obvious that the sequence $\set{a_t}$
defined in this way continues forever; but we shall show that it does. 

We give an informal explanation of how the sequence $\set{a_t}$ works. 
A good way to understand what is happening
is to suppose that there is some element $a_{t+1}$ in the sequence such that $d_{i,0}[a_{t+1}]$ (with $i$ even) and $\mbox{bt}[a_{t+1}]$.
(The formal proof below ensures that such $t$ exists; but for now we shall take this on trust.) 
Only two possible generation rules can apply: \eqref{eq:bou:dGenA} and~\eqref{eq:bou:dGenE}, depending on whether
$\mbox{dg}[a_{t+1}]$. If $\neg \mbox{dg}[a_{t+1}]$, then rule~\eqref{eq:bou:dGenA} applies and 
ensures that $d_{i,\moduSix{j+1}}[a_{t+2}]$. The\nb{I: small correction} first index in the local address remains as $i$, but the second index is incremented modulo 6.
Now the situation repeats, with the applicable generation rules being~\eqref{eq:bou:dGenA} and~\eqref{eq:bou:dGenE}. Thus, either the former
is applied forever, or we eventually generate an element $a_{t^+}$, say, such that $d_{i,j'}[a_{t^+}]$ (for some $j'$) and $\mbox{dg}[a_{t^+}]$. We will see presently
that the first of these alternatives is not possible; and on this assumption, we shall refer to the elements $a_{t+1}, \dots, a_{t^+}$ as an
{\em upward column}. The\nb{I: small correction} generation rule~\eqref{eq:bou:dGenA} ensures that each element in this sequence
is related to the next by two different colour-predicates. Let us call
these---in the order they appear in~\eqref{eq:bou:dGenA}---the {\em primary} colour and the {\em secondary} colour,
respectively. Since $d_{i,0}[a_{t+1}]$, and remembering our mnemonics {\em black}, {\em green} and {\em red} for $T_0$, $T_1$ and $T_2$, respectively,
we see that the sequences of primary and secondary colours on this upward column are
\begin{align*}
\text{black, black, green, green, red, red, \ldots}\\
\text{green, red, red, black, black, green, \ldots }
\end{align*}
repeating (as long as the column continues) with a period of six.
This is illustrated by the even-numbered columns in~ Fig.~\ref{fig:zig-zag} below the diagonal, where the primary colours are drawn to the left and the
secondaries to the right. Furthermore, rule~\eqref{eq:bou:dGenA} also ensures that the local addresses in the sequence are all $d_{i,j}$, with $i$
constant and $j$ cycling through the numbers $0, \dots, 5$.

A scan of the generation rules shows that $a_{t+1}$ can itself only have been generated by~\eqref{eq:bou:dGenD}, in which case
we have $d_{\moduSix{i-1},{0}}[a_{t}]$ and $\mbox{bt}[a_{t}]$, and indeed, by~\eqref{eq:bou:interact}, 
$\neg \mbox{dg}[a_{t}]$. Working backwards, the only we we could have generated $a_{t}$ is by~\eqref{eq:bou:dGenB}, whence 
$d_{\moduSix{i-1},1}[a_{t-1}]$ and $\neg \mbox{bt}[a_{t-1}]$. Comparing the local addresses of $a_{t}$ and $a_{t-1}$, we see
that the first index is $\moduSix{i-1}$ in both cases, but the second index has been incremented modulo 6. Let us continue to work
back. Only two possible generation rules could have yielded $a_{t-1}$: 
\eqref{eq:bou:dGenB} and~\eqref{eq:bou:cGenE}, depending on whether $\mbox{dg}[a_{t-1}]$. If $\neg \mbox{dg}[a_{t-1}]$, then
$a_{t-1}$ must have been generated by \eqref{eq:bou:dGenB}, whence $d_{\moduSix{i-1},2}[a_{t-2}]$. As
this cannot carry on for ever (for $a_0$ has local address $d_{0,0}$ and $\moduSix{i-1}$ is odd), we must have some $t^- < t$ such that 
$d_{\moduSix{i-1},j'}[a_{t^-}]$ (for some $j'$) and $\mbox{dg}[a_{t^-}]$. We refer to the subsequence $a_{t^-}, \dots, a_{t}$ as a
{\em downward column}. Each element in this sequence generates its successor via 
rule~\eqref{eq:bou:dGenB}, which ensures that the former
is related to the latter by two different colour-predicates, which we call---again in the order they appear in~\eqref{eq:bou:dGenB}---the 
{\em primary} colour and the {\em secondary} colour,
respectively. Since $d_{i,0}[a_{t}]$, we see that the sequences of primary and secondary colours on this upward again cycle through the
colours with period 6, but this time {\em ending} in the respective patterns
\begin{align*}
\text{\ldots, green, black, black, red, red, green}\\
\text{\ldots, red, red, green, green, black, black.}
\end{align*}
This is illustrated by the odd-numbered columns in~ Fig.~\ref{fig:zig-zag} below the diagonal, where the primary colours are drawn to the left and the
secondaries to the right. (To help the reader,
Table~\ref{table:3Tcolours} resolves the colour predicates in conjuncts~\eqref{eq:bou:dGenA}, \eqref{eq:bou:dGenB}, \eqref{eq:bou:cGenB} and \eqref{eq:bou:cGenA} for each $i$ and $j$.) In particular, we see that the sequence of primary colours counting forwards from $a_{t+1}$ is the
same as the sequence of secondary colours counting backwards from $a_t$. Furthermore, as we move backwards from $a_t$ to $a_{t^-}$, 
the elements all have local addresses $d_{\moduSix{i-1},j'}$, with $j'$ cycling\nb{I: $j \rightarrow j'$} through the numbers $0, \dots, 5$.

\begin{table}[tbh]
	\begin{center}
		\begin{minipage}{12cm}
			\begin{minipage}{0.45\textwidth}
				\begin{center}
					
					$	\begin{array}{c|c|c|c|c}
					&  \multicolumn{2}{c|}{\mbox{conjunct~\eqref{eq:bou:dGenA}:}} &  \multicolumn{2}{c}{\mbox{conjunct~\eqref{eq:bou:dGenB}:}}  \\
					&  \multicolumn{2}{c|}{\mbox{($i=0,2,4$)}} &  \multicolumn{2}{c}{\mbox{($i=1,3,5$)}}  \\
									j \quad & \mbox{\em pr.~c.} &  \mbox{\em sec.~c.}	&  \mbox{\em pr.~c.} & \mbox{\em sec.~c.}\\
					\hline
					
					0	&	0		&   1		& 	1	& 	2\\
					1	&	0		&   2		& 	0	& 	2\\
					2	&	1		&   2		& 	0	& 	1\\
					3	&	1		&  	0		& 	2	& 	1\\
					4	&	2		&   0		& 	2	& 	0\\
					5	&	2		&  	1		& 	1	& 	0\\

				\end{array}$
				\end{center}
				\end{minipage}
				\hspace{0.5cm}	
				\begin{minipage}{0.45\textwidth}
				\begin{center}

		$	\begin{array}{c|c|c|c|c}
			&  \multicolumn{2}{c|}{\mbox{conjunct~\eqref{eq:bou:cGenB}:}} &  \multicolumn{2}{c}{\mbox{conjunct~\eqref{eq:bou:cGenA}:}}  \\
			&  \multicolumn{2}{c|}{\mbox{($j$ even)}} &  \multicolumn{2}{c}{\mbox{($j$ odd)}}  \\
			i \quad & \mbox{\em pr.~c.} &  \mbox{\em sec.~c.}	&  \mbox{\em pr.~c.} & \mbox{\em sec.~c.}\\
			\hline
			
			0	&	2		&   0		& 	1	& 	2\\
			1	&	1		&   0		& 	1	& 	0\\
			2	&	1		&   2		& 	2	& 	0\\
			3	&	0		&  	2		& 	2	& 	1\\
			4	&	0		&   1		& 	0	& 	1\\
			5	&	2		&  	1		& 	0	& 	2\\

			\end{array}$
			\end{center}
			\end{minipage}
			\end{minipage}
			\end{center}
			
			\caption{Primary and secondary colours resolved. (Intended to help the reader.)}
			\label{table:3Tcolours}
		\end{table}
		
%

Now let us concentrate on the few elements surrounding $a_t$ and $a_{t+1}$.
We have established that $a_{t+1}$ was generated from $a_t$ by application of rule~\eqref{eq:bou:dGenD},\nb{I: small correction} so that 
$a_t$ and $a_{t+1}$ joined by $T_0$ (black).  Assuming that $t^- \leq t-2$ and $t^+ \geq t+3$, we have established that 
each element in the sequence 
$a_{t-2}, \dots, a_{t+3}$ is related to the next by $T_0$ (black), whence by transitivity, $T_0[a_{t-2}, a_{t+3}]$.
Thus we obtain a `black brick' of six elements connected in sequence by $T_0$, sitting on the bottom of the grid between
a downward column and a following upward column (see~Fig.~\ref{fig:zig-zag}).
We now add to $\phi_{grid}$ conjuncts which we refer to as {\em transfer} formulas: 
%

\begin{align}
& \bigwedge_{i=1,3,5}\bigwedge_{j=0,2,4}
\forall (d_{i,j} \rightarrow \forall ( d_{\moduSix{i+1},j} \wedge T_{\modu{3}{j/2-1}}   \rightarrow T_{\modu{3}{j/2}} ))
\label{eq:bou:dTransfer1}\\
& \bigwedge_{i=0,2,4}\bigwedge_{j=1,3,5}
\forall (d_{i,j} \rightarrow \forall  ( d_{\moduSix{i+1},j} \wedge T_{\modu{3}{j/2-1}}  \rightarrow T_{\modu{3}{j/2+1}}))
\label{eq:bou:dConvTransfer2}\\
& \bigwedge_{i=0,2,4}
\forall (d_{i,i} \wedge \mbox{dg}  \rightarrow \forall 
(c_{i,\moduSix{i+1}}\wedge  T_{\modu{3}{i/2}} \rightarrow T_{\modu{3}{i/2+1}}))
\label{eq:bou:TransferDC1}\\
& \bigwedge_{i=1,3,5}
\forall (d_{i,i} \wedge \mbox{dg}  \rightarrow \forall (c_{i,\moduSix{i+1}}\wedge  T_{\modu{3}{i/2}} \rightarrow T_{\modu{3}{i/2-1}}))
\label{eq:bou:TransferDC2}\\
& \bigwedge_{i=0,2,4}\bigwedge_{j=0,2,4}
\forall (c_{i,j} \rightarrow \forall ( c_{i,\moduSix{j+1}} \wedge T_{\modu{3}{i/2}} \rightarrow T_{\modu{3}{i/2+1}}))
\label{eq:bou:cTransfer1}\\
& \bigwedge_{i=1, 3, 5}\bigwedge_{j=1, 3, 5}
\forall (c_{i,j} \rightarrow (c_{i,\moduSix{j+1}} \wedge T_{\modu{3}{i/2}} \wedge  \rightarrow T_{\modu{3}{i/2-1}})).
\label{eq:bou:cTransfer2}
\end{align}
It follows from~\eqref{eq:bou:dTransfer1} (under the stated assumptions about the sequence $a_{t-2}, \dots,$ $a_{t+3}$),
that $T_1[a_{t-2}, a_{t+3}]$ (green). 
Now the argument repeats. 
Assuming that $t^- \leq t-4$ and $t^+ \geq t+5$, we see from the sequences of secondary colours in the 
downward column and primary colours in the upward column, that $T_1[a_{t-4},a_{t-2}]$ and  $T_1[a_{t+3},a_{t+5}]$. But we have just argued that
$T_1[a_{t-2}, a_{t+3}]$, so that, by transitivity, $T_2[a_{t-4}, a_{t+5}]$, giving us a 
`green brick' consisting of the six elements $a_{t-4}, a_{t-3}, a_{t-2}, a_{t+3}, a_{t+4}, a_{t+5}$. (Note that these are not consective in
the sequence $\set{a_t}$.) Furthermore, by~\eqref{eq:bou:dTransfer1}, $T_2[a_{t-4}, a_{t+5}]$. Continuing
this reasoning, as long as the downward and upward columns in question have at least $2\ell +1$ elements, 
we must have $T_{\modu{3}{\ell}}[a_{t-2\ell}, a_{t+1+2\ell}]$. That is, the elements $a_t$ display the pattern of `horizontal' colour links between every second element of the $(i-1)$st (downward) and $i$th (upward) columns, for $i$ non-zero and even,
as shown in Fig.~\ref{fig:zig-zag}.

Let us write $T_\diamond$ to abbreviate $T_0 \vee T_1 \vee T_2$; thus $T_\diamond[a,b]$ means that $a$ is related to $b$ by at least one of the colours. We now add to $\phi_{grid}$ conjuncts which we refer to as \textit{control formulas}: 
\begin{align}
& \bigwedge_{i=0}^5 
\forall (d_{i,i} \wedge \pm \mbox{dg} \rightarrow \forall(T_\diamond \wedge d_{\moduSix{i+1},\moduSix{i+1}} \rightarrow \pm \mbox{dg}))
\label{eq:bou:dControl1}
\\
& \bigwedge_{i= 0}^5 
\forall (d_{i,0} \wedge \pm \mbox{bt} \rightarrow \forall(T_\diamond \wedge d_{\moduSix{i+1},0} \rightarrow \pm \mbox{bt}))
\label{eq:bou:dControl2}
\\
& \bigwedge_{j=0}^{5} 
\forall (c_{\moduSix{j-1},j} \wedge \pm \mbox{dg}^+ \rightarrow \forall(T_\diamond \wedge c_{j,\moduSix{j+1}} \rightarrow \pm \mbox{dg}^+))
\label{eq:bou:cControl1}
\\
& \bigwedge_{j=0}^5
\forall (c_{0,j} \wedge \pm \mbox{lf} \rightarrow \forall(T_\diamond \wedge c_{0,\moduSix{j+1}} \rightarrow \pm \mbox{lf})).
\label{eq:bou:cControl2}
\end{align}
Here, the occurrences of $\pm$ are assumed to be resolved in the same way within a numbered display: thus, each of~\eqref{eq:bou:dControl1}--\eqref{eq:bou:cControl2} is actually a {\em pair} of formulas.
In particular, the formulas~\eqref{eq:bou:dControl1} say that, if $a$ if related to $b$ by any colour, and 
the local addresses of
$a$ and $b$ are as indicated, 
then $a$ satisfies $\mbox{dg}$ iff $b$ does.

Returning\nb{I: improved this paragraph} to our example of a downward column $a_{t^-}, \dots a_t$, followed by an 
upward column $a_{t+1}, a_{t+2}, \dots$, we observe from~\eqref{eq:bou:interact} that, since $j$ is odd and
$a_{t^-}$ by assumption satisfies $\mbox{dg}$, we may write $t^-= t-2\ell -1$ for some $\ell$.
Furthermore, since this downward column was generated by rule~\eqref{eq:bou:dGenB},
none of the elements 
$a_{t - 2\ell}, \dots,$ $a_{t}$ satisfies $\mbox{dg}$. It then follows from~\eqref{eq:bou:dControl1} and the colour
links just established  that successive elements 
$a_{t+2}, \dots, a_{t+2\ell+2}$ also do not satisfy $\mbox{dg}$, and indeed that the upward column extends 
at least to the point $a_{t+2\ell+3}$. But since $a_{t^-}= a_{t-2\ell-1}$ by assumption satisfies $\mbox{dg}$, 
it follows from~\eqref{eq:bou:dControl1} and the colour
links just established  that $a_{t+2\ell+3}$ does as well. Thus, the upward column ends precisely at the point $a_{t^+}= a_{t+2\ell+3}$.
Again, this is illustrated by adjacent columns below the diagonal in Fig.~\ref{fig:zig-zag}.

Similar reasoning applies to {\em rightward rows} (subsequences of $\set{a_t}$ in which elements
satisfy $c_{i,j}$ with $i$ fixed and odd, and with $j$ cycling through the indices 0, \dots, 5, as well
as {\em leftward rows}, defined similarly. 
Using the same argument as for the $d_{i,j}$, we see that, 
if there is a leftward row of length $2\ell+1$ ending in $a_{t-1}$ (where, by assumption, 
all elements satisfy $c_{i,j}$ with $i$ taking a common, even value), then there is a corresponding rightward row of length $2\ell+2$, and starting with $a_{t}$.
Moreover, elements in these rows are connected by vertical colour links as shown in Fig.~\ref{fig:zig-zag}, and 
none of the elements $a_{t}, \dots, a_{t + 2\ell +1}$ satisfies $\mbox{dg}^+$ (but $a_{t + 2\ell +2}$ does).   

Finally, we consider what happens at the end of an upward column (an element $a_{t^+}$ satisfying $\mbox{dg}$, and hence
$d_{i,i}$ with $i$ even). At that point~\eqref{eq:bou:dGenE} ensures that $a_{t^++1}$ satisfies 
$c_{\moduSix{i-1},i}$ and $\mbox{dg}^+$. Moreover, $T_{i/2}[a_{t^+},a_{t^++1}]$.  
Now consider the element $a_{t^-}$ at the start of the previous downward column. We have
already argued that $a_{t^-}$ is related to $a_{t^+}$ by $T_{i/2}$. But then, by transitivity,
$T_{i/2}[a_{t^-},a_{t^++1}]$, and hence, by~\eqref{eq:bou:TransferDC2}, $T_{\modu{3}{i/2-1}}[a_{t^-},a_{t^++1}]$. 
This allows us to coordinate the elements of the rightward row ending in $a_{t^--1}$ with the
leftward row beginning from $a_{t^++2}$. Similar remarks apply to columns.\nb{I: small corrections}

This concludes the informal presentation of the formula $\phi_{grid}$. Let us take stock. The generation rules~\eqref{eq:bou:initial}--\eqref{eq:bou:cGenE} generate a sequence of elements $\set{a_t}_{t\geq 0}$
satisfying certain local address predicates and control predicates. Quite independently, we define the boustrophedon curve
$\set{(X_t,Y_t)}_{t\geq 0}$ shown in Fig.~\ref{fig:zig-zag}. The transfer formulas~\eqref{eq:bou:dTransfer1}--\eqref{eq:bou:cTransfer2} 
and control formulas~\eqref{eq:bou:dControl1}--\eqref{eq:bou:cControl2} then ensure that the predicates
satisfied by each element $a_t$ are appropriate to the corresponding pair of coordinates $(X_t,Y_t)$. In particular,
the local address predicates tell us whether we are above or below the diagonal, and give the values $X_t$ and $Y_t$
modulo 6; and the control predicates tell us whether $(X_t,Y_t)$ lies on the bottom row, the left column, the diagonal or the super-diagonal. This is done by ensuring that (geometrically) neighbouring points are connected by colours as indicated in Fig.~\ref{fig:zig-zag}.

Let us now turn to the formal proof. Denote by $\varsigma(t)=(X_t, Y_t)$ the coordinates of the $t$\/\/\/th point
on the boustrophedon shown in Fig.~\ref{fig:zig-zag}, starting with $\varsigma(0)=(X_0,Y_0) = (0,0)$. We 
would like to show that, for each point in the sequence $\set{a_t}$,
the following properties are satisfied.
\begin{enumerate}[(P1)]
\item[(P1)] If $X_t < Y_t$, then $c_{i,j}[a_t]$, where $i= \moduSix{X_t}$ and $j= \moduSix{Y_t}$; 
if $X_t \geq Y_y$, then $d_{i,j}[a_t]$, where $i= \moduSix{X_t}$ and $j= \moduSix{Y_t}$.
\item[(P2)] We have: $\mbox{lf}[a_t]$ if and only if $X_t = 0$ and $Y_t>0$; $\mbox{bt}[a_t]$ if and only if $Y_t = 0$; $\mbox{dg}[a_t]$ if and only if $X_t = Y_t$; and $\mbox{dg}^+[a_t]$ if and only if $Y_t = X_t+1$.
\item[(P3)] If $s < t$ and the points $(X_s,Y_s)$ and $(X_t,Y_t)$ are connected by an arrow in Fig.~\ref{fig:zig-zag} of
colour $T_k$, then $T_k[a_s,a_t]$.
\end{enumerate}

\begin{lemma}
	Suppose $\fA \models \phi_{grid}$, and let the sequence $a_0, a_1, \ldots$ be constructed as described above. Then
	\mbox{(P1)}--\mbox{(P3)} hold for all $t\geq 0$.
\label{lem:boustrophedon}
\end{lemma}
\begin{proof}
By induction on $t$. For $t=0$, all statements in (P1)--(P3)
are either trivial or immediate from~\ref{eq:bou:initial}. Furthermore,
the only generation rule that applies in this case is~\eqref{eq:bou:dGenC}, in which case (P1)--(P3) are immediately secured 
for $t=1$. Suppose, then $t \geq 1$, and that (P1)--(P3) hold for all values up to $t$; we show that they hold for $t+1$. 
We proceed by cases, depending on whether $a_t$ satisfies either $d_{i,j}$ or $c_{i,j}$, and whether $i$ (respectively, $j$)
is odd or even. We give details for the case where $a_t$ satisfies $d_{i,j}$ and $i$ is even. The other cases are similar.

Assume first that $a_t$ does not satisfy $\mbox{dg}$. The generation rule that applies in this case is~\eqref{eq:bou:dGenA}, in 
whence $a_{t+1}$ satisfies $d_{i,\moduSix{j+1}}$ but not $\mbox{bt}$.
Now, by IH (P2), $X_t\neq Y_t$, hence by IH (P1): $X_t > Y_t$, with $X_t$ even. By the construction of the boustrophedon, then, $X_{t+1} = X_t$ and $Y_{t+1} = Y_t+1$ whence $X_{t+1} > 0$ and $X_{t+1} \geq Y_{t+1}$.  This immediately
secures all the conditions in (P1)--(P2) except for the condition that $\mbox{dg}[a_{t+1}]$ if and only if $X_{t+1} =Y_{t+1}$, which we must
establish. In addition, we must establish (P3).

We begin with the latter. That $T_{j/2}[a_{t},a_{t+1}]$ and $T_{\modu{3}{j/2+1}}[a_{t},a_{t+1}]$ is immediate from the generation rule~\eqref{eq:bou:dGenA}. Consulting Fig.~\ref{fig:zig-zag}, it remains only to show 
that, if $j$ is odd, and $s < t$ is such that $X_s = X_{t+1}-1$ and $Y_s = Y_{t+1}$, then $T_{j/2}[a_s,a_{t+1}]$ and $T_{\modu{3}{j/2+1}}[a_s,a_{t+1}]$. Now, using IH (P3), we see from Fig.~\ref{fig:zig-zag}
$T_{j/2}[a_s,a_{s+1}]$, $T_{j/2}[a_{s+1},a_{s+2}]$, $T_{j/2}[a_{s+2},a_{t-1}]$ and $T_{j/2}[a_{t-1},a_{t}]$; 
and we have just established that $T_{j/2}[a_{t},a_{t+1}]$. By transitivity of $T_{j/2}$, then, $T_{j/2}[a_s,a_{t+1}]$;
and by the transfer formula~\eqref{eq:bou:dTransfer1}, $T_{\modu{3}{j/2+1}}[a_s,a_{t+1}]$. Thus (P3) is established.
Returning to the missing condition in (P2), if $j$ is even, then, by IH (P1), so is $Y_t$; similarly, since $i$ is
even so is $X_t$. Thus $X_{t+1} = X_t \neq Y_{t+1} = Y_t+1$, and moreover, by~\eqref{eq:bou:interact}, $\neg \mbox{dg}[a_{t+1}]$, since $d_{i,j+1}[a_{t+1}]$. Thus, we may assume that $j$ is odd. 
But now let $s' <t-1$ be such that  
$X_{s'} = X_{t}-1$ and $Y_{s'} = Y_{t}$. By inspection of Fig.~\ref{fig:zig-zag} and applying IH (P3), we see
that $T_{j/2}[a_{s'-1},a_{s'}]$, $T_{j/2}[a_{s'},a_{t-1}]$,
and $T_{j/2}[a_{t-1},a_{t}]$; and, we have just established that 
$T_{j/2}[a_{t},a_{t+1}]$. By transitivity, then, $T_{j/2}[a_{s'-1},a_{t+1}]$. But
by IH (P2), $\mbox{dg}[a_{s'-1}] \Leftrightarrow X_{s'-1} = Y_{s'-1}$, and by the choice of $s'$,
$X_{s'-1} = Y_{s'-1} \Leftrightarrow X_{t+1} = Y_{t+1}$. Furthermore, 
having established that $T_{j/2}[a_{s'-1},a_{t+1}]$, it follows  
by the control formula~\eqref{eq:bou:dControl1} that $\mbox{dg}[a_{s'-1}] \Leftrightarrow \mbox{dg}[a_{t+1}]$.
Thus, $\mbox{dg}[a_{t+1}] \Leftrightarrow X_{t+1} = Y_{t+1}$ as required.

We assumed above that $a_t$ does not satisfy $\mbox{dg}$; now suppose that it does. The generation rule that applies in this case is~\eqref{eq:bou:dGenE}, and
(P1)--(P2) follow instantly. To establish (P3), referring to Fig.~\ref{fig:zig-zag}, 
we observe first that the generation rule itself ensures that $a_t$ is connected to $a_{t+1}$ by $T_{\modu{3}{i/2-1}}$ and $T_{\modu{3}{i/2}}$.
It remains to show that, if $s <t$ is such that $X_s= X_t-1$ and $Y_s= Y_t-1$, then  
$a_s$ is connected to $a_{t+1}$ by $T_{\modu{3}{i/2-1}}$ and $T_{\modu{3}{i/2+1}}$. By IH (P3), the successive
pairs in the sequence $a_s$, $a_{s+1}$, $a_{t-2}$, $a_{t-1}$, $a_{t}$ are connected by 
$T_{\modu{3}{i/2-1}}$; and we have just established that
$T_{\modu{3}{i/2-1}}[a_t,a_{t+1}]$. By transitivity, $T_{\modu{3}{i/2-1}}[a_s,a_{t+1}]$.
Since $d_{i-1,i-1}[a_s]$, and $c_{i-1,i}[a_{t+1}]$, it follows from the transfer formula~\eqref{eq:bou:TransferDC2}
that $T_{\modu{3}{i/2+1}}[a_s,a_{t+1}]$ as required.
\end{proof}

Lemma~\ref{lem:boustrophedon} justifies us in picturing the sequence $a_0, a_1, \ldots$ as laid out in Fig.~\ref{fig:zig-zag}, but it does not tell us that the elements of this sequence are distinct. However, we shall show that, in fact, $\phi_{grid}$\nb{I: was `this formula'} is an axiom of infinity.
As a preliminary, consider the rectangles into which the upper-right quadrant of the plane is divided by the black, green and red lines in Fig.~\ref{fig:zig-zag}. We refer to these rectangles as {\em bricks}. 
Each brick consists of four or six points in the plain, with the former kind confined to the left-hand and bottom edges; 
moreover, the bricks 
form a natural sequence following the boustrophedon. Since every point 
$\varsigma(t)=(X_t, Y_t)$ is associated with an element $a_t$ in some model of $\phi_{grid}$, we can think of bricks as
the set of associated elements.  And by inspection of Fig.~\ref{fig:zig-zag}, we see that for any brick $B$, 
there exists $k$ ($0 \leq k < 3$) such that, for all elements $a_s, a_t \in B$ with $s <t$, we have $T_k[a_s,a_t]$. In other words, each brick has a \textit{colour}, and, furthermore, an \textit{orientation} induced by the ordering of points on the boustrophedon. We call the bricks below
the diagonal having their left-hand margins in even columns  {\em downward-pointing}, while those 
below the diagonal having their left-hand margins in odd columns are {\em upward-pointing}; similarly for \textit{leftward-} and \textit{rightward-pointing} bricks above the diagonal, depicted by yellow arrows in Figure~\ref{fig:zig-zag}. Of course, while the elements of $B$ lie in order as the periphery of $B$ 
is traversed, they are not in general consecutive in the sequence $\set{a_t}$. In the light of the above discussion, the following are evident.
\begin{enumerate}[(E1)]
\item  Every element satisfying $d_{i,j}$ except for $a_0$
lies on at least one upward-pointing brick and at least one downward pointing brick.
\item  The colour and orientation of a brick $B$ is determined entirely by the local addresses of its elements;
hence two elements with the same local address lie on bricks with the same set of colours/orientations. 
\item In particular, if $B$ is a 6-element upward-pointing brick and its first element is a 
non-diagonal element, then that element  has local address $d_{i,j}$ ($i$ odd, $j$ even), while the last element
has local address $d_{\moduSix{i+1},j}$, and the colour of $B$ is $T_{\modu{3}{j/2-1}}$.
\item The first element of each brick $B$ is related to all the others by the colour of $B$,
and all the elements but the last are related to the last element by the colour of $B$.
\end{enumerate}


In the proof of the following lemma, recall that $\varsigma(t) = (X_t,Y_t)$, the $t$-th point in the  
boustrophedon.
\begin{lemma}
Suppose $\fA \models \phi_{grid}$, and let the sequence $\set{a_t}$ be as just constructed. Then the elements
of this sequence are all distinct.\label{lem:3Tinfinity}
\end{lemma}
\begin{proof}
Assume for contradiction that $a_s = a_t$ with $t < s$. 
We consider the case where $a_s = a_t$ satisfies some $d_{i,j}$; 
the case for elements satisfying some $c_{i,j}$ is handled similarly. 

Assume first that 
$Y_s = Y_t$.  Since $t <s$, and, $a_s$ has the same local address as $a_t$ (since they are identical), 
we must have $X_t < X_s$ and therefore, by (P1), $X_t < X_s - 5$.
As a preliminary, we claim that, if $a_s$ lies on a brick $B$
and $a_t$ on a brick $D$, then no element of either $B$ or $D$ can satisfy $\mbox{dg}$. For if $B$ has an
element $a_{s'}$ such that $\mbox{dg}[a_{s'}]$, then
$X_{t} < X_s - 5 \leq X_{s'} - 4 =  Y_{s'} - 4 \leq  Y_s - 2 =  Y_t - 2 < Y_t$ contradicting (P1) and the fact that 
$a_s= a_t$ satisfies some predicate $d_{i,j}$. In particular, $a_s = a_t$ itself does not satisfy $\mbox{dg}$. 
If, on the other hand, 
$D$ has an element $a_{t'}$ satisfying $\mbox{dg}$, then, by inspection of Fig~\ref{fig:zig-zag}, there is such
a $t'$ satisfying $t' > t$. Letting $s'= s+ (t' - t)$, we see that since the sequences $a_s, \dots a_{s'}$ and
$a_t, \dots a_{t'}$ are the same (and thus have the same local addresses), whence
$X_s, \dots X_{s'}$ and
$X_t, \dots X_{t'}$ move in the same way, so that $X_{t} < X_{s} -5$ implies $X_{t'} < X_{s'} -5$.
Thus, recalling that $\mbox{dg}[a_{s'}]$ implies $X_{t'} = Y_{t'}$, and that $Y_s = Y_t$ by assumption,
we have $X_{s'} > X_{t'} + 5 = Y_{t'} + 5 \geq Y_{t}+3 = Y_s+3 \geq Y_{s'}+1 > Y_{s'}$, contradicting the supposition
that $a_{t'} = a_{s'}$ satisfies $\mbox{dg}$. This proves the claim that neither $a_s$ nor $a_t$ lie on any brick containing
a diagonal element.
\input{rectangular-pathsExample1}

This claim having been established, we proceed to derive the promised contradiction. To make the proof easier, we 
suggest the reader follows with reference to the example $\varsigma(s)=(14,1)$ and $\varsigma(t)=(8,1)$
(see~Fig.~\ref{fig:zig-zag-Example1}).
Let $B_0$ and $D_0$ be the upward-pointing bricks containing, respectively, $a_s$ and $a_t$, and having the same
colour, say $T_{k_0}$. Let $a_{s_0}$ be the first element on the brick $B_0$, and $a_{t_0}$---the last element on the brick $D_0$, in our example, $\varsigma(s_0)=(13,2)$ and $\varsigma(t_0)=(8,2)$.    By (E4), $T_{k_0}[a_{s_0},a_{t_0}]$, i.e.~$a_{s_0}$ is connected to $a_{t_0}$ by an edge of some colour, $T_{k_0}$---in our example, black. The transfer formula~\eqref{eq:bou:dTransfer1} implies that $T_{\modu{3}{k_0+1}}[a_{s_0},a_{t_0}]$, in our case green. Now, write $k_1 = \modu{3}{k_0 +1}$, and let  $B_1$ and $D_1$ be the upward-pointing bricks of colour
$T_{k_1}$ (in our case, green), containing, respectively, $a_{s_0}$ and $a_{t_0}$. Let $a_{s_1}$ be the first element on the brick $B_1$, and $a_{t_1}$---the last element on the brick $D_1$, i.e.~$\varsigma(s_1)=(13,4)$ and $\varsigma(t_1)=(8,4)$. Again, by (E4), $a_{s_1}$ is connected to $a_{t_1}$ by a $T_{k_1}$-edge (green), hence by 
\eqref{eq:bou:dTransfer1}, also by an edge of colour $T_{k_2}$, where $k_2=\modu{3}{k_1+1}$ (red). 

Now the reasoning simply repeats, until either the brick above $B_\ell$ or the brick above $D_\ell$
contains an element satisfying $\mbox{dg}$. In particular, in our example, we consider   
$B_2$ and $D_2$---the red upward-pointing bricks containing, respectively, $a_{s_1}$ and $a_{t_1}$, and we let $a_{s_2}$ be the first element on the brick $B_2$, and $a_{t_2}$---the last element on the brick $D_2$. So, $\varsigma(s_2)=(13,6)$ and $\varsigma(s_2)=(8,6)$. Again, by (E4), $a_{s_2}$ is connected to $a_{t_2}$ by a red edge, hence by 
\eqref{eq:bou:dTransfer1}, also by a black one. Now the black brick above $D_2$ contains diagonal elements (i.e.~$l=2$); in particular, $\mbox{dg}[a_{t_2+2}]$, where $\varsigma(t_2+2)=(8,8)$. 

Recall that we are
assuming that $Y_s = Y_t$. By (E2), we have $Y_{s_0} = Y_{t_0}$, and, since we have been following the 
two columns of the boustrophedon upward, $Y_{s_\ell} = Y_{t_\ell}$. Moreover, since $t < s$, we have $X_{t_\ell} < X_{s_\ell}$, and indeed, $X_{t_\ell} < X_{s_\ell} - 5$. 
So, indeed, the process stops when the brick above $D_\ell$ contains an element satisfying $\mbox{dg}$ and, then,   
we necessarily
have 
$\mbox{dg}[a_{t_{\ell}+2}]$.
We have already established that
$d_{i_0, \moduSix{j_0+ 2\ell}}[a_{s_\ell}]$, $d_{\moduSix{i_0+1}, \moduSix{j_0+ 2\ell}}[a_{t_\ell}]$ and
$T_{k_0+l+1}[a_{s_\ell}, a_{t_\ell}]$ (black). By inspection of Fig.~\ref{fig:zig-zag}, we see that 
$T_{k_0+l+1}[a_{s_\ell-1}, a_{s_\ell}]$, and, indeed, 
$T_{k_0+l+1}[a_{t_\ell}, a_{t_\ell+2}]$. By transitivity, therefore
$T_{k_0+l+1}[a_{s_\ell-1}, a_{t_\ell+2}]$. 
On the other hand, since $X_{s_\ell-1}>Y_{s_\ell-1}$, {(P2)} implies that $a_{s_\ell-1}$  does not satisfy $\mbox{dg}$ .
But then we have
$d_{i_0, \moduSix{j_0+ 2\ell+1}}[a_{s_\ell-1}]$,
$d_{\moduSix{i_0+1}, \moduSix{j_0+ 2\ell+2}}[a_{t_\ell+2}]$
and $T_{k+l+1}[a_{s_\ell-1}, a_{t_\ell+2}]$, which, in the presence of~\eqref{eq:bou:interact},
violates the control formula~\eqref{eq:bou:dControl1}. 
In  our case,  $\varsigma(s_2-1)=(13,7)$ and we have $d_{1,1}[a_{s_2-1}]$,  $d_{2,2}[a_{t_2+2}]$,  $T_0[a_{s_2-1},a_{t_2+2}]$, $\neg \mbox{dg}[a_{s_2-1}]$ and $\mbox{dg}[a_{t_2+2}]$.

This deals with the case $Y_s = Y_t$. If $Y_s \neq Y_t$, then we 
let $B_0$ be any {\em downward}-pointing brick containing $a_s$, $T_{k}$ be the colour of $B_0$, and
$D_0$ the 
downward-pointing brick containing $a_t$ and having the same colour as $D_0$. Again, we let $s_0$ be the first element on 
$B_0$ and $t_0$ be the last element on $D_0$, following the preceding bricks $B_1, B_2, \dots$ and $D_1, D_2, \dots$.
This time, however, we will be moving \textit{down} the columns until we reach $B_\ell$ and $D_\ell$ such that
one of the elements $a_{s_\ell -1}$ or $a_{t_\ell + 1}$ satisfies $\mbox{bt}$. Now, the assumption that  
$Y_s \neq Y_t$ implies that at most one of $a_{s_\ell -1}$ and $a_{t_\ell + 1}$ satisfies $\mbox{bt}$, which
yields a violation of the control formula~\eqref{eq:bou:dControl2} using parallel reasoning to the upward case.
The process is depicted in 
Figure~\ref{fig:zig-zag-Example2} for one particular case. 
\end{proof}
\input{rectangular-pathsExample2}

%% file: rectangular-paths-with-yellow.tex
%

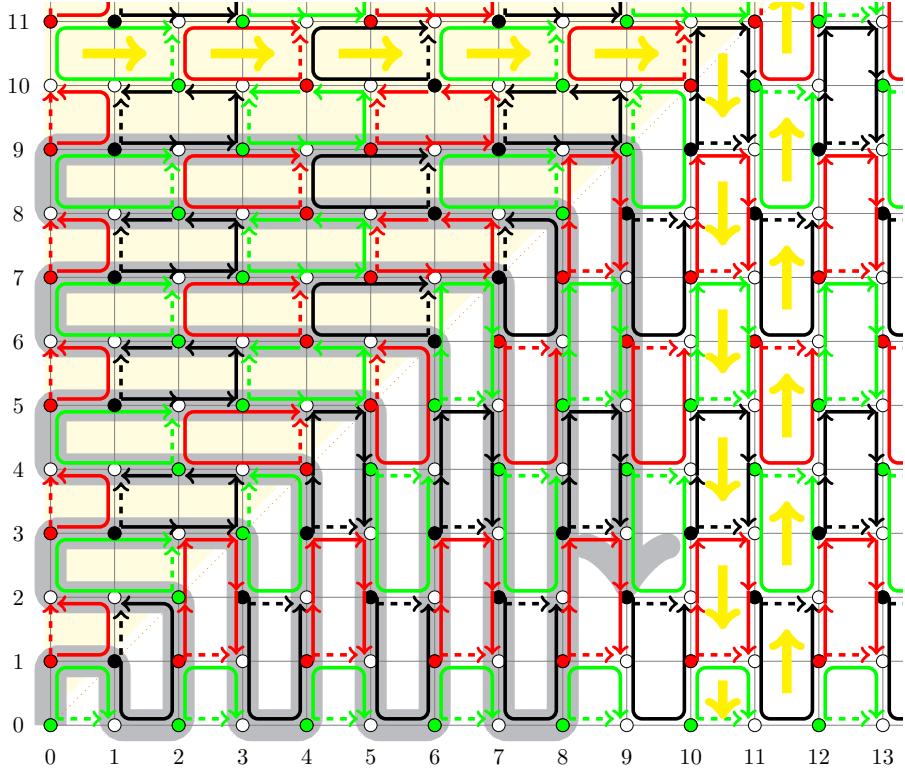
\begin{figure}[hbt]
	\begin{center}
		
		\resizebox{12.2cm}{!}{
			\begin{tikzpicture}
			
			\clip (-0.8,-0.8) rectangle (13.3,11.3);
	
	\path [fill=myYellow] (-0.1,0.1) to (14.0,14.1) to (-0.1,14.1) to (-0.1,0.1);

		
	\draw[->,ultra thick, myGray, line width=5mm, rounded corners]
	(0,0) -- (0,1) 
	-- (1,1) 
	-- (1,0) 
	-- (2,0) 
	-- (2,2)
	-- (0,2) 
	-- (0,3) 
	-- (3,3) 
	-- (3,0)
	-- (4,0) 
	-- (4,4) 
	-- (0,4)
	-- (0,5) 
	-- (5,5) 
	-- (5,0)
	-- (6,0) 
	-- (6,6) 
	-- (0,6)
	-- (0,7) 
	-- (7,7) 
	-- (7,0)
	-- (8,0) 
	-- (8,8) 
	-- (0,8)
	-- (0,9) 
	-- (9,9) 
	-- (9,2);
	

	\draw [help lines] (0,0) grid (18,18);
	\draw [orange, dotted] (0,0) -- (18,18);
	
			\foreach \x in {0,1,2,3,4,5,6,7,8,9,10,11,12,13,14,15,16,17} 
	\coordinate [label=center:\x] (A) at (\x, -0.5);
	\foreach \y in {0,1,2,3,4,5,6,7,8,9,10,11,12,13,14,15,16,17}  
	\coordinate [label=center:\y] (A) at (-0.5,\y);
	
		\foreach \x in {0,1,2,3,4,5,6,7,8,9,10,11,12,13,14,15,16,17}
	\foreach \y in {0,1,2,3,4,5,6,7,8,9,10,11,12,13,14,15,16,17}
	{
		\filldraw[fill=white] (\x, \y) circle (0.1); 
	}

\foreach \x in {1,7,13} \foreach \y in {\x}
\foreach \s in {0.1}
{	\filldraw[fill=black] (\x, \y) circle (0.1);
	\draw [->, rounded corners, ultra thick, black]  (\x+\s,\y) -- (\x+\s,\y-1+\s) -- (\x+1-\s,\y-1+\s)  -- (\x+1-\s,\y+1-\s) -- (\x+\s,\y+1-\s);
	\draw [->, dashed, ultra thick, black] (\x+\s,\y+\s) -- (\x+\s,\y+1-\s-\s);
}

\foreach \x in {3,9,15} \foreach \y in {\x}
\foreach \s in {0.1}
{	\filldraw[fill=green] (\x, \y) circle (0.1);
	\draw [->, rounded corners, ultra thick, green]  (\x+\s,\y) -- (\x+\s,\y-1+\s) -- (\x+1-\s,\y-1+\s)  -- (\x+1-\s,\y+1-\s) -- (\x+\s,\y+1-\s);
	\draw [->, dashed, ultra thick, green] (\x+\s,\y+\s) -- (\x+\s,\y+1-\s-\s);
}

\foreach \x in {5,11} \foreach \y in {\x}
\foreach \s in {0.1}
{	\filldraw[fill=red] (\x, \y) circle (0.1);
	\draw [->, rounded corners, ultra thick, red]  (\x+\s,\y) -- (\x+\s,\y-1+\s) -- (\x+1-\s,\y-1+\s)  -- (\x+1-\s,\y+1-\s) -- (\x+\s,\y+1-\s);
	\draw [->, dashed, ultra thick, red] (\x+\s,\y+\s) -- (\x+\s,\y+1-\s-\s);
}

\foreach \y in {2,8,14} \foreach  \x in {\y+1,\y+3,\y+5,\y+7,\y+9,\y+11,\y+13,\y+15,\y+17}  \foreach \s in {0.1}
{	\filldraw[fill=black] (\x, \y) circle (0.1);
	\draw [->, rounded corners, ultra thick, black]  (\x+\s,\y) -- (\x+\s,\y-2+\s) -- (\x+1-\s,\y-2+\s)  -- (\x+1-\s,\y-\s);
	\draw [->, dashed, ultra thick, black] (\x+\s,\y-\s) -- (\x+1-\s-\s,\y-\s);
}

\foreach \y in {0} \foreach  \x in {\y,\y+2,\y+4,\y+6,\y+8,\y+10,\y+12,\y+14,\y+16}  \foreach \s in {0.1}
{	\filldraw[fill=green] (\x, \y) circle (0.1);
	\draw [->, rounded corners, ultra thick, green]  (\x+\s,\y+\s) -- (\x+\s,\y+1-\s) -- (\x+1-\s,\y+1-\s)  -- (\x+1-\s,\y+\s);
	\draw [->, dashed, ultra thick, green] (\x+\s,\y+\s) -- (\x+1-\s-\s,\y+\s);
}

\foreach \y in {4,10,16} \foreach  \x in {\y+1,\y+3,\y+5,\y+7,\y+9,\y+11,\y+13,\y+15,\y+17}  \foreach \s in {0.1}
{	\filldraw[fill=green] (\x, \y) circle (0.1);
	\draw [->, rounded corners, ultra thick, green]  (\x+\s,\y) -- (\x+\s,\y-2+\s) -- (\x+1-\s,\y-2+\s)  -- (\x+1-\s,\y-\s);
	\draw [->, dashed, ultra thick, green] (\x+\s,\y-\s) -- (\x+1-\s-\s,\y-\s);
}

\foreach \y in {6,12,18} \foreach  \x in {\y+1,\y+3,\y+5,\y+7,\y+9,\y+11,\y+13,\y+15,\y+17}  \foreach \s in {0.1}
{	\filldraw[fill=red] (\x, \y) circle (0.1);
	\draw [->, rounded corners, ultra thick, red]  (\x+\s,\y) -- (\x+\s,\y-2+\s) -- (\x+1-\s,\y-2+\s)  -- (\x+1-\s,\y-\s);
	\draw [->, dashed, ultra thick, red] (\x+\s,\y-\s) -- (\x+1-\s-\s,\y-\s);
}

\foreach \y in {1,7,13} \foreach  \x in {\y+1,\y+3,\y+5,\y+7,\y+9,\y+11,\y+13,\y+15,\y+17}  \foreach \s in {0.1}
{	\filldraw[fill=red] (\x, \y) circle (0.1);
	\draw [->, rounded corners, ultra thick, red]   (\x+\s,\y) -- (\x+\s,\y+1-\s);
	\draw [->, rounded corners, ultra thick, red]	(\x+\s,\y+1-\s) --	(\x+\s,\y+2-\s);
	\draw [->, rounded corners, ultra thick, red]	(\x+\s,\y+2-\s) -- (\x+1-\s,\y+2-\s); 
	\draw [->, rounded corners, ultra thick, red]	(\x+1-\s,\y+2-\s) -- (\x+1-\s,\y+1+\s);
	\draw [->, rounded corners, ultra thick, red]	(\x+1-\s,\y+1+\s) -- (\x+1-\s,\y+\s);
	\draw [->, dashed, ultra thick, red] (\x+\s,\y+\s) -- (\x+1-\s-\s,\y+\s);
}

\foreach \y in {3,9,15} \foreach  \x in {\y+1,\y+3,\y+5,\y+7,\y+9,\y+11,\y+13,\y+15,\y+17}  \foreach \s in {0.1}
{	\filldraw[fill=black] (\x, \y) circle (0.1);
	\draw [->, rounded corners, ultra thick, black]   (\x+\s,\y) -- (\x+\s,\y+1-\s);
	\draw [->, rounded corners, ultra thick, black]	(\x+\s,\y+1-\s) --	(\x+\s,\y+2-\s);
	\draw [->, rounded corners, ultra thick, black]	(\x+\s,\y+2-\s) -- (\x+1-\s,\y+2-\s); 
	\draw [->, rounded corners, ultra thick, black]	(\x+1-\s,\y+2-\s) -- (\x+1-\s,\y+1+\s);
	\draw [->, rounded corners, ultra thick, black]	(\x+1-\s,\y+1+\s) -- (\x+1-\s,\y+\s);
	\draw [->, dashed, ultra thick, black] (\x+\s,\y+\s) -- (\x+1-\s-\s,\y+\s);
}

\foreach \y in {5,11,17} \foreach  \x in {\y+1,\y+3,\y+5,\y+7,\y+9,\y+11,\y+13,\y+15,\y+17}  \foreach \s in {0.1}
{	\filldraw[fill=green] (\x, \y) circle (0.1);
	\draw [->, rounded corners, ultra thick, green]   (\x+\s,\y) -- (\x+\s,\y+1-\s);
	\draw [->, rounded corners, ultra thick, green]	(\x+\s,\y+1-\s) --	(\x+\s,\y+2-\s);
	\draw [->, rounded corners, ultra thick, green]	(\x+\s,\y+2-\s) -- (\x+1-\s,\y+2-\s); 
	\draw [->, rounded corners, ultra thick, green]	(\x+1-\s,\y+2-\s) -- (\x+1-\s,\y+1+\s);
	\draw [->, rounded corners, ultra thick, green]	(\x+1-\s,\y+1+\s) -- (\x+1-\s,\y+\s);
	\draw [->, dashed, ultra thick, green] (\x+\s,\y+\s) -- (\x+1-\s-\s,\y+\s);
}

\foreach \x in {2,8,14} \foreach \y in {\x,\x+2,\x+4,\x+6,\x+8,\x+10,\x+12,\x+14,\x+16} \foreach \s in {0.1}
{	\filldraw[fill=green] (\x, \y) circle (0.1);
	\draw [->, rounded corners, ultra thick, green]  (\x-\s,\y+\s) -- (\x-2+\s,\y+\s) -- (\x-2+\s,\y+1-\s)  -- (\x-\s,\y+1-\s);
	\draw [->, dashed, ultra thick, green] (\x-\s,\y+\s) -- (\x-\s,\y+1-\s-\s);
}

\foreach \x in {4,10,16} \foreach \y in {\x,\x+2,\x+4,\x+6,\x+8,\x+10,\x+12,\x+14,\x+16} \foreach \s in {0.1}
{	\filldraw[fill=red] (\x, \y) circle (0.1);
	\draw [->, rounded corners, ultra thick, red]  (\x-\s,\y+\s) -- (\x-2+\s,\y+\s) -- (\x-2+\s,\y+1-\s)  -- (\x-\s,\y+1-\s);
	\draw [->, dashed, ultra thick, red] (\x-\s,\y+\s) -- (\x-\s,\y+1-\s-\s);
}

\foreach \x in {6,12,18} \foreach \y in {\x,\x+2,\x+4,\x+6,\x+8,\x+10,\x+12,\x+14,\x+16} \foreach \s in {0.1}
{	\filldraw[fill=black] (\x, \y) circle (0.1);
	\draw [->, rounded corners, ultra thick, black]  (\x-\s,\y+\s) -- (\x-2+\s,\y+\s) -- (\x-2+\s,\y+1-\s)  -- (\x-\s,\y+1-\s);
	\draw [->, dashed, ultra thick, black] (\x-\s,\y+\s) -- (\x-\s,\y+1-\s-\s);
}

\foreach \x in {1,7,13} \foreach \y in {\x+2,\x+4,\x+6,\x+8,\x+10,\x+12,\x+14,\x+16} \foreach \s in {0.1}
{	\filldraw[fill=black] (\x, \y) circle (0.1);
	\draw [->, rounded corners, ultra thick, black]  (\x+\s,\y+\s) -- (\x+1,\y+\s); 
	\draw [->, rounded corners, ultra thick, black]  (\x+1-\s,\y+\s) -- (\x+2-\s,\y+\s);
	\draw [->, rounded corners, ultra thick, black]  (\x+2-\s,\y+\s) -- (\x+2-\s,\y+1-\s); 
	\draw [->, rounded corners, ultra thick, black] (\x+2-\s,\y+1-\s) -- (\x+1+\s,\y+1-\s);
	\draw [->, rounded corners, ultra thick, black] 	(\x+1-\s,\y+1-\s) -- (\x+\s,\y+1-\s);
	\draw [->, dashed, ultra thick, black] (\x+\s,\y+\s) -- (\x+\s,\y+1-\s-\s);
}

\foreach \x in {3,9,15} \foreach \y in {\x+2,\x+4,\x+6,\x+8,\x+10,\x+12,\x+14,\x+16} \foreach \s in {0.1}
{	\filldraw[fill=green] (\x, \y) circle (0.1);
	\draw [->, rounded corners, ultra thick, green]  (\x+\s,\y+\s) -- (\x+1,\y+\s); 
	\draw [->, rounded corners, ultra thick, green]  (\x+1-\s,\y+\s) -- (\x+2-\s,\y+\s);
	\draw [->, rounded corners, ultra thick, green]  (\x+2-\s,\y+\s) -- (\x+2-\s,\y+1-\s); 
	\draw [->, rounded corners, ultra thick, green] (\x+2-\s,\y+1-\s) -- (\x+1+\s,\y+1-\s);
	\draw [->, rounded corners, ultra thick, green] 	(\x+1-\s,\y+1-\s) -- (\x+\s,\y+1-\s);
	\draw [->, dashed, ultra thick, green] (\x+\s,\y+\s) -- (\x+\s,\y+1-\s-\s);
}

\foreach \x in {5,11,17} \foreach \y in {\x+2,\x+4,\x+6,\x+8,\x+10,\x+12,\x+14,\x+16} \foreach \s in {0.1}
{	\filldraw[fill=red] (\x, \y) circle (0.1);
	\draw [->, rounded corners, ultra thick, red]  (\x+\s,\y+\s) -- (\x+1,\y+\s); 
	\draw [->, rounded corners, ultra thick, red]  (\x+1-\s,\y+\s) -- (\x+2-\s,\y+\s);
	\draw [->, rounded corners, ultra thick, red]  (\x+2-\s,\y+\s) -- (\x+2-\s,\y+1-\s); 
	\draw [->, rounded corners, ultra thick, red] (\x+2-\s,\y+1-\s) -- (\x+1+\s,\y+1-\s);
	\draw [->, rounded corners, ultra thick, red] 	(\x+1-\s,\y+1-\s) -- (\x+\s,\y+1-\s);
	\draw [->, dashed, ultra thick, red] (\x+\s,\y+\s) -- (\x+\s,\y+1-\s-\s);
}

\foreach \x in {0} \foreach  \y in {\x+1,\x+3,\x+5,\x+7,\x+9,\x+11,\x+13,\x+15,\x+17}  \foreach \s in {0.1}
{	\filldraw[fill=red] (\x, \y) circle (0.1);
	\draw [->, rounded corners, ultra thick, red]  (\x+\s,\y+\s) -- (\x+1-\s,\y+\s) -- (\x+1-\s,\y+1-\s)  -- (\x+\s,\y+1-\s);
	\draw [->, dashed, ultra thick, red] (\x,\y+\s) -- (\x,\y+1-\s);
}


\foreach \x in {10.5} \foreach  \y in {\x+2,\x,\x-2,\x-4,\x-6,\x-8}
	{\draw[->,ultra thick, Yellow, line width=1.5mm, rounded corners]
			(\x,\y) -- (\x,\y-1);
	\draw[->,ultra thick, Yellow, line width=1.5mm, rounded corners]
			(\x+1,\y-2) -- (\x+1,\y-1);
	}		 


\foreach \y in {10.5} \foreach  \x in {\y,\y-2,\y-4,\y-6,\y-8}
{\draw[->,ultra thick, Yellow, line width=1.5mm, rounded corners]
	(\x-2,\y) -- (\x-1,\y);
}		 

\foreach \x in {10.5} \foreach  \y in {0.7}
{
	\draw[->,ultra thick, Yellow, line width=1.5mm, rounded corners]
	(\x,\y) -- (\x,\y-0.6);
}

		\end{tikzpicture}
		}
		\vspace{0cm}		
	\end{center}
	
	\caption{Intended expansion of the  $\N \times \N$ grid and the boustrophedon order (thick gray path).}
	\label{fig:zig-zag}\label{fig:zig-zag-with-yellow}
	
\end{figure}

%% file: rectangular-pathsExample1.tex
%

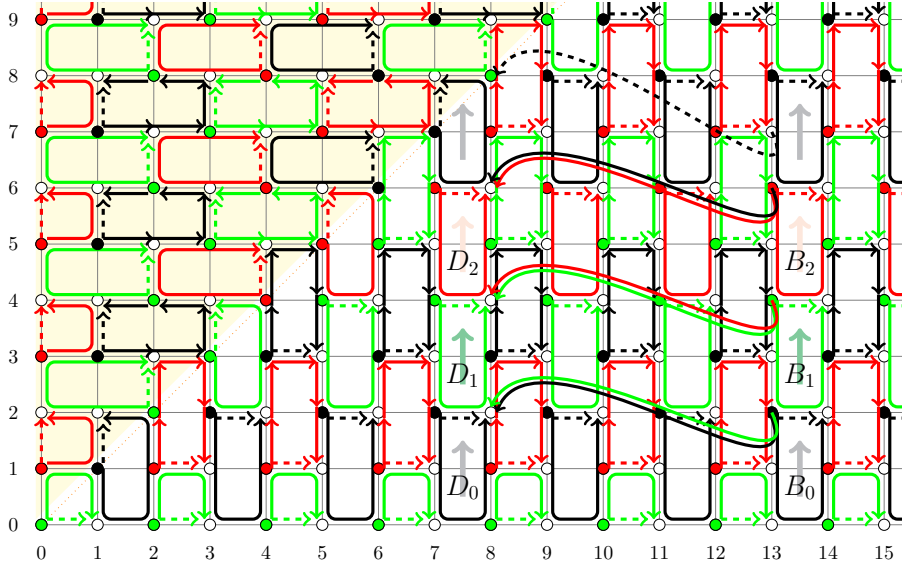
\begin{figure}[hbt]
	\begin{center}
		
		\resizebox{12.2cm}{!}{
			\begin{tikzpicture}
			
			\clip (-0.8,-0.8) rectangle (15.3,9.3);
	
	\path [fill=myYellow] (-0.1,0.1) to (14.0,14.1) to (-0.1,14.1) to (-0.1,0.1);

%
%

	\draw [help lines] (0,0) grid (18,18);
	\draw [orange, dotted] (0,0) -- (18,18);
	
			\foreach \x in {0,1,2,3,4,5,6,7,8,9,10,11,12,13,14,15,16,17} 
	\coordinate [label=center:\x] (A) at (\x, -0.5);
	\foreach \y in {0,1,2,3,4,5,6,7,8,9,10,11,12,13,14,15,16,17}  
	\coordinate [label=center:\y] (A) at (-0.5,\y);
	
		\foreach \x in {0,1,2,3,4,5,6,7,8,9,10,11,12,13,14,15,16,17}
	\foreach \y in {0,1,2,3,4,5,6,7,8,9,10,11,12,13,14,15,16,17}
	{
		\filldraw[fill=white] (\x, \y) circle (0.1); 
	}

\foreach \x in {1,7,13} \foreach \y in {\x}
\foreach \s in {0.1}
{	\filldraw[fill=black] (\x, \y) circle (0.1);
	\draw [->, rounded corners, ultra thick, black]  (\x+\s,\y) -- (\x+\s,\y-1+\s) -- (\x+1-\s,\y-1+\s)  -- (\x+1-\s,\y+1-\s) -- (\x+\s,\y+1-\s);
	\draw [->, dashed, ultra thick, black] (\x+\s,\y+\s) -- (\x+\s,\y+1-\s-\s);
}

\foreach \x in {3,9,15} \foreach \y in {\x}
\foreach \s in {0.1}
{	\filldraw[fill=green] (\x, \y) circle (0.1);
	\draw [->, rounded corners, ultra thick, green]  (\x+\s,\y) -- (\x+\s,\y-1+\s) -- (\x+1-\s,\y-1+\s)  -- (\x+1-\s,\y+1-\s) -- (\x+\s,\y+1-\s);
	\draw [->, dashed, ultra thick, green] (\x+\s,\y+\s) -- (\x+\s,\y+1-\s-\s);
}

\foreach \x in {5,11} \foreach \y in {\x}
\foreach \s in {0.1}
{	\filldraw[fill=red] (\x, \y) circle (0.1);
	\draw [->, rounded corners, ultra thick, red]  (\x+\s,\y) -- (\x+\s,\y-1+\s) -- (\x+1-\s,\y-1+\s)  -- (\x+1-\s,\y+1-\s) -- (\x+\s,\y+1-\s);
	\draw [->, dashed, ultra thick, red] (\x+\s,\y+\s) -- (\x+\s,\y+1-\s-\s);
}

\foreach \y in {2,8,14} \foreach  \x in {\y+1,\y+3,\y+5,\y+7,\y+9,\y+11,\y+13,\y+15,\y+17}  \foreach \s in {0.1}
{	\filldraw[fill=black] (\x, \y) circle (0.1);
	\draw [->, rounded corners, ultra thick, black]  (\x+\s,\y) -- (\x+\s,\y-2+\s) -- (\x+1-\s,\y-2+\s)  -- (\x+1-\s,\y-\s);
	\draw [->, dashed, ultra thick, black] (\x+\s,\y-\s) -- (\x+1-\s-\s,\y-\s);
}

\foreach \y in {0} \foreach  \x in {\y,\y+2,\y+4,\y+6,\y+8,\y+10,\y+12,\y+14,\y+16}  \foreach \s in {0.1}
{	\filldraw[fill=green] (\x, \y) circle (0.1);
	\draw [->, rounded corners, ultra thick, green]  (\x+\s,\y+\s) -- (\x+\s,\y+1-\s) -- (\x+1-\s,\y+1-\s)  -- (\x+1-\s,\y+\s);
	\draw [->, dashed, ultra thick, green] (\x+\s,\y+\s) -- (\x+1-\s-\s,\y+\s);
}

\foreach \y in {4,10,16} \foreach  \x in {\y+1,\y+3,\y+5,\y+7,\y+9,\y+11,\y+13,\y+15,\y+17}  \foreach \s in {0.1}
{	\filldraw[fill=green] (\x, \y) circle (0.1);
	\draw [->, rounded corners, ultra thick, green]  (\x+\s,\y) -- (\x+\s,\y-2+\s) -- (\x+1-\s,\y-2+\s)  -- (\x+1-\s,\y-\s);
	\draw [->, dashed, ultra thick, green] (\x+\s,\y-\s) -- (\x+1-\s-\s,\y-\s);
}

\foreach \y in {6,12,18} \foreach  \x in {\y+1,\y+3,\y+5,\y+7,\y+9,\y+11,\y+13,\y+15,\y+17}  \foreach \s in {0.1}
{	\filldraw[fill=red] (\x, \y) circle (0.1);
	\draw [->, rounded corners, ultra thick, red]  (\x+\s,\y) -- (\x+\s,\y-2+\s) -- (\x+1-\s,\y-2+\s)  -- (\x+1-\s,\y-\s);
	\draw [->, dashed, ultra thick, red] (\x+\s,\y-\s) -- (\x+1-\s-\s,\y-\s);
}

\foreach \y in {1,7,13} \foreach  \x in {\y+1,\y+3,\y+5,\y+7,\y+9,\y+11,\y+13,\y+15,\y+17}  \foreach \s in {0.1}
{	\filldraw[fill=red] (\x, \y) circle (0.1);
	\draw [->, rounded corners, ultra thick, red]   (\x+\s,\y) -- (\x+\s,\y+1-\s);
	\draw [->, rounded corners, ultra thick, red]	(\x+\s,\y+1-\s) --	(\x+\s,\y+2-\s);
	\draw [->, rounded corners, ultra thick, red]	(\x+\s,\y+2-\s) -- (\x+1-\s,\y+2-\s); 
	\draw [->, rounded corners, ultra thick, red]	(\x+1-\s,\y+2-\s) -- (\x+1-\s,\y+1+\s);
	\draw [->, rounded corners, ultra thick, red]	(\x+1-\s,\y+1+\s) -- (\x+1-\s,\y+\s);
	\draw [->, dashed, ultra thick, red] (\x+\s,\y+\s) -- (\x+1-\s-\s,\y+\s);
}

\foreach \y in {3,9,15} \foreach  \x in {\y+1,\y+3,\y+5,\y+7,\y+9,\y+11,\y+13,\y+15,\y+17}  \foreach \s in {0.1}
{	\filldraw[fill=black] (\x, \y) circle (0.1);
	\draw [->, rounded corners, ultra thick, black]   (\x+\s,\y) -- (\x+\s,\y+1-\s);
	\draw [->, rounded corners, ultra thick, black]	(\x+\s,\y+1-\s) --	(\x+\s,\y+2-\s);
	\draw [->, rounded corners, ultra thick, black]	(\x+\s,\y+2-\s) -- (\x+1-\s,\y+2-\s); 
	\draw [->, rounded corners, ultra thick, black]	(\x+1-\s,\y+2-\s) -- (\x+1-\s,\y+1+\s);
	\draw [->, rounded corners, ultra thick, black]	(\x+1-\s,\y+1+\s) -- (\x+1-\s,\y+\s);
	\draw [->, dashed, ultra thick, black] (\x+\s,\y+\s) -- (\x+1-\s-\s,\y+\s);
}

\foreach \y in {5,11,17} \foreach  \x in {\y+1,\y+3,\y+5,\y+7,\y+9,\y+11,\y+13,\y+15,\y+17}  \foreach \s in {0.1}
{	\filldraw[fill=green] (\x, \y) circle (0.1);
	\draw [->, rounded corners, ultra thick, green]   (\x+\s,\y) -- (\x+\s,\y+1-\s);
	\draw [->, rounded corners, ultra thick, green]	(\x+\s,\y+1-\s) --	(\x+\s,\y+2-\s);
	\draw [->, rounded corners, ultra thick, green]	(\x+\s,\y+2-\s) -- (\x+1-\s,\y+2-\s); 
	\draw [->, rounded corners, ultra thick, green]	(\x+1-\s,\y+2-\s) -- (\x+1-\s,\y+1+\s);
	\draw [->, rounded corners, ultra thick, green]	(\x+1-\s,\y+1+\s) -- (\x+1-\s,\y+\s);
	\draw [->, dashed, ultra thick, green] (\x+\s,\y+\s) -- (\x+1-\s-\s,\y+\s);
}

\foreach \x in {2,8,14} \foreach \y in {\x,\x+2,\x+4,\x+6,\x+8,\x+10,\x+12,\x+14,\x+16} \foreach \s in {0.1}
{	\filldraw[fill=green] (\x, \y) circle (0.1);
	\draw [->, rounded corners, ultra thick, green]  (\x-\s,\y+\s) -- (\x-2+\s,\y+\s) -- (\x-2+\s,\y+1-\s)  -- (\x-\s,\y+1-\s);
	\draw [->, dashed, ultra thick, green] (\x-\s,\y+\s) -- (\x-\s,\y+1-\s-\s);
}

\foreach \x in {4,10,16} \foreach \y in {\x,\x+2,\x+4,\x+6,\x+8,\x+10,\x+12,\x+14,\x+16} \foreach \s in {0.1}
{	\filldraw[fill=red] (\x, \y) circle (0.1);
	\draw [->, rounded corners, ultra thick, red]  (\x-\s,\y+\s) -- (\x-2+\s,\y+\s) -- (\x-2+\s,\y+1-\s)  -- (\x-\s,\y+1-\s);
	\draw [->, dashed, ultra thick, red] (\x-\s,\y+\s) -- (\x-\s,\y+1-\s-\s);
}

\foreach \x in {6,12,18} \foreach \y in {\x,\x+2,\x+4,\x+6,\x+8,\x+10,\x+12,\x+14,\x+16} \foreach \s in {0.1}
{	\filldraw[fill=black] (\x, \y) circle (0.1);
	\draw [->, rounded corners, ultra thick, black]  (\x-\s,\y+\s) -- (\x-2+\s,\y+\s) -- (\x-2+\s,\y+1-\s)  -- (\x-\s,\y+1-\s);
	\draw [->, dashed, ultra thick, black] (\x-\s,\y+\s) -- (\x-\s,\y+1-\s-\s);
}

\foreach \x in {1,7,13} \foreach \y in {\x+2,\x+4,\x+6,\x+8,\x+10,\x+12,\x+14,\x+16} \foreach \s in {0.1}
{	\filldraw[fill=black] (\x, \y) circle (0.1);
	\draw [->, rounded corners, ultra thick, black]  (\x+\s,\y+\s) -- (\x+1,\y+\s); 
	\draw [->, rounded corners, ultra thick, black]  (\x+1-\s,\y+\s) -- (\x+2-\s,\y+\s);
	\draw [->, rounded corners, ultra thick, black]  (\x+2-\s,\y+\s) -- (\x+2-\s,\y+1-\s); 
	\draw [->, rounded corners, ultra thick, black] (\x+2-\s,\y+1-\s) -- (\x+1+\s,\y+1-\s);
	\draw [->, rounded corners, ultra thick, black] 	(\x+1-\s,\y+1-\s) -- (\x+\s,\y+1-\s);
	\draw [->, dashed, ultra thick, black] (\x+\s,\y+\s) -- (\x+\s,\y+1-\s-\s);
}

\foreach \x in {3,9,15} \foreach \y in {\x+2,\x+4,\x+6,\x+8,\x+10,\x+12,\x+14,\x+16} \foreach \s in {0.1}
{	\filldraw[fill=green] (\x, \y) circle (0.1);
	\draw [->, rounded corners, ultra thick, green]  (\x+\s,\y+\s) -- (\x+1,\y+\s); 
	\draw [->, rounded corners, ultra thick, green]  (\x+1-\s,\y+\s) -- (\x+2-\s,\y+\s);
	\draw [->, rounded corners, ultra thick, green]  (\x+2-\s,\y+\s) -- (\x+2-\s,\y+1-\s); 
	\draw [->, rounded corners, ultra thick, green] (\x+2-\s,\y+1-\s) -- (\x+1+\s,\y+1-\s);
	\draw [->, rounded corners, ultra thick, green] 	(\x+1-\s,\y+1-\s) -- (\x+\s,\y+1-\s);
	\draw [->, dashed, ultra thick, green] (\x+\s,\y+\s) -- (\x+\s,\y+1-\s-\s);
}

\foreach \x in {5,11,17} \foreach \y in {\x+2,\x+4,\x+6,\x+8,\x+10,\x+12,\x+14,\x+16} \foreach \s in {0.1}
{	\filldraw[fill=red] (\x, \y) circle (0.1);
	\draw [->, rounded corners, ultra thick, red]  (\x+\s,\y+\s) -- (\x+1,\y+\s); 
	\draw [->, rounded corners, ultra thick, red]  (\x+1-\s,\y+\s) -- (\x+2-\s,\y+\s);
	\draw [->, rounded corners, ultra thick, red]  (\x+2-\s,\y+\s) -- (\x+2-\s,\y+1-\s); 
	\draw [->, rounded corners, ultra thick, red] (\x+2-\s,\y+1-\s) -- (\x+1+\s,\y+1-\s);
	\draw [->, rounded corners, ultra thick, red] 	(\x+1-\s,\y+1-\s) -- (\x+\s,\y+1-\s);
	\draw [->, dashed, ultra thick, red] (\x+\s,\y+\s) -- (\x+\s,\y+1-\s-\s);
}

\foreach \x in {0} \foreach  \y in {\x+1,\x+3,\x+5,\x+7,\x+9,\x+11,\x+13,\x+15,\x+17}  \foreach \s in {0.1}
{	\filldraw[fill=red] (\x, \y) circle (0.1);
	\draw [->, rounded corners, ultra thick, red]  (\x+\s,\y+\s) -- (\x+1-\s,\y+\s) -- (\x+1-\s,\y+1-\s)  -- (\x+\s,\y+1-\s);
	\draw [->, dashed, ultra thick, red] (\x,\y+\s) -- (\x,\y+1-\s);
}


\foreach \x in {7.5,13.5} \foreach  \y in {0.5,6.5}
	{
	\draw[->,ultra thick, myGray, line width=1mm, rounded corners]
			(\x,\y) -- (\x,\y+1);
	}		 

\foreach \x in {7.5,13.5} \foreach  \y in {2.5}
{
	\draw[->,ultra thick, myGreen, line width=1mm, rounded corners]
	(\x,\y) -- (\x,\y+1);
}		 

\foreach \x in {7.5,13.5} \foreach  \y in {4.5}
{
	\draw[->,ultra thick, myMelon, line width=1mm, rounded corners]
	(\x,\y) -- (\x,\y+1);
}

\coordinate [label=below:\Large{$B_0$}] (A) at (13.5, 1);
\coordinate [label=below:\Large{$B_1$}] (A) at (13.5, 3);
\coordinate [label=below:\Large{$B_2$}] (A) at (13.5, 5);
\coordinate [label=below:\Large{$D_0$}] (A) at (7.5, 1);
\coordinate [label=below:\Large{$D_1$}] (A) at (7.5, 3);
\coordinate [label=below:\Large{$D_2$}] (A) at (7.5, 5);

\draw [->,ultra thick, black] (13,1.9) to [out=290,in=70] (8.1,2.0);	
\draw [->,ultra thick, green] (13,2) to [out=295,in=65] (8.0,2.1);	

\draw [->,ultra thick, green] (13,3.9) to [out=290,in=70] (8.1,4);	
\draw [->,ultra thick, red] (13,4) to [out=295,in=65] (8,4.1);	

\draw [->,ultra thick, red] (13,5.9) to [out=290,in=70] (8.1,6);	
\draw [->,ultra thick, black] (13,6) to [out=295,in=65] (8,6.1);	

\draw [->,ultra thick, dashed, black] (13,7) to [out=295,in=65] (8.1,8);	
			
		\end{tikzpicture}
		}
		\vspace{0cm}		
	\end{center}
	
	\caption{Proof of Lemma~\ref{lem:3Tinfinity}: $\varsigma(s_0)=(13,2)$, $\varsigma(t_0)=(8,2)$.  $T_0[a_{s_0},a_{t_0}]$ implies $T_1[a_{s_0},a_{t_0}]$ implies $T_1[a_{s_1},a_{t_1}]$ implies $T_2[a_{s_1},a_{t_1}]$ implies $T_2[a_{s_2},a_{t_2}]$ implies $T_0[a_{s_2},a_{t_2}]$. The black edge from $(13,7)$ to $(8,8)$ yields the desired contradiction.}
\label{fig:zig-zag-Example1}
\end{figure}

%% file: rectangular-pathsExample2.tex
%

\begin{figure}[hbt]
	\begin{center}
		
		\resizebox{12.2cm}{!}{
			\begin{tikzpicture}
			
			\clip (-0.8,-0.8) rectangle (13.3,11.3);
	
	\path [fill=myYellow] (-0.1,0.1) to (14.0,14.1) to (-0.1,14.1) to (-0.1,0.1);

	\draw [help lines] (0,0) grid (18,18);
	\draw [orange, dotted] (0,0) -- (18,18);
	
			\foreach \x in {0,1,2,3,4,5,6,7,8,9,10,11,12,13,14,15,16,17} 
	\coordinate [label=center:\x] (A) at (\x, -0.5);
	\foreach \y in {0,1,2,3,4,5,6,7,8,9,10,11,12,13,14,15,16,17}  
	\coordinate [label=center:\y] (A) at (-0.5,\y);
	
		\foreach \x in {0,1,2,3,4,5,6,7,8,9,10,11,12,13,14,15,16,17}
	\foreach \y in {0,1,2,3,4,5,6,7,8,9,10,11,12,13,14,15,16,17}
	{
		\filldraw[fill=white] (\x, \y) circle (0.1); 
	}

\foreach \x in {1,7,13} \foreach \y in {\x}
\foreach \s in {0.1}
{	\filldraw[fill=black] (\x, \y) circle (0.1);
	\draw [->, rounded corners, ultra thick, black]  (\x+\s,\y) -- (\x+\s,\y-1+\s) -- (\x+1-\s,\y-1+\s)  -- (\x+1-\s,\y+1-\s) -- (\x+\s,\y+1-\s);
	\draw [->, dashed, ultra thick, black] (\x+\s,\y+\s) -- (\x+\s,\y+1-\s-\s);
}

\foreach \x in {3,9,15} \foreach \y in {\x}
\foreach \s in {0.1}
{	\filldraw[fill=green] (\x, \y) circle (0.1);
	\draw [->, rounded corners, ultra thick, green]  (\x+\s,\y) -- (\x+\s,\y-1+\s) -- (\x+1-\s,\y-1+\s)  -- (\x+1-\s,\y+1-\s) -- (\x+\s,\y+1-\s);
	\draw [->, dashed, ultra thick, green] (\x+\s,\y+\s) -- (\x+\s,\y+1-\s-\s);
}

\foreach \x in {5,11} \foreach \y in {\x}
\foreach \s in {0.1}
{	\filldraw[fill=red] (\x, \y) circle (0.1);
	\draw [->, rounded corners, ultra thick, red]  (\x+\s,\y) -- (\x+\s,\y-1+\s) -- (\x+1-\s,\y-1+\s)  -- (\x+1-\s,\y+1-\s) -- (\x+\s,\y+1-\s);
	\draw [->, dashed, ultra thick, red] (\x+\s,\y+\s) -- (\x+\s,\y+1-\s-\s);
}

\foreach \y in {2,8,14} \foreach  \x in {\y+1,\y+3,\y+5,\y+7,\y+9,\y+11,\y+13,\y+15,\y+17}  \foreach \s in {0.1}
{	\filldraw[fill=black] (\x, \y) circle (0.1);
	\draw [->, rounded corners, ultra thick, black]  (\x+\s,\y) -- (\x+\s,\y-2+\s) -- (\x+1-\s,\y-2+\s)  -- (\x+1-\s,\y-\s);
	\draw [->, dashed, ultra thick, black] (\x+\s,\y-\s) -- (\x+1-\s-\s,\y-\s);
}

\foreach \y in {0} \foreach  \x in {\y,\y+2,\y+4,\y+6,\y+8,\y+10,\y+12,\y+14,\y+16}  \foreach \s in {0.1}
{	\filldraw[fill=green] (\x, \y) circle (0.1);
	\draw [->, rounded corners, ultra thick, green]  (\x+\s,\y+\s) -- (\x+\s,\y+1-\s) -- (\x+1-\s,\y+1-\s)  -- (\x+1-\s,\y+\s);
	\draw [->, dashed, ultra thick, green] (\x+\s,\y+\s) -- (\x+1-\s-\s,\y+\s);
}

\foreach \y in {4,10,16} \foreach  \x in {\y+1,\y+3,\y+5,\y+7,\y+9,\y+11,\y+13,\y+15,\y+17}  \foreach \s in {0.1}
{	\filldraw[fill=green] (\x, \y) circle (0.1);
	\draw [->, rounded corners, ultra thick, green]  (\x+\s,\y) -- (\x+\s,\y-2+\s) -- (\x+1-\s,\y-2+\s)  -- (\x+1-\s,\y-\s);
	\draw [->, dashed, ultra thick, green] (\x+\s,\y-\s) -- (\x+1-\s-\s,\y-\s);
}

\foreach \y in {6,12,18} \foreach  \x in {\y+1,\y+3,\y+5,\y+7,\y+9,\y+11,\y+13,\y+15,\y+17}  \foreach \s in {0.1}
{	\filldraw[fill=red] (\x, \y) circle (0.1);
	\draw [->, rounded corners, ultra thick, red]  (\x+\s,\y) -- (\x+\s,\y-2+\s) -- (\x+1-\s,\y-2+\s)  -- (\x+1-\s,\y-\s);
	\draw [->, dashed, ultra thick, red] (\x+\s,\y-\s) -- (\x+1-\s-\s,\y-\s);
}

\foreach \y in {1,7,13} \foreach  \x in {\y+1,\y+3,\y+5,\y+7,\y+9,\y+11,\y+13,\y+15,\y+17}  \foreach \s in {0.1}
{	\filldraw[fill=red] (\x, \y) circle (0.1);
	\draw [->, rounded corners, ultra thick, red]   (\x+\s,\y) -- (\x+\s,\y+1-\s);
	\draw [->, rounded corners, ultra thick, red]	(\x+\s,\y+1-\s) --	(\x+\s,\y+2-\s);
	\draw [->, rounded corners, ultra thick, red]	(\x+\s,\y+2-\s) -- (\x+1-\s,\y+2-\s); 
	\draw [->, rounded corners, ultra thick, red]	(\x+1-\s,\y+2-\s) -- (\x+1-\s,\y+1+\s);
	\draw [->, rounded corners, ultra thick, red]	(\x+1-\s,\y+1+\s) -- (\x+1-\s,\y+\s);
	\draw [->, dashed, ultra thick, red] (\x+\s,\y+\s) -- (\x+1-\s-\s,\y+\s);
}

\foreach \y in {3,9,15} \foreach  \x in {\y+1,\y+3,\y+5,\y+7,\y+9,\y+11,\y+13,\y+15,\y+17}  \foreach \s in {0.1}
{	\filldraw[fill=black] (\x, \y) circle (0.1);
	\draw [->, rounded corners, ultra thick, black]   (\x+\s,\y) -- (\x+\s,\y+1-\s);
	\draw [->, rounded corners, ultra thick, black]	(\x+\s,\y+1-\s) --	(\x+\s,\y+2-\s);
	\draw [->, rounded corners, ultra thick, black]	(\x+\s,\y+2-\s) -- (\x+1-\s,\y+2-\s); 
	\draw [->, rounded corners, ultra thick, black]	(\x+1-\s,\y+2-\s) -- (\x+1-\s,\y+1+\s);
	\draw [->, rounded corners, ultra thick, black]	(\x+1-\s,\y+1+\s) -- (\x+1-\s,\y+\s);
	\draw [->, dashed, ultra thick, black] (\x+\s,\y+\s) -- (\x+1-\s-\s,\y+\s);
}

\foreach \y in {5,11,17} \foreach  \x in {\y+1,\y+3,\y+5,\y+7,\y+9,\y+11,\y+13,\y+15,\y+17}  \foreach \s in {0.1}
{	\filldraw[fill=green] (\x, \y) circle (0.1);
	\draw [->, rounded corners, ultra thick, green]   (\x+\s,\y) -- (\x+\s,\y+1-\s);
	\draw [->, rounded corners, ultra thick, green]	(\x+\s,\y+1-\s) --	(\x+\s,\y+2-\s);
	\draw [->, rounded corners, ultra thick, green]	(\x+\s,\y+2-\s) -- (\x+1-\s,\y+2-\s); 
	\draw [->, rounded corners, ultra thick, green]	(\x+1-\s,\y+2-\s) -- (\x+1-\s,\y+1+\s);
	\draw [->, rounded corners, ultra thick, green]	(\x+1-\s,\y+1+\s) -- (\x+1-\s,\y+\s);
	\draw [->, dashed, ultra thick, green] (\x+\s,\y+\s) -- (\x+1-\s-\s,\y+\s);
}

\foreach \x in {2,8,14} \foreach \y in {\x,\x+2,\x+4,\x+6,\x+8,\x+10,\x+12,\x+14,\x+16} \foreach \s in {0.1}
{	\filldraw[fill=green] (\x, \y) circle (0.1);
	\draw [->, rounded corners, ultra thick, green]  (\x-\s,\y+\s) -- (\x-2+\s,\y+\s) -- (\x-2+\s,\y+1-\s)  -- (\x-\s,\y+1-\s);
	\draw [->, dashed, ultra thick, green] (\x-\s,\y+\s) -- (\x-\s,\y+1-\s-\s);
}

\foreach \x in {4,10,16} \foreach \y in {\x,\x+2,\x+4,\x+6,\x+8,\x+10,\x+12,\x+14,\x+16} \foreach \s in {0.1}
{	\filldraw[fill=red] (\x, \y) circle (0.1);
	\draw [->, rounded corners, ultra thick, red]  (\x-\s,\y+\s) -- (\x-2+\s,\y+\s) -- (\x-2+\s,\y+1-\s)  -- (\x-\s,\y+1-\s);
	\draw [->, dashed, ultra thick, red] (\x-\s,\y+\s) -- (\x-\s,\y+1-\s-\s);
}

\foreach \x in {6,12,18} \foreach \y in {\x,\x+2,\x+4,\x+6,\x+8,\x+10,\x+12,\x+14,\x+16} \foreach \s in {0.1}
{	\filldraw[fill=black] (\x, \y) circle (0.1);
	\draw [->, rounded corners, ultra thick, black]  (\x-\s,\y+\s) -- (\x-2+\s,\y+\s) -- (\x-2+\s,\y+1-\s)  -- (\x-\s,\y+1-\s);
	\draw [->, dashed, ultra thick, black] (\x-\s,\y+\s) -- (\x-\s,\y+1-\s-\s);
}

\foreach \x in {1,7,13} \foreach \y in {\x+2,\x+4,\x+6,\x+8,\x+10,\x+12,\x+14,\x+16} \foreach \s in {0.1}
{	\filldraw[fill=black] (\x, \y) circle (0.1);
	\draw [->, rounded corners, ultra thick, black]  (\x+\s,\y+\s) -- (\x+1,\y+\s); 
	\draw [->, rounded corners, ultra thick, black]  (\x+1-\s,\y+\s) -- (\x+2-\s,\y+\s);
	\draw [->, rounded corners, ultra thick, black]  (\x+2-\s,\y+\s) -- (\x+2-\s,\y+1-\s); 
	\draw [->, rounded corners, ultra thick, black] (\x+2-\s,\y+1-\s) -- (\x+1+\s,\y+1-\s);
	\draw [->, rounded corners, ultra thick, black] 	(\x+1-\s,\y+1-\s) -- (\x+\s,\y+1-\s);
	\draw [->, dashed, ultra thick, black] (\x+\s,\y+\s) -- (\x+\s,\y+1-\s-\s);
}

\foreach \x in {3,9,15} \foreach \y in {\x+2,\x+4,\x+6,\x+8,\x+10,\x+12,\x+14,\x+16} \foreach \s in {0.1}
{	\filldraw[fill=green] (\x, \y) circle (0.1);
	\draw [->, rounded corners, ultra thick, green]  (\x+\s,\y+\s) -- (\x+1,\y+\s); 
	\draw [->, rounded corners, ultra thick, green]  (\x+1-\s,\y+\s) -- (\x+2-\s,\y+\s);
	\draw [->, rounded corners, ultra thick, green]  (\x+2-\s,\y+\s) -- (\x+2-\s,\y+1-\s); 
	\draw [->, rounded corners, ultra thick, green] (\x+2-\s,\y+1-\s) -- (\x+1+\s,\y+1-\s);
	\draw [->, rounded corners, ultra thick, green] 	(\x+1-\s,\y+1-\s) -- (\x+\s,\y+1-\s);
	\draw [->, dashed, ultra thick, green] (\x+\s,\y+\s) -- (\x+\s,\y+1-\s-\s);
}

\foreach \x in {5,11,17} \foreach \y in {\x+2,\x+4,\x+6,\x+8,\x+10,\x+12,\x+14,\x+16} \foreach \s in {0.1}
{	\filldraw[fill=red] (\x, \y) circle (0.1);
	\draw [->, rounded corners, ultra thick, red]  (\x+\s,\y+\s) -- (\x+1,\y+\s); 
	\draw [->, rounded corners, ultra thick, red]  (\x+1-\s,\y+\s) -- (\x+2-\s,\y+\s);
	\draw [->, rounded corners, ultra thick, red]  (\x+2-\s,\y+\s) -- (\x+2-\s,\y+1-\s); 
	\draw [->, rounded corners, ultra thick, red] (\x+2-\s,\y+1-\s) -- (\x+1+\s,\y+1-\s);
	\draw [->, rounded corners, ultra thick, red] 	(\x+1-\s,\y+1-\s) -- (\x+\s,\y+1-\s);
	\draw [->, dashed, ultra thick, red] (\x+\s,\y+\s) -- (\x+\s,\y+1-\s-\s);
}

\foreach \x in {0} \foreach  \y in {\x+1,\x+3,\x+5,\x+7,\x+9,\x+11,\x+13,\x+15,\x+17}  \foreach \s in {0.1}
{	\filldraw[fill=red] (\x, \y) circle (0.1);
	\draw [->, rounded corners, ultra thick, red]  (\x+\s,\y+\s) -- (\x+1-\s,\y+\s) -- (\x+1-\s,\y+1-\s)  -- (\x+\s,\y+1-\s);
	\draw [->, dashed, ultra thick, red] (\x,\y+\s) -- (\x,\y+1-\s);
}


\foreach \x in {4.5,10.5} \foreach  \y in {\x}
{
	\draw[->,ultra thick, myGray, line width=1mm, rounded corners]
	(\x,\y) -- (\x,\y-1);
}		 

\foreach \x in {4.5,10.5} \foreach  \y in {\x-2}
{
	\draw[->,ultra thick, myMelon, line width=1mm, rounded corners]
	(\x,\y) -- (\x,\y-1);
}		 

\foreach \x in {10.5} \foreach  \y in {\x-4}
{
	\draw[->,ultra thick, myGreen, line width=1mm, rounded corners]
	(\x,\y) -- (\x,\y-1);
}		 

\foreach \x in {4.5} \foreach  \y in {0.7}
{
	\draw[->,ultra thick, myGreen, line width=1mm, rounded corners]
	(\x,\y) -- (\x,\y-0.6);
}

\coordinate [label=above:\Large{$B_0$}] (A) at (10.5, 10);
\coordinate [label=above:\Large{$B_1$}] (A) at (10.5, 8);
\coordinate [label=above:\Large{$D_0$}] (A) at (4.5, 4);
\coordinate [label=above:\Large{$D_1$}] (A) at (4.5, 2);

\draw [->,ultra thick, black] (10,9) to 
(5,3);	
\draw [->,ultra thick, red] (10,8.9) to 
(5.1,3.0);	

\draw [->,ultra thick, red] (10,7) to 
(5,1);	
\draw [->,ultra thick, green] (10,6.9) to 
(5.1,1);	


\draw [->,ultra thick, dashed, green] (10,6) to 
(5,0);	

		\end{tikzpicture}
		}
		\vspace{0cm}		
	\end{center}
	
	\caption{Proof of Lemma~\ref{lem:3Tinfinity}: $\varsigma(s_0)=(10,9)$, $\varsigma(t_0)=(5,3)$.  $T_0[a_{s_0},a_{t_0}]$ implies $T_2[a_{s_0},a_{t_0}]$ implies $T_2[a_{s_1},a_{t_1}]$ implies $T_1[a_{s_1},a_{t_1}]$. $\varsigma(s_1-1)=(10,6)$, $\varsigma(t_1+1)=(5,0)$ and the green edge from $(10,6)$ to $(5,0)$ yields the desired contradiction with~\eqref{eq:bou:dControl2}.}
	\label{fig:zig-zag-Example2}
\end{figure}
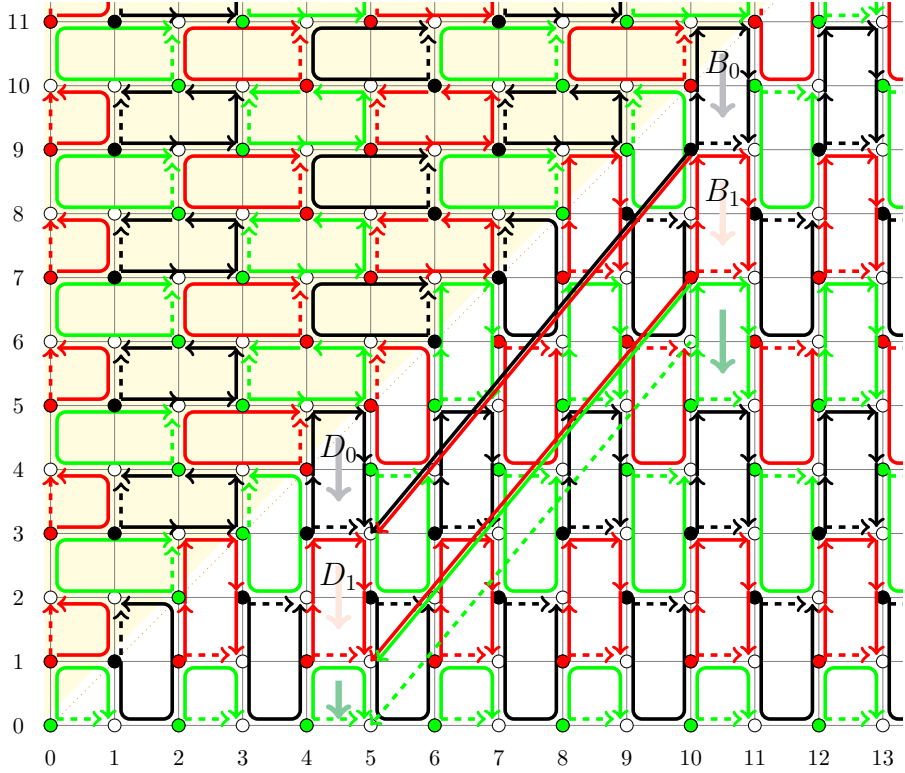

%% file: conclusions.tex
\section{Conclusions}
In this paper, we considered the logics
$\FlTransVars{k}{m}$  and $\FlEqTransVars{k}{m}$,  the $m$-variable fluted fragment in the presence of (equality and) $k$ transitive relations. We showed that the satisfiability problem for $\FlEqTransVars{1}{m}$ is in $m$-\NExpTime, and indeed
that the corresponding finite
satisfiability problem is in ($m+1$)-\NExpTime. (It seems probable that this latter bound, at least, can be improved.) Together with known
lower bounds on the $m$-variable fluted fragment, it follows that the satisfaibility and finite satisfiability problems for 
$\FlEqTrans{1}$, the fluted fragment with equality and a single transitive relation, are both \Tower-complete. (This extends the result 
of~\cite{P-HST-FLrev}, which establishes the same complexity for the fluted fragment without equality or any transitive relations.)
We also showed, however, that decidability is easily lost when additional transitive relations are added: even the two-variable fluted fragments 
$\FlEqTransVars{2}{2}$ (two transitive relations plus equality) and $\FlEqTransVars{3}{2}$ 
(three transitive relations, but without equality) have undecidable satisfiability and finite satisfiability problems. 

It is open whether the satisfiability or finite satisfiability problems for $\FlTrans{2}$  (\textit{two} transitive relations, but without equality) are decidable.
We point out that Lemma~\ref{lma:FLotransmEqToFLotransmEq} in Section~\ref{sec:onetrans} could be generalized to normal-form formulas of $\FlTransVars{2}{m+1}$ (defined in the natural way). Hence,  the (finite) satisfiability problem for  $\FlTransVars{2}{m}$ ($m >2$) is decidable if and only if the corresponding problem $\FlTransVars{2}{2}$  is. Unfortunately neither the method of Sec.~\ref{sec:onetrans} (to show
decidability) nor that of Sec.~\ref{sec:undecidable} (to show
undecidability) appears to apply here. The barrier in the former case is that pairs of elements can be related by both of the transitive
relations, $T_1$ and $T_2$, via \textit{distinct} $T_1$- and $T_2$-chains, so that simple certificates of the kind employed for $\FLotranstEqMinus$ do not guarantee the existence of models. The barrier in the latter case is that the grid construction has to build models featuring transitive paths of {\em bounded} length, and this seems not to be achievable with just two transitive relations. 

\bigskip
{\bf Acknowledgements.} This work is supported by the Polish National Science Centre grant 2018/31/B/ST6/03662.